%% file: cap.tex
\documentclass[journal]{IEEEtran}
\usepackage{amssymb,stmaryrd,amsmath,amsfonts,rotating}
\usepackage[noadjust]{cite} \usepackage{color} 
\usepackage{epic}
\usepackage{color}
\newcommand{\ww}{\underline{w}} 
\newcommand{\x}{\underline{x}}
\newcommand{\y}{\underline{y}} \newcommand{\z}{\underline{z}}
\newcommand{\vc}{\underline{v}}
 
\newcommand{\xunstable}{x_{\text{u}}(\epsilon)}
\newcommand{\xunstab}{x_{\text{u}}}
 
\newcommand{\xstable}{x_{\text{s}}(\epsilon)}
\newcommand{\xstab}{x_{\text{s}}}
\newcommand{\xavg}{\bar{x}} 
 
\RequirePackage{bbm} \input{./definition} \allowdisplaybreaks
\begin{document} 
\title{
Threshold Saturation via Spatial Coupling:
Why Convolutional LDPC Ensembles Perform so well over the BEC}
\author{\authorblockN{Shrinivas
Kudekar\authorrefmark{1}, Tom Richardson\authorrefmark{2} and R{\"u}diger
Urbanke\authorrefmark{1} \\ } \authorblockA{\authorrefmark{1}School of
Computer and Communication Sciences\\ EPFL, Lausanne, Switzerland\\
Email: \{shrinivas.kudekar, ruediger.urbanke\}@epfl.ch}\\
\authorblockA{\authorrefmark{2} Qualcomm, USA\\ Email: tjr@qualcomm.com} }

\maketitle
\begin{abstract}
Convolutional LDPC ensembles, introduced by Felstr{\"{o}}m and Zigangirov,
have excellent thresholds and these thresholds are rapidly increasing 
functions of the average degree. Several variations on the basic theme
have been proposed to date, all of which share the good performance
characteristics of convolutional LDPC ensembles.

We describe the fundamental mechanism which explains why
``convolutional-like'' or ``spatially coupled'' codes perform so well.
In essence, the spatial coupling of the individual code structure has the
effect of increasing the belief-propagation threshold of the new ensemble
to its maximum possible value, namely the maximum-a-posteriori
threshold of the underlying ensemble.  For this reason we call this
phenomenon ``threshold saturation''.

This gives an entirely new way of approaching capacity.  One
significant advantage of such a construction is that one can create
capacity-approaching ensembles with an error correcting radius which is
increasing in the blocklength. Our proof makes use of the area theorem
of the belief-propagation EXIT curve and the connection between the
maximum-a-posteriori and belief-propagation threshold recently pointed
out by M{\'e}asson, Montanari, Richardson, and Urbanke.

Although we prove the connection between the maximum-a-posteriori and
the belief-propagation threshold only for a very specific ensemble and
only for the binary erasure channel, empirically a threshold saturation
phenomenon occurs for a wide class of ensembles and channels. More
generally, we conjecture that for a large range of graphical systems a
similar saturation of the ``dynamical'' threshold occurs once individual
components are coupled sufficiently strongly.  This might give rise
to improved algorithms as well as to new techniques for analysis.
\end{abstract}


\section{Introduction}
We consider the design of capacity-approaching codes based on the
connection between the belief-propagation (BP) and maximum-a-posteriori
(MAP) threshold of sparse graph codes. Recall that the BP threshold is
the threshold of the ``locally optimum'' BP message-passing algorithm. As
such it has low complexity. The MAP threshold, on the other hand, is the
threshold of the ``globally optimum'' decoder. No decoder can do better,
but the complexity of the MAP decoder is in general high. 
The threshold itself is the unique channel parameter so that for
channels with lower (better) parameter decoding succeeds with high probability
(for large instances) whereas for channels with higher (worse) parameters
decoding fails with high probability.
Surprisingly, for sparse graph codes there is a connection between
these two thresholds, see \cite{MMRU04,MMU08}.\footnote{There are some trivial
instances in which the two thresholds coincide.  This is e.g. the case
for so-called ``cycle ensembles'' or, more generally, for irregular
LDPC ensembles that have a large fraction of degree-two variable nodes.
In these cases the reason for this agreement is that for both decoders
the performance is dominated by small structures in the graph.  But for
general ensembles these two thresholds are distinct and, indeed, they
can differ significantly.}

We discuss a fundamental mechanism which ensures that these two
thresholds coincide (or at least are very close).  We call this phenomenon
``threshold saturation via spatial coupling.'' A prime example where this
mechanism is at work are {\em convolutional low-density parity-check}
(LDPC) ensembles.

It was Tanner who introduced the method of ``unwrapping'' a cyclic block
code into a convolutional structure \cite{Tan81b,Tan87}.  The first {\em
low-density} convolutional ensembles were introduced by Felstr\"{o}m and
Zigangirov \cite{FeZ99}.  Convolutional LDPC ensembles are constructed
by {\em coupling} several standard $(\dl, \dr)$-regular LDPC ensembles
together in a chain. Perhaps surprisingly, due to the coupling, and
assuming that the chain is finite and properly terminated, the threshold
of the resulting ensemble is considerably improved. Indeed, if we start
with a $(3, 6)$-regular ensemble, then on the binary erasure channel (BEC)
the threshold is improved from $\epsilon^{\BPsmall}(\dl=3, \dr=6) \approx
0.4294$ to roughly $0.4881$ (the capacity for this case is $\frac12$). The
latter number is the MAP threshold $\epsilon^{\MAPsmall}(\dl, \dr)$ of the
underlying $(3, 6)$-regular ensemble.  This opens up an entirely new way
of constructing capacity-approaching ensembles.  It is a folk theorem that
for standard constructions improvements in the BP threshold go hand in
hand with increases in the error floor.  More precisely, a large fraction
of degree-two variable nodes is typically needed in order to get large
thresholds under BP decoding. Unfortunately, the higher the fraction of
degree-two variable nodes, the more low-weight codewords (small cycles,
small stopping sets, ...) appear.  Under MAP decoding on the other
hand these two quantities are positively correlated. To be concrete,
if we consider the sequence of $(\dl, 2\dl)$-regular ensembles of rate
one-half, by increasing $\dl$ we increase both the MAP threshold as well
as the typical minimum distance.  It is therefore possible to construct
ensembles that have large MAP thresholds {\em and} low error floors.

The potential of convolutional LDPC codes has long been recognized.
Our contribution lies therefore not in the introduction of a new coding
scheme, but in clarifying the basic mechanism that make convolutional-like
ensembles perform so well.

There is a considerable literature on convolutional-like LDPC ensembles.
Variations on the constructions as well as some analysis can be found in
Engdahl and Zigangirov \cite{EnZ99}, Engdahl, Lentmaier, and Zigangirov
\cite{ELZ99}, Lentmaier, Truhachev, and Zigangirov \cite{LTZ01}, as well
as Tanner, D. Sridhara, A. Sridharan, Fuja, and Costello \cite{TSSFC04}.
In \cite{SLCZ04, LSZC10}, Sridharan, Lentmaier, Costello and Zigangirov
consider density evolution (DE) for convolutional LDPC ensembles and
determine thresholds for the BEC. The equivalent observations for general
channels were reported by Lentmaier, Sridharan, Zigangirov and Costello
in \cite{LSZC05, LSZC10}.  The preceding two sets of works are perhaps
the most pertinent to our setup. By considering the resulting thresholds
and comparing them to the thresholds of the underlying ensembles under
MAP decoding (see e.g. \cite{RiU08}) it becomes quickly apparent that
an interesting physical effect must be at work.  Indeed, in a recent
paper \cite{LeF10}, Lentmaier and Fettweis followed this route and
independently formulated the equality of the BP threshold of convolutional
LDPC ensembles and the MAP threshold of the underlying ensemble as a
conjecture. They attribute this numerical observation to G. Liva.

A representation of convolutional LDPC ensembles in terms of
a protograph was introduced by Mitchell, Pusane, Zigangirov and
Costello \cite{MPZC08}. The corresponding representation for terminated
convolutional LDPC ensembles was introduced by Lentmaier, Fettweis,
Zigangirov and Costello \cite{LFZC09}. A pseudo-codeword analysis of
convolutional LDPC codes was performed by Smarandache, Pusane, Vontobel,
and Costello in \cite{SPVC06,SPVC09}. In \cite{PISWC10}, Papaleo, Iyengar,
Siegel, Wolf, and Corazza consider windowed decoding of convolutional
LDPC codes on the BEC to study the trade-off between the decoding latency
and the code performance.

In the sequel we will assume that the reader is familiar with basic
notions of sparse graph codes and message-passing decoding, and in
particular with the asymptotic analysis of LDPC ensembles for
transmission over the binary erasure channel as it was accomplished
in \cite{LMSSS97}. We summarized the most important facts which are
needed for our proof in Section~\ref{sec:standard}, but this summary
is not meant to be a gentle introduction to the topic.  Our notation
follows for the most part the one in \cite{RiU08}.

\section{Convolutional-Like LDPC Ensembles}
The principle that underlies the good performance of convolutional-like
LDPC ensembles is very broad and there are many degrees of freedom
in constructing such ensembles. In the sequel we introduce two
basic variants. The $(\dl, \dr, L)$-ensemble is very close to the
ensemble discussed in \cite{LFZC09}. Experimentally it has a very good
performance. We conjecture that it is capable of achieving capacity.

We also introduce the ensemble $(\dl, \dr, L, w)$.  Experimentally it
shows a worse trade-off between rate, threshold, and blocklength.  But it
is easier to analyze and we will show that it is capacity achieving.
One can think of $w$ as a ``smoothing parameter'' and we investigate
the behavior of this ensemble when $w$ tends to infinity.

\subsection{The $(\dl, \dr, L)$ Ensemble}
To start, consider a protograph of a standard $(3, 6)$-regular ensemble 
(see \cite{T03, DDJ06} for the definition of protographs).
It is shown in Figure~\ref{fig:36protograph}. There are two variable
nodes and there is one check node.  Let $M$ denote the number of variable
nodes at each position.  For our example, $M=100$ means that we have $50$
copies of the protograph so that we have $100$ variable nodes at each
position. For all future discussions we will consider the regime
where $M$ tends to infinity.  
\begin{figure}[htp] \begin{centering}
\input{ps/36protograph} \caption{Protograph of a standard $(3, 6)$-regular
ensemble.} \label{fig:36protograph} \end{centering} 
\end{figure}
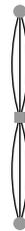

Next, consider a collection of $(2 L+1)$ such protographs as shown in
Figure~\ref{fig:36protographchain}.
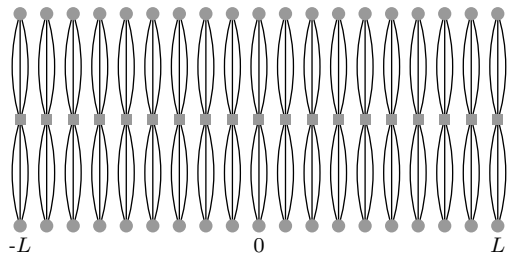
\begin{figure}[htp]
\begin{centering}
\input{ps/36protographchain}
\caption{
A chain of $(2L+1)$ protographs of the standard $(3, 6)$-regular ensembles
for $L=9$.  These protographs do not interact.}
\label{fig:36protographchain}
\end{centering}
\end{figure}
These protographs are non-interacting and so each component behaves
just like a standard $(3, 6)$-regular component. In particular, the
belief-propagation (BP) threshold of each protograph is just the
standard threshold, call it $\epsilon^{\BPsmall}(\dl=3, \dr=6)$
(see Lemma~\ref{lem:standardthresholds} for an analytic characterization
of this threshold).  Slightly more generally: start with an $(\dl,
\dr=k \dl)$-regular ensemble where $\dl$ is odd so that $\dlh=(\dl-1)/2
\in \naturals$.

An interesting phenomenon occurs if we couple these
components.  To achieve this coupling, connect each protograph to $\dlh$\footnote{If we think of this as a convolutional code, then $2\dlh$ is the {\em syndrome former memory} of the code.}
protographs ``to the left'' and to $\dlh$ protographs ``to the right.''
This is  shown in Figure~\ref{fig:chain} for the two cases $(\dl=3, \dr=6)$ and $(\dl=7, \dr=14)$.
In this figure, $\dlh$ extra check nodes are added on each
side to connect the ``overhanging'' edges at the boundary.
\begin{figure}[htp]
\begin{centering}
\input{ps/chain}
\caption{Two coupled chains of protographs with $L=9$ and $(\dl=3, \dr=6)$ (top) 
and $L=7$ and $(\dl=7, \dr=14)$ (bottom), 
respectively. \label{fig:chain}}
\end{centering}
\end{figure}
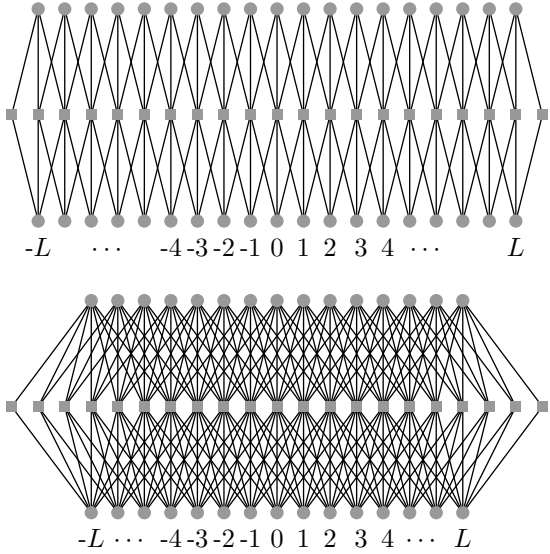

There are two main effects resulting from this coupling:
\begin{itemize}
\item[(i)] {\em Rate Reduction:} 
Recall that the design rate of the underlying standard $(\dl, \dr=k
\dl)$-regular ensemble is $1-\frac{\dl}{\dr}=\frac{k-1}{k}$.  Let us
determine the design rate of the corresponding $(\dl, \dr=k \dl, L)$
ensemble. By design rate we mean here the rate that we get if we assume
that every involved check node imposes a linearly independent constraint.

The variable nodes are indexed from $-L$ to $L$ so that in total
there are $(2 L+1)M$ variable nodes.  The check nodes are indexed from
$-(L+\dlh)$ to $(L+\dlh)$, so that in total there are $(2 (L+\dlh)+1)M/k$
check nodes.  We see that,
due to boundary effects, the design rate is reduced to
\begin{align*}
R(\dl, \dr=k\dl, L) 
& = \frac{(2L+1)-(2(L+\dlh)+1)/k}{2L+1} \\
& = \frac{k-1}{k} - \frac{2 \dlh}{k(2L+1)},
\end{align*} 
where the first term on the right represents the design rate of the underlying standard $(\dl, \dr=k
\dl)$-regular ensemble and the second term represents the rate loss.
As we see, this rate reduction effect vanishes at a speed $1/L$. 
\item[(ii)] {\em Threshold Increase:} The threshold changes dramatically
from $\epsilon^{\BPsmall}(\dl, \dr)$
to something close to $\epsilon^{\MAPsmall}(\dl, \dr)$ (the MAP 
threshold of the underlying standard $(\dl, \dr)$-regular ensemble;
see Lemma~\ref{lem:standardthresholds}).
This phenomenon (which we call ``threshold saturation'') is much less intuitive and it is the
aim of this paper to explain why this happens.  
\end{itemize}

So far we have considered $(\dl, \dr=k \dl)$-regular
ensembles.  Let us now give a general definition of the $(\dl, \dr,
L)$-ensemble which works for all parameters $(\dl, \dr)$ so that $\dl$
is odd.  Rather than starting from a protograph, place variable nodes
at positions $[-L, L]$. At each position there are $M$ such variable
nodes. Place $\frac{\dl}{\dr} M$ check nodes at each position $[-L-\dlh,
L+\dlh]$.  Connect exactly one of the $\dl$ edges of each variable node
at position $i$ to a check node at position $i-\dlh, \dots, i+\dlh$.

Note that at each position $i \in [-L+\dlh, L-\dlh]$, there are exactly
$M \frac{\dl}{\dr} \dr = M \dl$ check node sockets\footnote{
{\em Sockets} are connection points where edges can be attached
to a node. E.g., if a node has degree $3$ then we imagine that it
has $3$ sockets. This terminology arises from the so-called {\em configuration}
model of LDPC ensembles. In this model we imagine that we label all
check-node sockets and all variable-node sockets with the set of integers
from one to the cardinality of the sockets. To construct then a particular
element of the ensemble we pick a permutation on this set uniformly at random 
from the set of all permutations and connect variable-node sockets to check-node
sockets according to this permutation.}.  Exactly $M$ of those come
from variable nodes at each position $i-\dlh, \dots, i+\dlh$. For
check nodes at the boundary the number of sockets is decreased
linearly according to their position.  The probability distribution
of the ensemble is defined by choosing a random permutation on the
set of all edges for each check node position.

The next lemma, whose proof can be found in Appendix~\ref{app:lrLweight},
asserts that the minimum stopping set distance of most codes in this ensemble is at least a fixed
fraction of $M$. 
With respect to the technique used in the proof we follow the lead of 
\cite{MPZC08,SPVC09} and \cite{DDJ06,SPVC06} which consider distance and pseudo-distance analysis of
convolutional LDPC ensembles, respectively.
\blemma[Stopping Set Distance of $(\dl, \dr, L)$-Ensemble]\label{lem:lrLweight}
Consider the $(\dl, \dr, L)$-ensemble with $\dl=2 \dlh +1$, $\dlh \geq
1$, and $\dr \geq \dl$. Define
\begin{align*}
p(x) & = \sum_{i \neq 1} \binom{\dr}{i} x^i, \;\; a(x)  = (\sum_{i \neq 1} \binom{\dr}{i} i x^i)/(\sum_{i \neq 1} \binom{\dr}{i} x^i), \\
b(x) & =\!-\!(\dl\!-\!1) h_2 (a(x)\!/\!\dr) \!+\! \frac{\dl}{\dr}\! \log_2(p(x))\!-\! a(x) \frac{\dl}{\dr}\! \log_2(x), \\
\omega(x) & = a(x)/\dr, \;\; h_2(x)= -x\log_2(x)-(1-x) \log_2(1-x).
\end{align*}
Let $\hat{x}$ denote the unique strictly positive solution of the equation $b(x)=0$ and
let $\hat{\omega}(\dl, \dr)=\omega(\hat{x})$. Then, for any $\delta>0$,
\begin{align*}
\lim_{M \rightarrow \infty} \prob\{d_{\text{ss}}({\mathcal C})/M < 
(1-\delta) \dl \hat{\omega}(\dl, \dr) \}  = 0,
\end{align*}
where $d_{\text{ss}}({\mathcal C})$ denotes the minimum stopping set
distance of the code ${\mathcal C}$.
\elemma
{\em Discussion:} 
The quantity $\hat{\omega}(\dl, \dr)$ is the relative weight (normalized
to the blocklength) at which the exponent of the expected stopping set
distribution of the underlying standard $(\dl, \dr)$-regular ensemble
becomes positive. It is perhaps not too surprising that the same quantity
also appears in our context.  The lemma asserts that the minimum stopping
set distance grows linearly in $M$. But the stated bound 
does {\em not} scale with $L$.  We leave it as an interesting open problem to determine 
whether this is due to the looseness of our bound or whether our bound indeed
reflects the correct behavior.

\begin{example}[$(\dl=3, \dr=6, L)$]
An explicit calculation shows that $\hat{x} \approx 0.058$ and
$3 \hat{\omega}(3, 6) \approx  0.056$. Let $n=M (2 L+1)$ be the
blocklength.  If we assume that $2L+1=M^\alpha$, $\alpha \in (0, 1)$,
then $M=n^{\frac{1}{1+\alpha}}$. Lemma~\ref{lem:lrLweight} asserts that
the minimum stopping set distance grows in the blocklength at least as
$0.056 n^{\frac{1}{1+\alpha}}$.  \end{example}

\subsection{The $(\dl, \dr, L, w)$ Ensemble}
In order to simplify the analysis we modify the ensemble $(\dl, \dr, L)$
by adding a randomization of the edge connections.  For the remainder
of this paper we always assume that $\dr \geq \dl$, so that the ensemble
has a non-trivial design rate.

We assume that the variable nodes are at positions $[-L, L]$, $L \in
\naturals$. At each position there are $M$ variable nodes, $M \in
\naturals$. Conceptually we think of the check nodes to be located at
all integer positions from $[- \infty, \infty]$.  Only some of these
positions actually interact with the variable nodes.  At each position
there are $\frac{\dl}{\dr} M$ check nodes. It remains to describe how the
connections are chosen.

Rather than assuming that a variable at position $i$ has exactly one
connection to a check node at position $[i-\dlh, \dots, i+\dlh]$,
we assume that each of the $\dl$ connections of a variable node at
position $i$ is uniformly and independently chosen from the range $[i,
\dots, i+w-1]$, where $w$ is a ``smoothing'' parameter. In the same way,
we assume that each of the $\dr$ connections of a check node at position
$i$ is independently chosen from the range $[i-w+1, \dots, i]$.  We no
longer require that $\dl$ is odd.

More precisely, the ensemble is defined as follows. Consider a variable
node at position $i$. The variable node has $\dl$
outgoing edges.  A {\em type} $t$ is a $w$-tuple of non-negative integers,
$t=(t_0, t_1, \dots, t_{w-1})$, so that $\sum_{j=0}^{w-1} t_j=\dl$. The
operational meaning of $t$ is that the variable node has $t_j$ edges which
connect to a check node at position $i+j$. There are
$\binom{\dl+w-1}{w-1}$ types.  Assume that for each variable we
order its edges in an arbitrary but fixed order.  A {\em constellation}
$c$ is an $\dl$-tuple, $c=(c_1, \dots, c_{\dl})$ with elements in $[0,
w-1]$. Its operational significance is that if a variable node at
position $i$ has constellation $c$ then its $k$-th edge is
connected to a check node at position $i+c_k$.  Let $\tau(c)$ denote the
type of a constellation. Since we want the position of each edge
to be chosen independently we impose a uniform distribution on the
set of all constellations.  This imposes the following distribution on
the set of all types. We assign the probability 
\begin{align*} p(t) =
\frac{|\{c: \tau(c)=t\}|}{w^{\dl}}.  
\end{align*} 
Pick $M$ so that $M p(t)$ is a natural number for all
types $t$. For each position $i$ pick $M p(t)$ variables which have their
edges assigned according to type $t$. Further, use a random permutation
for each variable, uniformly chosen from the set of all permutations on
$\dl$ letters, to map a type to a constellation.

Under this assignment, and ignoring boundary effects, for each check position $i$, the number
of edges that come from variables at position $i-j$, $j \in [0,
w-1]$, is $M \frac{\dl}{w}$. In other words, it is exactly a fraction
$\frac{1}{w}$ of the total number $M \dl$ of sockets at position $i$. 
At the check nodes, distribute these edges according to a permutation
chosen uniformly at random from the set of all permutations on $M \dl$ letters,
to the $M \frac{\dl}{\dr}$ check nodes at this position. It is then not
very difficult to see that, under this distribution, for each check node
each edge is roughly independently chosen to be connected to one of its
nearest $w$ ``left'' neighbors. Here, ``roughly independent'' means that
the corresponding probability deviates at most by a term of order $1/M$
from the desired distribution.  As discussed beforehand, we will always
consider the limit in which $M$ first tends to infinity and then the
number of iterations tends to infinity. Therefore, for any fixed number
of rounds of DE the probability model is exactly the independent model
described above.

\begin{lemma}[Design Rate]\label{lem:designrate}
The design rate of the ensemble $(\dl, \dr, L, w)$, with $w \leq 2 L$,
is given by
\begin{align*}
R(\dl, \dr, L, w) & = 
(1-\frac{\dl}{\dr}) - \frac{\dl}{\dr} \frac{w+1-2\sum_{i=0}^{w} 
\bigl(\frac{i}{w}\bigr)^{\dr}}{2 L+1}.
\end{align*}
\end{lemma}
\begin{proof}
Let $V$ be the number of variable nodes and $C$ be the number
of check nodes that are connected to at least one of these variable nodes.
Recall that we {\em define} the design rate as $1-C/V$.

There are $V=M (2 L+1)$ variables in the graph.  The check nodes that
have potential connections to variable nodes
in the range $[-L, L]$ are indexed from $-L$ to $L+w-1$. Consider the
$M \frac{\dl}{\dr}$ check nodes at position $-L$. Each of the $\dr$
edges of each such check node is chosen independently from the range
$[-L-w+1, -L]$. The probability that such a check node
has at least one connection in the range $[-L, L]$ is equal to
$1-\bigl(\frac{w-1}{w}\bigr)^{\dr}$.  Therefore, the expected number
of check nodes at position $-L$ that are connected to the code is equal to
$M \frac{\dl}{\dr} (1-\bigl(\frac{w-1}{w}\bigr)^{\dr})$. In a similar
manner, the expected number of check nodes at position $-L+i$, $i=0,
\dots, w-1$, that are connected to the code is equal to $M \frac{\dl}{\dr}
(1-\bigl(\frac{w-i-1}{w}\bigr)^{\dr})$. All check nodes at positions
$-L+w, \dots, L-1$ are connected. Further, by symmetry, check
nodes in the range $L, \dots, L+w-1$ have an identical contribution
as check nodes in the range $-L, \dots, -L+w-1$. Summing up all
these contributions, we see that the number of check nodes which
are connected is equal to \begin{align*} C & = M \frac{\dl}{\dr} [2 L -
w + 2 \sum_{i=0}^{w} (1-\bigl(\frac{i}{w}\bigr)^{\dr})].  \end{align*}
\end{proof}
{\em Discussion:} \label{dis:designrate} In the above lemma we have
{\em defined} the design rate as the normalized difference of the number of
variable nodes and the number of check nodes that are involved in the
ensemble. This leads to a relatively simple expression which is
suitable for our purposes. But in this ensemble there is a non-zero
probability that there are two or more degree-one check nodes attached
to the same variable node. In this case, some of these degree-one check
nodes are redundant and do not impose constraints. This effect
only happens for variable nodes close to the boundary. 
Since we consider the case where $L$ tends to infinity, this slight difference
between the ``design rate'' and the ``true rate'' does not play a role. We
therefore opt for this simple definition. The design rate is 
a lower bound on the true rate.

\subsection{Other Variants}
There are many variations on the theme that show the same
qualitative behavior.  For real applications these and possibly other
variations are vital to achieve the best trade-offs. 
Let us give a few select examples.

\begin{itemize}
\item[(i)] {\em Diminished Rate Loss}:
One can start with a {\em cycle} (as is the case for tailbiting codes)
rather than a chain so that some of the
extra check nodes which we add at the boundary can be used for
the termination on both sides.  This reduces the rate-loss.
\item[(ii)] {\em Irregular and Structured Ensembles}:
We can start with irregular or structured ensembles.  Arrange a number
of graphs next to each other in a horizontal order.  Couple them by
connecting neighboring graphs up to some order. Emperically, once
the coupling is ``strong'' enough and spread out sufficiently,
the threshold is ``very close'' to the MAP threshold of the underlying ensembles.
 See also \cite{MLC10} for a study of such ensembles. 
\end{itemize}
The main aim of this paper is to explain why coupled LDPC
codes perform so well rather than optimizing the ensemble.  Therefore,
despite the practical importance of these variations, we focus on
the ensemble $(\dl, \dr, L, w)$. It is the simplest to analyze.

\section{General Principle}\label{sec:generalprinciple}
As mentioned before, the basic reason why coupled ensembles have such good
thresholds is that their BP threshold is very close to the MAP threshold
of the underlying ensemble. Therefore, as a starting point, let us
review how the BP and the MAP threshold of the underlying ensemble can
be characterized. A detailed explanation of the following
summary can be found in \cite{RiU08}.

\subsection{The Standard $(\dl, \dr)$-Regular Ensemble: BP versus MAP}\label{sec:standard}
Consider density evolution (DE) of the standard $(\dl, \dr)$-regular ensemble.
More precisely, consider the fixed point (FP) equation
\begin{align} \label{equ:standadde}
x = \epsilon(1-(1-x)^{\dr-1})^{\dl-1},
\end{align}
where $\epsilon$ is the channel erasure value and $x$ is the average erasure probability flowing from the variable
node side to the check node side.  Both the BP as well as the MAP
threshold of the $(\dl, \dr)$-regular ensemble can be characterized in
terms of solutions (FPs) of this equation.

\begin{lemma}[Analytic Characterization of Thresholds]\label{lem:standardthresholds}
Consider the$(\dl, \dr)$-regular ensemble.
Let $\epsilon^{\BPsmall}(\dl, \dr)$ denote its BP threshold
and let $\epsilon^{\MAPsmall}(\dl, \dr)$ denote its MAP threshold.
Define
\begin{align*}
p^{\BPsmall}(x) & = ((\dl-1)(\dr-1) -1) (1-x)^{\dr-2} - \sum_{i=0}^{\dr-3} (1-x)^i, \\
p^{\MAPsmall}(x) & = x + \frac{1}{\dr}(1-x)^{\dr-1}(\dl+\dl(\dr-1)x-\dr x)-\frac{\dl}{\dr}, \\
\epsilon(x) & = \frac{x}{(1-(1-x)^{\dr-1})^{\dl-1}}.
\end{align*}
Let $x^{\BPsmall}$ be the unique positive solution of the equation
$p^{\BPsmall}(x)=0$ and let $x^{\MAPsmall}$ be the unique positive solution of the
equation $p^{\MAPsmall}(x)=0$. 
Then $\epsilon^{\BPsmall}(\dl, \dr)=\epsilon(x^{\BPsmall})$ and
$\epsilon^{\MAPsmall}(\dl, \dr)=\epsilon(x^{\MAPsmall})$.
We remark that above, for ease of notation, we drop the dependence of $x^{\BPsmall}$ and
$x^{\MAPsmall}$ on $\dl$ and $\dr$.  
\end{lemma}
\begin{example}[Thresholds of $(3, 6)$-Ensemble]
Explicit computations show that 
$\epsilon^{\BPsmall}(\dl=3, \dr=6) \approx 0.42944$ and
$\epsilon^{\MAPsmall}(\dl=3, \dr=6) \approx 0.488151$.
\end{example}

\begin{lemma}[Graphical Characterization of Thresholds]
The left-hand side of Figure~\ref{fig:ebpexit36} shows the so-called extended BP (EBP)
EXIT curve associated to the $(3, 6)$-regular ensemble.  This is the
curve given by $\{\epsilon(x), (1-(1-x)^{\dr-1})^{\dl}\}$, $0 \leq x \leq 1$.  For all regular
ensembles with $\dl \geq 3$ this curve has a characteristic ``C'' shape.
It starts at the point $(1, 1)$ for $x=1$ and then moves downwards
until it ``leaves'' the unit box at the point $(1, \xunstab(1))$
and extends to infinity.
\begin{figure}[htp]
\centering
\input{ps/ebpexit36}
\caption{\label{fig:ebpexit36} Left: The EBP EXIT curve
$\exitf^{\EBPsmall}$ of the $(\dl=3, \dr=6)$-regular ensemble. 
The curve goes ``outside the box'' at the point $(1, \xunstab(1))$ and tends to infinity.  Right:
The BP EXIT function $\exitf^{\BPsmall}(\epsilon)$. Both
the BP as well as the MAP threshold are determined by $\exitf^{\BPsmall}(\epsilon)$.} 
\end{figure}
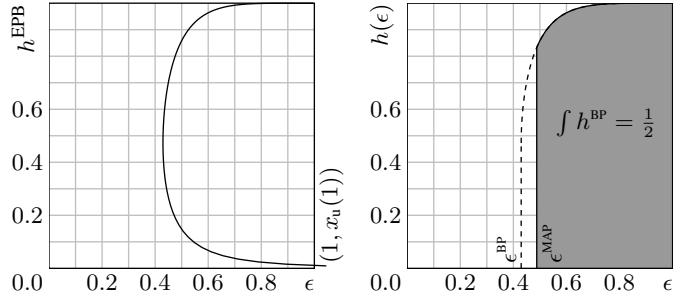
The right-hand side of Figure~\ref{fig:ebpexit36} shows the BP EXIT
curve (dashed line). It is constructed from the EBP EXIT curve by ``cutting off'' the
lower branch and by completing the upper branch via a vertical line.

The BP threshold $\epsilon^{\BPsmall}(\dl, \dr)$ is the point at which
this vertical line hits the $x$-axis. In other words, the BP threshold
$\epsilon^{\BPsmall}(\dl, \dr)$ is equal to the smallest $\epsilon$-value
which is taken on along the EBP EXIT curve.

\begin{lemma}[Lower Bound on $x^{\BPsmall}$]\label{lem:lowerboundxBP}
For the $(\dl, \dr)$-regular ensemble 
\begin{align*}
x^{\BPsmall}(\dl, \dr) & \geq 1-(\dl-1)^{-\frac{1}{\dr-2}}.
\end{align*}
\end{lemma}
\begin{proof}
Consider the polynomial $p^{\BPsmall}(x)$. Note 
that $p^{\BPsmall}(x) \geq \tilde{p}(x)=((\dl-1)(\dr-1)-1) (1-x)^{\dr-2} - (\dr-2)$
for $x \in [0, 1]$.  Since 
$p^{\BPsmall}(0) \geq \tilde{p}(0) = (\dl-2)(\dr-1) > 0$, 
the positive root of $\tilde{p}(x)$ is a
lower bound on the positive root of $p^{\BPsmall}(x)$.  
But the positive root of $\tilde{p}(x)$ is at 
$1- (\frac{\dr-2}{(\dl-1)(\dr-1)-1})^{\frac{1}{\dr-2}}$.
This in turn is lower bounded by $1-(\dl-1)^{-\frac{1}{\dr-2}}$.
\end{proof}

To construct the MAP threshold $\epsilon^{\MAPsmall}(\dl, \dr)$,
integrate the BP EXIT curve starting at $\epsilon=1$ until the area
under this curve is equal to the design rate of the code.  The point at
which equality is achieved is the MAP threshold (see the right-hand side
of Figure~\ref{fig:ebpexit36}).  \end{lemma}

\begin{lemma}[MAP Threshold for Large Degrees]\label{lem:largethreshold}
Consider the $(\dl, \dr)$-regular ensemble.
Let $r(\dl, \dr)=1-\frac{\dl}{\dr}$ denote the design rate so that $\dr = \frac{\dl}{1-r}$. Then,
for $r$ fixed and $\dl$ increasing, the MAP threshold
$\epsilon^{\MAPsmall}(\dl, \dr)$ converges exponentially fast (in $\dl$) to $1-r$.
\end{lemma}
\begin{proof}
Recall that the MAP threshold is determined by the unique positive solution
of the polynomial equation $p^{\MAPsmall}(x)=0$, where $p^{\MAPsmall}(x)$ is given in Lemma~\ref{lem:standardthresholds}.
A closer look at this equation shows that this solution has the form
\begin{align*}
x = (1-r)\bigl(1-\frac{r^{\frac{\dl}{1-r}-1}(\dl+r-1)}{1-r^{\frac{\dl}{1-r}-2}
(1+\dl(\dl+r-2))} + o(\dl r^{\frac{\dl}{1-r}}) \bigr).
\end{align*}
We see that the root converges exponentially fast (in $\dl$) to $1-r$.
Further, in terms of this root we can write the MAP threshold as
\begin{align*}
x(1+\frac{1-r-x}{(\dl+r-1) x})^{\dl-1}.
\end{align*}
\end{proof}
\begin{lemma}[Stable and Unstable Fixed Points -- \cite{RiU08}]
\label{lem:stableandunstable}
Consider the standard $(\dl, \dr)$-regular ensemble with $\dl \geq 3$.
Define 
\begin{align}\label{equ:hfunction}
h(x) & = \epsilon (1-(1-x)^{\dr-1})^{\dl-1} - x.
\end{align}
Then, for $\epsilon^{\BP}(\dl, \dr) < \epsilon \leq 1$, there are exactly 
two strictly positive solutions of the equation $h(x)=0$
and they are both in the range $[0, 1]$.

Let  $\xstable$ be the larger of the two
and let $\xunstable$ be the smaller of
the two.  Then $\xstable$ is a strictly increasing
function in $\epsilon$ and $\xunstable$
is a strictly decreasing function in $\epsilon$. Finally,
$\xstab(\epsilon^\BPsmall)=\xunstab(\epsilon^\BPsmall)$.
\elemma
{\em Discussion:} Recall that $h(x)$ represents the change of the erasure
probability of DE in one iteration, assuming that the system has current
erasure probability $x$.  This change can be negative (erasure
probability decreases), it can be positive, or it can be zero (i.e., there is
a FP). We discuss some useful properties of $h(x)$ in Appendix~\ref{app:propertyofh(x)}. 

As the notation indicates, $\xstab$ corresponds to a {\em stable} FP
whereas $\xunstab$ corresponds to an unstable FP.  Here
stability means that if we initialize DE with the value $\xstable+\delta$
for a sufficiently small $\delta$ then DE converges back to $\xstable$.

\subsection{The $(\dl, \dr, L)$ Ensemble}
Consider the EBP EXIT curve of the $(\dl, \dr, L)$ ensemble.  To
compute this curve we proceed as follows. We fix a desired ``entropy''
value, see Definition~\ref{def:entropy}, call it $\chi$. We initialize
DE with the constant $\chi$. We then repeatedly perform one step
of DE, where in each step we fix the channel parameter in such a
way that the resulting entropy is equal to $\chi$. This is equivalent
to the procedure introduced in \cite[Section VIII]{MMRU09} to compute
the EBP EXIT curve for general binary-input memoryless output-symmetric
channels. Once the procedure has converged, we plot its EXIT value
versus the resulting channel parameter. We then repeat the procedure for
many different entropy values to produce a whole curve.

Note that DE here is not just DE for the underlying ensemble.  Due
to the spatial structure we in effect deal with a multi-edge ensemble \cite{RiU04b}
with many edge types. For our current casual discussion the exact
form of the DE equations is not important, but if you are curious
please fast forward to Section~\ref{sec:proof}.

Why do we use this particular procedure?  By using forward DE, one
can only reach stable FPs. But the above procedure allows one to
find points along the whole EBP EXIT curve, i.e., one can in
particular also produce {\em unstable} FPs of DE.

The resulting curve is shown in Figure~\ref{fig:lrLexit} for various values of $L$.
\begin{figure}[htp]
\begin{centering}
\input{ps/lrLexit}
\caption{EBP EXIT curves of the ensemble $(\dl=3, \dr=6, L)$ 
for $L=1, 2, 4, 8, 16, 32, 64$, and $128$.
The BP/MAP thresholds are 
$\epsilon^{\BPsmall/\MAPsmall}(3, 6, 1)=0.714309/0.820987$, 
$\epsilon^{\BPsmall/\MAPsmall}(3, 6, 2)=0.587842/0.668951$, 
$\epsilon^{\BPsmall/\MAPsmall}(3, 6, 4)=0.512034/0.574158$, 
$\epsilon^{\BPsmall/\MAPsmall}(3, 6, 8)=0.488757/0.527014$, 
$\epsilon^{\BPsmall/\MAPsmall}(3, 6, 16)=0.488151/0.505833$, 
$\epsilon^{\BPsmall/\MAPsmall}(3, 6, 32)=0.488151/0.496366$, 
$\epsilon^{\BPsmall/\MAPsmall}(3, 6, 64)=0.488151/0.492001$, 
$\epsilon^{\BPsmall/\MAPsmall}(3, 6, 128)=0.488151/0.489924$.
The light/dark gray areas mark the interior of the BP/MAP EXIT function
of the underlying $(3, 6)$-regular ensemble, respectively.
}
\label{fig:lrLexit}
\end{centering}
\end{figure}
Note that these EBP EXIT curves show a dramatically different behavior
compared to the EBP EXIT curve of the underlying ensemble. These
curves appear to be ``to the right'' of the threshold $\epsilon^{\MAPsmall}(3,
6) \approx 0.48815$.  For small values of $L$ one might be led to
believe that this is true since the design rate of such an ensemble is
considerably smaller than $1-\dl/\dr$.  But even for large values of
$L$, where the rate of the ensemble is close to $1-\dl/\dr$, this dramatic
increase in the threshold is still true.
Emperically we see that, for $L$ increasing, the EBP EXIT curve
approaches the MAP EXIT curve of the underlying $(\dl=3, \dr=6)$-regular ensemble. In
particular, for $\epsilon \approx \epsilon^{\MAPsmall}(\dl, \dr)$ the
EBP EXIT curve drops essentially vertically until it hits zero. We will
see that this is a fundamental property of this construction.

\subsection{Discussion}
A look at Figure~\ref{fig:lrLexit} might convey the impression that the
transition of the EBP EXIT function is completely flat and that the threshold
of the ensemble $(\dl, \dr, L)$ is exactly equal to the MAP threshold of
the underlying $(\dl, \dr)$-regular ensemble when $L$ tends to infinity.

Unfortunately, the actual behavior is more subtle.
Figure~\ref{fig:lrLwiggle} shows the EBP EXIT curve for $L=32$ with a
small section of the transition greatly magnified.  As one can see from
this magnification, the curve is not flat but exhibits small ``wiggles''
in $\epsilon$ around $\epsilon^{\MAPsmall}(\dl, \dr)$.  These wiggles
do not vanish as $L$ tends to infinity but their width remains constant.
As we will discuss in much more detail later, area considerations imply
that, in the limit as $L$ diverges to infinity, the BP threshold is
slightly below $\epsilon^{\MAPsmall}(\dl, \dr)$. Although this does not
play a role in the sequel, let us remark that the number of wiggles is
(up to a small additive constant) equal to $L$. 

Where do these wiggles come from?  They stem from the fact that the
system is discrete. If, instead of considering a system with sections
at integer points, we would deal with a continuous system where
neighboring "sections" are infinitesimally close, then these wiggles
would vanish.  This ``discretization'' effect is well-known in the
physics literature. By letting $w$ tend to infinity we can in effect
create a continuous system.  This is in fact our main motivation
for introducing this parameter.

Emperically, these wiggles are very small (e.g., they are of width
$10^{-7}$ for the $(\dl=3, \dr=6, L)$ ensemble), and further, these
wiggles tend to $0$ when $\dl$ is increased.  Unfortunately this is hard
to prove.
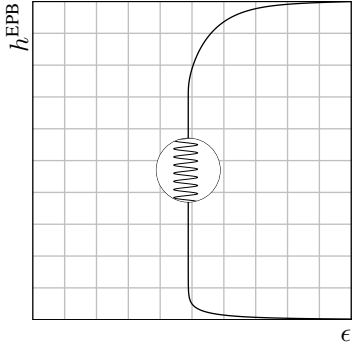
\begin{figure}[htp]
\begin{centering}
\input{ps/lrLwiggle}
\caption{
EBP EXIT curve for the $(\dl=3, \dr=6, L=32)$ ensemble.  The circle shows a
magnified portion of the curve. The horizontal magnification is $10^7$,
the vertical one is $1$.}
\label{fig:lrLwiggle}
\end{centering}
\end{figure}

We therefore study the ensemble $(\dl, \dr, L, w)$. The
wiggles for this ensemble are in fact larger, see e.g. Figure~\ref{fig:lrLwwiggle}.
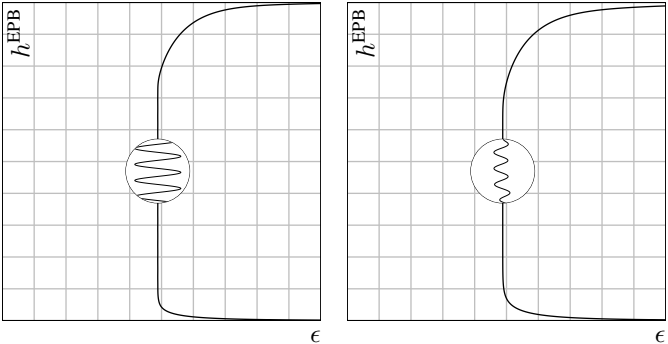
\begin{figure}[htp]
\begin{centering}
\input{ps/lrLwwiggle}
\caption{EBP EXIT curve for the $(\dl=3, \dr=6, L=16, w)$ ensemble. Left:
$w=2$; The circle shows a magnified portion of the curve. The horizontal
magnification is $10^3$, the vertical one is $1$. Right: $w=3$; The circle
shows a magnified portion of the curve. The horizontal magnification is
$10^6$, the vertical one is $1$.
}
\label{fig:lrLwwiggle}
\end{centering}
\end{figure}
But, as mentioned above, the wiggles can be made arbitrarily small by letting
$w$ (the smoothing parameter) tend to infinity. E.g., in the left-hand
side of Figure~\ref{fig:lrLwwiggle}, $w=2$, whereas in the right-hand
side we have $w=3$. We see that the wiggle size has decreased by more
than a factor of $10^3$.

\section{Main Statement and Interpretation}
As pointed out in the introduction, numerical experiments indicate that
there is a large class of convolutional-like LDPC ensembles that all have
the property that their BP threshold is ``close'' to the MAP threshold
of the underlying ensemble. Unfortunately, no general theorem is known
to date that states when this is the case.  The following theorem gives
a particular instance of what we believe to be a general principle. The
bounds stated in the theorem are loose and can likely be improved
considerably. Throughout the paper we assume that $\dl \geq 3$.

\subsection{Main Statement}
\begin{theorem}[BP Threshold of the $(\dl, \dr, L, w)$ Ensemble]\label{the:main}
Consider transmission over the BEC$(\epsilon)$ using random elements from
the ensemble $(\dl, \dr, L, w)$.  Let $\epsilon^{\BPsmall}(\dl, \dr,
L, w)$ denote the BP threshold and let $R(\dl, \dr, L, w)$ denote the
design rate of this ensemble.  

Then, in the limit as $M$ tends to infinity, and for $w>\max\Big\{2^{16}, 2^4\dl^2\dr^2, \frac{(2\dl\dr(1+\frac{2\dl}{1\!-\!2^{\!-\!1\!/\!(\dr\!-\!2)}}))^8}{(1\!-\!2^{\!-1\!/\!(\dr\!-\!2)})^{16}(\frac12(1\!-\!\frac{\dl}{\dr}))^8}\Big\}$,
\begin{align}
 \epsilon^{\BPsmall}(\dl, \dr, L, w) & \leq 
 \epsilon^{\MAPsmall}(\dl, \dr, L, w) \leq \nonumber \\ 
 & \epsilon^{\MAPsmall}(\dl, \dr)\!+\!\frac{w-1}{2 L(1\!-\!(1\!-\!x^{\MAPsmall}(\dl, \dr))^{\dr\!-\!1})^{\dl}}\\
 \epsilon^{\BPsmall}(\dl, \dr, L, w) & \geq 
\Big( \epsilon^{\MAPsmall}(\dl, \dr) \!-\! 
 w^{-\frac18}\frac{8\dl\dr  + \frac{4\dr\dl^2}{(1-4w^{-\frac18})^{\dr}}}{(1\!-\!2^{-\frac{1}{\dr}})^2}
 \Big) \nonumber \\ & \,\,\,\,\,\,\,\,\times \big(1-4w^{-1/8}\big)^{\dr\dl}. 
\label{equ:epslowerbound}
\end{align}

In the limit as $M$, $L$ and $w$ (in that order) tend to infinity,
\begin{align}
\lim_{w \rightarrow \infty}
\lim_{L \rightarrow \infty}
 R(\dl, \dr, L, w) & = 1 - \frac{\dl}{\dr}, \label{equ:ratelimit} \\
\lim_{w \rightarrow \infty}
\lim_{L \rightarrow \infty}
 \epsilon^{\BPsmall}(\dl, \dr, L, w) & = 
\lim_{w \rightarrow \infty}
\lim_{L \rightarrow \infty}
\epsilon^{\MAPsmall}(\dl, \dr, L, w)  \nonumber \\ 
& = \epsilon^{\MAPsmall}(\dl,\dr) . \label{equ:epslimit} 
\end{align}
\end{theorem}

{\em Discussion:}
\begin{itemize}
\item[(i)]
The lower bound on $\epsilon^{\BPsmall}(\dl,\dr,L,w)$ is the main result
of this paper. It shows that, up to a term which tends to zero when
$w$ tends to infinity, the threshold of the chain is equal to the MAP
threshold of the underlying ensemble.  The statement in the theorem
is weak. As we discussed earlier, the convergence speed w.r.t. $w$
is most likely exponential.  We prove only a convergence speed of
$w^{-\frac18}$. We pose it as an open problem to improve this bound. We
also remark that, as seen in \eqref{equ:epslimit}, the MAP threshold
of the $(\dl, \dr, L, w)$ ensemble tends to $\epsilon^{\MAPsmall}(\dl,
\dr)$ for any finite $w$ when $L$ tends to infinity, whereas the BP threshold is bounded away from
$\epsilon^{\MAPsmall}(\dl, \dr)$ for any finite $w$.

\item[(ii)] 
We right away prove the upper bound on $\epsilon^{\BPsmall}(\dl, \dr,
L, w)$.  For the purpose of our proof, we first consider a ``circular''
ensemble.  This ensemble is defined in an identical manner as the $(\dl,
\dr, L, w)$ ensemble except that the  positions are now from $0$ to $K-1$
and index arithmetic is performed modulo $K$. This circular ensemble has
design rate equal to $1-\dl/\dr$.  Set $K=2L+w$. The original ensemble
is recovered by setting any consecutive $w-1$ positions to zero. 
We first provide a lower bound on the conditional entropy for
the circular ensemble when transmitting over a BEC with parameter $\epsilon$.
We then show that setting $w-1$ sections to $0$, does not significantly decrease this
entropy. Overall this gives an upper bound on the MAP threshold of the original ensemble. 

It is not hard to see that the BP EXIT curve\footnote{The BP EXIT curve is the
plot of the extrinsic estimate of the BP decoder versus the channel erasure
fraction (see \cite{RiU08} for details).} is the same for both the $(\dl,
\dr)$-regular ensemble and the circular ensemble.  Indeed, the forward DE (see
Definition~\ref{def:forwardDE}) converges to the same fixed-point for both
ensembles. Consider the $(\dl, \dr)$-regular ensemble and let $\epsilon \in
[\epsilon^{\MAPsmall}(\dl,\dr), 1]$.  The conditional entropy when transmitting
over a BEC with parameter $\epsilon$ is at least equal to $1-\dl/\dr$ minus the
area under the BP EXIT curve between $[\epsilon, 1]$ (see Theorem 3.120 in
\cite{RiU08}).  Call this area $A(\epsilon)$.  Here, the entropy is normalized
by $K M$, where $K$ is the length of the circular ensemble and $M$ denotes the
number of variable nodes per section.  Assume now that we set $w-1$ consecutive
sections of the circular ensemble to $0$ in order to recover the original
ensemble.  As a consequence, we ``remove'' an entropy (degrees of freedom) of
at most $(w-1)/K$ from the circular system. The remaining entropy is therefore
positive (and hence we are above the MAP threshold of the circular ensemble) as
long as $1-\dl/\dr-(w-1)/K-A(\epsilon) > 0$. Thus the MAP threshold of the
circular ensemble is given by the supremum over all $\epsilon$ such that
$1-\dl/\dr-(w-1)/K-A(\epsilon) \leq 0$.  Now note that
$A(\epsilon^{\MAPsmall}(\dl,\dr))=1-\dl/\dr$, so that the above condition becomes
$A(\epsilon^{\MAPsmall}(\dl,\dr))-A(\epsilon) \leq (w-1)/K$.  But the BP EXIT
curve is an increasing function in $\epsilon$ so that
$A(\epsilon^{\MAPsmall}(\dl,\dr))-A(\epsilon) >
(\epsilon-\epsilon^{\MAPsmall}(\dl,\dr)) (1\!-\!(1\!-\!x^{\MAPsmall}(\dl,
\dr))^{\dr\!-\!1})^{\dl}$. We get the stated upper bound on
$\epsilon^{\MAPsmall}(\dl, \dr, L, w)$ by lower bounding $K$ by $2 L$.


\item[(iii)]
According to Lemma~\ref{lem:designrate}, $\lim_{L \rightarrow \infty}
\lim_{M \rightarrow \infty} R(\dl, \dr, L, w) = 1 - \frac{\dl}{\dr}$. This
immediately implies the limit (\ref{equ:ratelimit}). The limit for
the BP threshold $\epsilon^{\BPsmall}(\dl, \dr, L, w)$ follows from
(\ref{equ:epslowerbound}).

\item[(iv)] According to Lemma~\ref{lem:largethreshold}, the MAP threshold
$\epsilon^{\MAPsmall}(\dl, \dr)$ of the underlying ensemble quickly
approaches the Shannon limit. We therefore see that convolutional-like
ensembles provide a way of approaching capacity with low
complexity.  E.g., for a rate equal to one-half, we get 
$\epsilon^{\MAPsmall}(\dl=3, \dr=6)=0.48815$, 
$\epsilon^{\MAPsmall}(\dl=4, \dr=8)=0.49774$, 
$\epsilon^{\MAPsmall}(\dl=5, \dr=10)=0.499486$, 
$\epsilon^{\MAPsmall}(\dl=6, \dr=12)=0.499876$, 
$\epsilon^{\MAPsmall}(\dl=7, \dr=14)=0.499969$.
\end{itemize}

\subsection{Proof Outline}\label{sec:proofoutline}
The proof of the lower bound in Theorem~\ref{the:main} is long.  We therefore break it
up into several steps.  Let us start by discussing each of the steps
separately. This hopefully clarifies the main ideas. But it will also be
useful later when we discuss how the main statement can potentially be
generalized. We will see that some steps are quite generic, whereas other
steps require a rather detailed analysis of the particular chosen system.

\begin{itemize}
\item[(i)] {\em Existence of FP:} 
``The'' key to the proof is to show the existence
of a unimodal FP $(\epsilon^*, \x^*)$ which takes on
an essentially constant value in the ``middle'', has a fast
``transition'', and has arbitrarily small values towards the
boundary (see Definition~\ref{def:fixedpoints}). Figure~\ref{fig:accordeonfp}
\begin{figure}[htp]
\begin{centering}
\input{ps/accordeonfp}
\caption{Unimodal FP of the $(\dl=3, \dr=6, L=16, w=3)$ 
ensemble with small values towards the boundary, a fast transition,
and essentially constant values in the middle.}
\label{fig:accordeonfp}
\end{centering}
\end{figure}
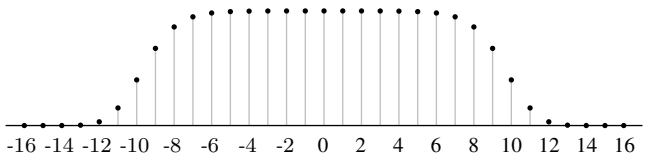
shows a typical such example.  We will see later that the
associated channel parameter of such a FP, $\epsilon^*$, is necessarily very close to
$\epsilon^{\MAPsmall}(\dl, \dr)$.

\item[(ii)] {\em Construction of EXIT Curve:}
Once we have established the existence of such a special FP we construct
from it a whole {\em FP family}.  The elements in this family of FPs
look essentially identical. They differ only in their 
``width.'' This width changes continuously,
initially being equal to roughly $2L+1$ until
it reaches zero. As we will see, this family ``explains'' how the
overall constellation (see Definition~\ref{def:fixedpoints}) collapses once the channel parameter has reached a
value close to $\epsilon^{\MAPsmall}(\dl, \dr)$: starting from the two
boundaries, the whole constellation ``moves in'' like a wave until the
two wave ends meet in the middle. The EBP EXIT curve is a projection of
this wave (by computing the EXIT value of each member of the family).
If we look at the EBP EXIT curve, this phenomenon corresponds to the
very steep vertical transition close to $\epsilon^{\MAPsmall}(\dl, \dr)$.

Where do the wiggles in the EBP EXIT curve come from? Although the
various FPs look ``almost'' identical (other than the place of the
transition) they are not exactly identical. The $\epsilon$ value changes
very slightly (around $\epsilon^*$).  The larger we choose $w$ the smaller we can make the
changes (at the cost of a longer transition).

When we construct the above family of FPs it is mathematically convenient
to allow the channel parameter $\epsilon$ to depend on the position. Let
us describe this in more detail.

We start with a special FP as depicted in Figure~\ref{fig:accordeonfp}.
From this we construct a smooth family $(\underline{\epsilon}(\alpha),
\x(\alpha))$, parameterized by $\alpha$, $\alpha \in [0, 1]$, where
$\x(1)=\underline{1}$ and where $\x(0)=\underline{0}$.  The components
of the vector $\underline{\epsilon}(\alpha)$ are essentially constants
(for $\alpha$ fixed).  The possible exceptions are components 
towards the boundary. We allow those components to take on larger (than
in the middle) values.

From the family $(\underline{\epsilon}(\alpha), \x(\alpha))$ we derive
an EBP EXIT curve and we then measure the area enclosed by this curve.
We will see that this area is close to the design rate.  From this we will
be able to conclude that $\epsilon^* \approx \epsilon^{\MAPsmall}(\dl,
\dr)$.

\item[(iii)] {\em Operational Meaning of EXIT Curve:} 
We next show that the EBP EXIT curve constructed in step (ii) has an
operational meaning. More precisely, we show that if we pick a channel
parameter sufficiently below $\epsilon^*$ then forward DE converges to
the trivial FP.

\item[(iv)] {\em Putting it all Together:}
The final step is to combine all the constructions and bounds discussed
in the previous steps to show that $\epsilon^{\BPsmall}(\dl, \dr, w,
L)$ converges to $\epsilon^{\MAPsmall}(\dl, \dr)$ when $w$ and $L$
tend to infinity.  \end{itemize}

\section{Proof of Theorem~\ref{the:main}}\label{sec:proof}
This section contains the technical details of Theorem~\ref{the:main}.
We accomplish the proof by following the steps outlined in the previous
section. To enhance the readability of this section we have moved some
of the long proofs to the appendices.

\subsection{Step (i): Existence of FP}
\begin{definition}[Density Evolution of $(\dl, \dr, L, w)$ Ensemble]
Let $x_i$, $i\in \integers$, denote the average erasure probability which
is emitted by variable nodes at position $i$. For $i \not \in [-L, L]$
we set $x_i=0$.
For $i \in [-L, L]$ the FP condition implied by DE is
\begin{align}\label{eq:densevolxi}
x_i 
& = \epsilon\Bigl(1-\frac{1}{w} \sum_{j=0}^{w-1} \bigl(1-\frac{1}{w} \sum_{k=0}^{w-1} x_{i+j-k} \bigr)^{\dr-1} \Bigr)^{\dl-1}.
\end{align}
If we define
\begin{align}\label{equ:fi}
f_i & =  \Bigl( 1- \frac1w\sum_{k=0}^{w-1}x_{i-k} \Bigr)^{\dr-1},
\end{align}
then (\ref{eq:densevolxi}) can be rewritten as
\begin{align*}
x_i & = \epsilon \Bigl(1-\frac{1}{w} \sum_{j=0}^{w-1} f_{i+j}\Bigr)^{\dl-1}.
\end{align*}
In the sequel it will be handy to have an even shorter form for the right-hand
side of (\ref{eq:densevolxi}).  Therefore, let
\begin{align}
g(x_{i-w+1}, \dots, x_{i+w-1}) &= \Bigl(1-\frac{1}{w} \sum_{j=0}^{w-1} f_{i+j}\Bigr)^{\dl-1}.
\end{align}
Note that 
\begin{align*}
g(x, \dots, x) = (1-(1-x)^{\dr-1})^{\dl-1},
\end{align*}
where the right-hand side represents DE for the underlying $(\dl,
\dr)$-regular ensemble.

The function $f_i(x_{i-w+1}, \dots, x_{i})$ defined in
(\ref{equ:fi}) is decreasing in all its arguments $x_j \in
[0, 1]$, $j=i-w+1, \dots, i$. In the sequel, it is understood
that $x_i \in [0,1]$. The channel parameter $\epsilon$ is allowed
to take values in $\mathbb{R}^+$.
\end{definition}

\begin{definition}[FPs of Density Evolution]\label{def:fixedpoints}
Consider DE for the $(\dl, \dr, L, w)$ ensemble.
Let $\x=(x_{-L}, \dots, x_{L})$. We call $\x$ the {\em constellation}. 
We say that $\x$ forms a FP
of DE with parameter $\epsilon$ if $\x$ fulfills (\ref{eq:densevolxi})
for $i \in [-L, L]$.  As a short hand we then say that $(\epsilon, \x)$
is a FP.  We say that $(\epsilon, \x)$ is a {\em non-trivial}
FP if $\x$ is not identically zero.
More generally, let 
$$
\underline{\epsilon} = (\epsilon_{-L}, \dots, \epsilon_0,\dots, \epsilon_L),  
$$
where $\epsilon \in \mathbb{R}^+$ for $i\in [-L,L]$. 
We say that $(\underline{\epsilon}, \x)$ forms a FP if 
\begin{align}\label{eq:densevolxigeneral}
x_i = \epsilon_i g(x_{i-w+1}, \dots, x_{i+w-1}), \quad i\in [-L,L].
\end{align}\qed
\end{definition}

\begin{definition}[Forward DE and Admissible Schedules]\label{def:forwardDE} 
Consider DE for the $(\dl, \dr, L, w)$ ensemble.  More
precisely, pick a parameter $\epsilon \in [0,1]$. Initialize 
$\x^{(0)}=(1, \dots, 1)$. Let $\x^{(\ell)}$ be the result of
$\ell$ rounds of DE. I.e., $\x^{(\ell+1)}$ is generated from
$\x^{(\ell)}$ by applying the DE equation \eqref{eq:densevolxi} to each
section $i\in [-L,L]$,
\begin{align*}
x_i^{(\ell+1)} & = \epsilon g(x_{i-w+1}^{(\ell)},\dots,x_{i+w-1}^{(\ell)}).
\end{align*}
We call this the {\em parallel} schedule. 

More generally, consider a schedule in which in each step $\ell$
an arbitrary subset of the sections is updated, constrained only by
the fact that every section is updated in infinitely many steps. We
call such a schedule {\em admissible}. Again, we call $\x^{(\ell)}$
the resulting sequence of constellations.  

In the sequel we will refer to this procedure as {\em forward} DE
by which we mean the appropriate {\em initialization} and the
subsequent DE procedure. E.g., in the next lemma we will discuss
the FPs which are reached under forward DE.  These FPs have special
properties and so it will be convenient to be able to refer to them
in a succinct way and to be able to distinguish them from general
FPs of DE.  
\end{definition}

\begin{lemma}[FPs of Forward DE]\label{lem:forwardDE} 
Consider forward DE for the $(\dl, \dr, L, w)$ ensemble.  Let
$\x^{(\ell)}$ denote the sequence of constellations under an admissible
schedule.  Then $\x^{(\ell)}$ converges to a FP of DE and this
FP is independent of the schedule.  In particular, it is equal
to the FP of the parallel schedule.
\end{lemma}
\begin{proof}
Consider first the parallel schedule.  We claim that the vectors
$\x^{(\ell)}$ are ordered, i.e., $\x^{(0)}\geq \x^{(1)}\geq \dots
\geq \underline{0}$ (the ordering is pointwise).  This is true since
$\x^{(0)}=(1, \dots, 1)$, whereas $\x^{(1)}\leq(\epsilon, \dots, \epsilon)
\leq (1, \dots, 1)=\x^{(0)}$. It now follows by induction on the number
of iterations that the sequence $\x^{(\ell)}$ is monotonically decreasing.

Since the sequence $\x^{(\ell)}$ is also bounded from below it
converges. Call the limit $\x^{(\infty)}$.  Since the DE equations
are continuous it follows that $\x^{(\infty)}$ is a fixed point of DE
\eqref{eq:densevolxi} with parameter $\epsilon$. We call $\x^{(\infty)}$
the forward FP of DE.

That the limit (exists in general and that it) does not depend on the
schedule follows by standard arguments and we will be brief.  The idea
is that for any two admissible schedules the corresponding computation
trees are nested. This means that if we look at the computation graph
of schedule let's say 1 at time $\ell$ then there exists a time $\ell'$
so that the computation graph under schedule $2$ is a superset of the
first computation graph. To be able to come to this conclusion we have
crucially used the fact that for an admissible schedule every section is
updated infinitely often. This shows that the performance under schedule
2 is at least as good as the performance under schedule 1.  The converse
claim, and hence equality, follows by symmetry.
\end{proof}

\begin{definition}[Entropy]\label{def:entropy}
Let $\x$ be a constellation.  We define the (normalized) {\em entropy} of $\x$ to be
$$\chi(\x)  = \frac1{2L+1}\sum_{i=-L}^L x_i.$$ 
\end{definition} 
{\em Discussion:} More precisely, we should call $\chi(\x)$ the average {\em 
message} entropy. But we will stick with the shorthand entropy in the
sequel.

\begin{lemma}[Nontrivial FPs of Forward DE]\label{lem:nontrivialtwosided}
Consider the ensemble $(\dl, \dr, L, w)$.  Let $\x$ be the FP of forward
DE for the parameter $\epsilon$.  
For $\epsilon \in (\frac{\dl}{\dr}, 1]$ and $\chi \in [0, \epsilon^{\frac{1}{\dl-1}}(\epsilon-\frac{\dl}{\dr}))$, if 
\begin{align}\label{equ:chibound}
L  &  \geq
\frac{w}{2(\frac{\dr}{\dl}(\epsilon-\chi \epsilon^{-\frac{1}{\dl-1}})-1)}
\end{align}
then $\chi(\x) \geq \chi$.
\end{lemma} 
\begin{proof} 
Let $R(\dl, \dr, L, w)$ be the design rate of the $(\dl, \dr, L, w)$
ensemble as stated in Lemma~\ref{lem:designrate}.  Note that the design
rate is a lower bound on the actual rate. It follows that the system has
at least $(2 L+1) R(\dl, \dr, L, w) M$ degrees of freedom. If we transmit
over a channel with parameter $\epsilon$ then in expectation at most
$(2 L+1)(1-\epsilon)M$ of these degrees of freedom are resolved. Recall
that we are considering the limit in which $M$ diverges to infinity. 
Therefore we can work with
averages and do not need to worry about the variation of the quantities
under consideration.  It follows that the number of degrees of freedom left
unresolved, measured per position and normalized by $M$, is at least
$(R(\dl, \dr, L, w)-1+\epsilon)$.

Let $\x$ be the forward DE FP corresponding to parameter $\epsilon$.
Recall that $x_i$ is the average message which flows from a variable
at position $i$ towards the check nodes. From this we can compute the
corresponding probability that the node value at position $i$ has not
been recovered.  It is equal to $\epsilon \bigl( \frac{x_i}{\epsilon}
\bigr)^{\frac{\dl}{\dl-1}} = \epsilon^{-\frac{1}{\dl-1}}
x_i^{\frac{\dl}{\dl-1}}$. Clearly, the BP decoder cannot be better than the
MAP decoder. Further, the MAP decoder cannot resolve the unknown degrees of
freedom. It follows that we must have
\begin{align*}
\epsilon^{-\frac{1}{\dl-1}} \frac{1}{2 L+1} \sum_{i=-L}^{L} x_i^{\frac{\dl}{\dl-1}} \geq 
R(\dl, \dr, L, w)-1+\epsilon.
\end{align*}
Note that $x_i \in [0, 1]$ so that 
$x_i \geq x_i^{\frac{\dl}{\dl-1}}$.
We conclude that
\begin{align*}
\chi(\x) =  
\frac{1}{2 L+1} \sum_{i=-L}^{L} x_i \geq 
\epsilon^{\frac{1}{\dl-1}}(R(\dl, \dr, L, w)-1+\epsilon).
\end{align*}
Assume that we want a constellation with entropy at least $\chi$.
Using the expression for $R(\dl, \dr, L, w)$ from Lemma~\ref{lem:designrate},
this leads to the inequality
\begin{align} \label{equ:intermediate}
\epsilon^{\frac{1}{\dl-1}} (-\frac{\dl}{\dr}- 
\frac{\dl}{\dr} \frac{w+1-2\sum_{i=0}^{w} 
\bigl(\frac{i}{w}\bigr)^{\dr}}{2 L+1}
+\epsilon) \geq \chi.
\end{align}
Solving for $L$ and simplifying the inequality by upper bounding
$1-2\sum_{i=0}^{w} \bigl(\frac{i}{w}\bigr)^{\dr}$ by $0$ and lower
bounding $2L+1$ by $2L$ leads to (\ref{equ:chibound}).
\end{proof}

Not all FPs can be constructed by forward DE. In particular, one can only
reach (marginally) ``stable'' FPs by the above procedure.  Recall from
Section~\ref{sec:proofoutline}, step (i), that we want to construct an unimodal
FP which ``explains'' how the constellation collapses. Such a FP is by
its very nature unstable. 

It is difficult to prove the existence of such a FP by direct
methods. We therefore proceed in stages.  We first show the existence of
a ``one-sided'' increasing FP. We then construct the desired unimodal FP
by taking two copies of the one-sided FP, flipping one copy, and gluing these FPs together.  
\begin{definition}[One-Sided Density Evolution] 
Consider the tuple $\x = (x_{-L}, \dots, x_0)$. The
FP condition implied by {\em one-sided} DE is equal to
\eqref{eq:densevolxi} with $x_{i}= 0$  for $ i < -L$ and $x_{i}= x_{0}$
for $i > 0$.  \end{definition}

\begin{definition}[FPs of One-Sided DE]
We say that $\x$ is a one-sided { \em FP} (of DE) with parameter
$\epsilon$ and length $L$ if (\ref{eq:densevolxi}) is fulfilled for $i \in [-L, 0]$, with 
$x_{i}= 0$  for $ i < -L$ and $x_{i}= x_{0}$
for $i > 0$.

In the same manner as we have done this for two-sided FPs, if
$\underline{\epsilon} = (\epsilon_{-L}, \dots, \epsilon_0)$, then we
define one-sided FPs with respect to $\underline{\epsilon}$.

We say that $\x$ is {\em non-decreasing} if $x_i \leq x_{i+1}$ for
$i=-L,\dots, 0$.  \end{definition}

\begin{definition}[Entropy]\label{def:entropyonesided}
Let $\x$ be a one-sided { \em FP}.  We define the (normalized)
{\em entropy} of $\x$ to be $$\chi(\x)  = \frac1{L+1}\sum_{i=-L}^0 x_i.
$$ \end{definition} 

\begin{definition}[Proper One-Sided FPs]
Let $(\epsilon, \x)$ be a {\em non-trivial} and {\em non-decreasing}
one-sided FP.  As a short hand, we then say that $(\epsilon,
\x)$ is a {\em proper one-sided FP}.  \end{definition}
A proper one-sided FP is shown in Figure~\ref{fig:one-sided_fixed_point}.

\begin{definition}[One-Sided Forward DE and Schedules]\label{def:onesidedforwardDE}
Similar to Definition~\ref{def:forwardDE}, one can define the {\em
one-sided forward DE} by initializing all sections with $1$
and by applying DE according to an admissible schedule. 
\end{definition}

\begin{lemma}[FPs of One-Sided Forward DE]\label{lem:nontrivialonesided}
Consider an $(\dl, \dr, L, w)$ ensemble and let $\epsilon \in [0, 1]$.
Let $\x^{(0)}=(1, \dots, 1)$ and let $\x^{(\ell)}$ denote the result of
applying $\ell$ steps of one-sided forward DE according to an admissible
schedule (cf. Definition~\ref{def:onesidedforwardDE}). Then

\begin{itemize}
\item[(i)]
$\x^{(\ell)}$ converges to a limit which is a FP of one-sided
DE. This limit is independent of the schedule and the limit is either
proper or trivial. As a short hand we say that $(\epsilon,\x)$
is a one-sided FP of forward DE.

\item[(ii)]
For $\epsilon \in (\frac{\dl}{\dr}, 1]$ and $\chi \in [0,
\epsilon^{\frac{1}{\dl-1}}(\epsilon-\frac{\dl}{\dr}))$, if $L$ fulfills
(\ref{equ:chibound}) then $\chi(\x) \geq \chi$.
\end{itemize}
\end{lemma}
\begin{proof}
The existence of the FP and the independence of the schedule
follows along the same line as the equivalent statement
for two-sided FPs in Lemma~\ref{lem:forwardDE}. We hence skip the details.
Assume that this limit $\x^{(\infty)}$ is non-trivial. We want to show
that it is proper. This means we want to show that it is non-decreasing.
We use induction. The initial constellation is non-decreasing.
Let us now show that this property
stays preserved in each step of DE if we apply a parallel schedule. More
precisely, for any section $i \in [-L,0]$, 
\begin{align*}
x_i^{(\ell+1)} & = \epsilon g(x_{i-w+1}^{(\ell)},\dots,x_{i+w-1}^{(\ell)}) \\
& \stackrel{(a)}{\leq} \epsilon g(x_{i+1-w+1}^{(\ell)},\dots,x_{i+1+w-1}^{(\ell)}) \\
& = x_{i+1}^{(\ell+1)},
\end{align*} 
where $(a)$ follows from the monotonicity of $g(\dots)$ and the induction
hypothesis that $\x^{(\ell)}$ is non-decreasing.  

Let us now show that for $\epsilon \in (\frac{\dl}{\dr}, 1]$ and $\chi
\in [0, \epsilon^{\frac{1}{\dl-1}}(\epsilon-\frac{\dl}{\dr}))$, if $L$
fulfills (\ref{equ:chibound}) then $\chi(\x) \geq \chi$.  First, recall
from Lemma~\ref{lem:nontrivialtwosided} that the corresponding two-sided
FP of forward DE has entropy at least $\chi$ under the stated conditions.
Now compare one-sided and two-sided DE for the same initialization with
the constant value $1$ and the parallel schedule. We claim that for
any step the values of the one-sided constellation at position $i$,
$i \in [-L, 0]$, are larger than or equal to the values of the two-sided
constellation at the same position $i$. To see this we use induction. The
claim is trivially true for the initialization. Assume therefore that
the claim is true at a particular iteration $\ell$. For all points $i
\in [-L, -w+1]$ it is then trivially also true in iteration $\ell+1$,
using the monotonicity of the DE map.  For points $i \in [-w+2, 0]$,
recall that the one sided DE ``sees'' the value $x_0$ for all positions
$x_i$, $i \geq 0$, and that $x_0$ is the largest of all $x$-values.
For the two-sided DE on the other hand, by symmetry, $x_i=x_{-i} \leq
x_0$ for all $i\geq 0$. Again by monotonicity, we see that the desired
conclusion holds.

To conclude the proof: note that if for a unimodal two-sided constellation
we compute the average over the positions $[-L, 0]$ then we get at least
as large a number as if we compute it over the whole length $[-L, L]$. This follows
since the value at position $0$ is maximal.  \end{proof}

\begin{figure}[htp]
\begin{centering}
\input{ps/one-sided_fixed_point}
\caption{A proper one-sided FP $(\epsilon, \x)$ for the ensemble $(\dl=3,\dr=6,L=16,w=3)$,
where $\epsilon=0.488151$.
As we will discuss in Lemma~\ref{lem:maximum}, for sufficiently large
$L$, the maximum value of $\x$, namely $x_0$, approaches the stable value
$\xstable$. Further, as discussed in Lemma~\ref{lem:transitionlength}, the
width of the transition is of order $O( \frac{w}{\delta})$, where $\delta>0$ is
a parameter that indicates which elements of the constellation we want to include
in the transition.
} 
\label{fig:one-sided_fixed_point}
\end{centering}
\end{figure}
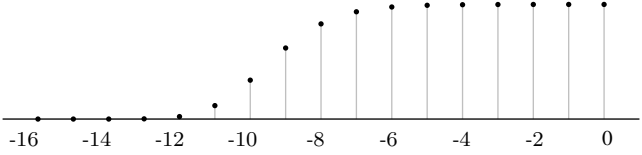

Let us establish some basic properties of proper one-sided FPs. 
\begin{lemma}[Maximum of FP]\label{lem:maximum}
Let $(\epsilon, \x)$, $0 \leq \epsilon \leq 1$, be a proper one-sided FP of length $L$.
Then  $\epsilon > \epsilon^{\BPsmall}(\dl, \dr)$ and
$$
\xunstable \le x_{0} \le \xstable,
$$ 
where $\xstable$ and $\xunstable$ denote the stable and unstable
non-zero FP associated to $\epsilon$, respectively.
\end{lemma}
\begin{proof}
We start by proving that $\epsilon \geq \epsilon^{\BPsmall}(\dl, \dr)$.
Assume to the contrary that $\epsilon < \epsilon^{\BPsmall}(\dl, \dr)$.
Then
\begin{align*}
x_{0} & = \epsilon g(x_{-w+1}, \dots, x_{w-1}) \leq 
\epsilon g(x_{0}, \dots, x_{0})
< x_0,
\end{align*}
a contradiction. Here, the last step follows since $\epsilon <
\epsilon^{\BPsmall}(\dl, \dr)$ and $0 < x_0 \leq 1$.

Let us now consider the claim that $\xunstable \le x_{0} \le \xstable$.
The proof follows along a similar line of arguments.  Since
$\epsilon^{\BPsmall}(\dl, \dr) \leq \epsilon \leq 1$, both $\xstable$
and $\xunstable$ exist and are strictly positive.  Suppose that $x_{0}
> \xstable$ or that $x_{0} < \xunstable$. Then
\begin{align*}
x_{0} & = \epsilon g(x_{-w+1}, \dots, x_{w-1}) \le \epsilon g(x_{0}, \dots, x_{0}) < x_0,
\end{align*}
a contradiction.

A slightly more careful analysis shows that $\epsilon \neq \epsilon^{\BPsmall}$,
so that in fact we have strict inequality, namely $\epsilon > \epsilon^{\BPsmall}(\dl, \dr)$.
We skip the details.
\end{proof}

\begin{lemma}[Basic Bounds on FP]\label{lem:avgprop}
Let $(\epsilon,\x)$ be a proper one-sided FP of length $L$. 
Then for all $i\in [-L,0]$, 
\begin{align*}
& \text{(i)}\,\,\, x_i  \leq  
\epsilon (1-(1-\frac1{w^2}\sum_{j, k=0}^{w-1}x_{i+j-k})^{\dr-1})^{\dl-1}, \\
& \text{(ii)}\,\,\, x_i \leq \epsilon\Big(\frac{\dr-1}{w^2}\sum_{j, k=0}^{w-1}
x_{i+j-k} \Big)^{\dl-1}, \\
& \text{(iii)}\,\,\, x_i  \geq   \epsilon \Big(\frac{1}{w^2}\sum_{j, k=0}^{w-1}x_{i+j-k} \Big)^{\dl-1}, \nonumber \\
& \text{(iv)}\,\,\, x_i  \geq \nonumber \\ 
& \epsilon \Big(\Big(1 - \frac1w\sum_{k=0}^{w-1}x_{i+w-1-k}\Big)^{\dr-2}\frac{\dr-1}{w^2}
\sum_{j, k=0}^{w-1}x_{i+j-k} \Big)^{\dl-1}. \nonumber \\ 
\end{align*}
\end{lemma}
\begin{proof}
We have
\begin{align*}
x_i 
& = \epsilon\Bigl(1-\frac{1}{w} \sum_{j=0}^{w-1} \bigl(1-\frac{1}{w} \sum_{k=0}^{w-1} x_{i+j-k} \bigr)^{\dr-1} \Bigr)^{\dl-1} . 
\end{align*}

Let $\mathfrak{f}(x) = (1-x)^{\dr-1}$, $x \in [0, 1]$. Since $\mathfrak{f}''(x) =
(\dr-1)(\dr-2)(1-x)^{\dr-3}\geq 0$, $\mathfrak{f}(x)$ is convex.
Let $y_{j} = \frac1w\sum_{k=0}^{w-1}x_{i+j-k}$. We have 
\begin{align*}
\frac{1}{w} \sum_{j=0}^{w-1} \bigl(1-\frac{1}{w} \sum_{k=0}^{w-1} x_{i+j-k} \bigr)^{\dr-1} = \frac{1}{w} \sum_{j=0}^{w-1} \mathfrak{f}(y_{j}).
\end{align*}
Since $\mathfrak{f}(x)$ is convex, using Jensen's inequality, we obtain
\begin{align*}
\frac{1}{w} \sum_{j=0}^{w-1} \mathfrak{f}(y_{j}) \geq \mathfrak{f}(\frac1w\sum_{j=0}^{w-1}y_j),
\end{align*}
which proves claim (i).

The derivation of the remaining inequalities is based on the following
identity: 
\begin{align}\label{equ:identity}
1-B^{\dr-1} = (1-B)(1+B+\dots+B^{\dr-2}).
\end{align}
For $0\leq B \leq 1$ this gives rise to the following inequalities:
\begin{align}
& 1-B^{\dr-1} \geq  (\dr-1)B^{\dr-2}(1-B), \label{equ:prodinequa} \\
& 1-B^{\dr-1} \geq (1-B), \label{equ:prodinequb} \\
& 1-B^{\dr-1} \leq (\dr-1) (1-B). \label{equ:prodinequc}
\end{align}
Let $B_j=1 - \frac1w\sum_{k=0}^{w-1}x_{i+j-k}$, so that $1 - f_{i+j} = 1 -
B_j^{\dr-1}$ (recall the definition of $f_{i+j}$ from \eqref{equ:fi}). 
Using (\ref{equ:prodinequb}) this proves (iii):
\begin{align*}
x_i & = \epsilon\Big(\frac1w\sum_{j=0}^{w-1}(1-f_{i+j}) \Big)^{\dl-1} \geq \epsilon\Big(\frac1w\sum_{j=0}^{w-1}(1-B_{j}) \Big)^{\dl-1}\\
& = \epsilon\Big(\frac1w\sum_{j=0}^{w-1}\frac1w\sum_{k=0}^{w-1}x_{i+j-k} \Big)^{\dl-1}.
\end{align*}
If we use (\ref{equ:prodinequc}) instead then we get (ii).
To prove (iv) we use (\ref{equ:prodinequa}):
\begin{align*}
& x_i  \geq \epsilon\Big(\frac{\dr-1}w\sum_{j=0}^{w-1}(1-B_{j})B_j^{\dr-2} \Big)^{\dl-1} = \\
& \epsilon\Big(\frac{\dr-1}w\sum_{j=0}^{w-1}\Bigl(\frac1w\sum_{k=0}^{w-1}x_{i+j-k}\Bigr)\Big(1 - \frac1w\sum_{k=0}^{w-1}x_{i+j-k}\Big)^{\dr-2} \Big)^{\dl-1}.
\end{align*}
Since $\x$ is increasing, $\sum_{k=0}^{w-1}x_{i+j-k} \leq
\sum_{k=0}^{w-1}x_{i+w-1-k}$. Hence,
\begin{align*}
x_i\geq 
\epsilon\Big(\Big(1 \!-\! \frac1w\sum_{k=0}^{w-1}x_{i+w-1-k}\Big)^{\dr-2}\frac{\dr\!-\!1}{w^2}
\sum_{j, k=0}^{w-1}x_{i+j-k} \Big)^{\dl-1}.
\end{align*}
\end{proof}

\begin{lemma}[Spacing of FP]\label{lem:spacing}
Let $(\epsilon, \x)$, $\epsilon \geq 0$, be a proper one-sided FP of length $L$.
Then for $i \in [-L+1, 0]$,
\begin{align*}
x_i-x_{i-1} & \leq \epsilon \frac{ (\dl-1)(\dr-1) \big(\frac{x_i}{\epsilon}\big)^{\frac{\dl-2}{\dl-1}}}{w^2} \Bigl(\sum_{k=0}^{w-1}x_{i+k}\Bigr) \\
& \leq \epsilon \frac{ (\dl-1)(\dr-1) \big(\frac{x_i}{\epsilon}\big)^{\frac{\dl-2}{\dl-1}}}{w}.
\end{align*} 
Let $\xavg_i$ denote the weighted average $\xavg_i=\frac1{w^2}\sum_{j, k=0}^{w-1} x_{i+j-k}$. Then, for any $i \in [-\infty, 0]$, 
\begin{align*}
\xavg_i - \xavg_{i-1} 
& \leq \frac{1}{w^2} \sum_{k=0}^{w-1} x_{i+k} \leq \frac{1}{w}.
\end{align*} 
\end{lemma}
\begin{proof}
Represent both $x_i$ as well as $x_{i-1}$ in terms of the DE 
equation (\ref{eq:densevolxigeneral}).
Taking the difference,
\begin{align}\label{equ:spacingbnd1}
& \frac{x_{i}-x_{i-1}}{\epsilon}  = \nonumber \\  & 
\Bigl(1-\frac{1}{w} \sum_{j=0}^{w-1} f_{i+j}\Bigr)^{\dl-1} -  \Bigl(1-\frac{1}{w} \sum_{j=0}^{w-1} f_{i+j-1}\Bigr)^{\dl-1}.
\end{align}
Apply the identity 
\begin{align}\label{equ:abidentity}
A^m - B^m = (A-B)(A^{m-1} + A^{m-2}B + \dots + B^{m-1}),
\end{align}
where we set  $A=\Bigl(1-\frac{1}{w} \sum_{j=0}^{w-1} f_{i+j}\Bigr)$,
$B=\Bigl(1-\frac{1}{w} \sum_{j=0}^{w-1} f_{i+j-1}\Bigr)$, and $m=\dl-1$.
Note that $A \geq B$.  Thus 
\begin{align} 
& \Bigl(1-\frac{1}{w} \sum_{j=0}^{w-1} f_{i+j}\Bigr)^{\dl-1} -  \Bigl(1-\frac{1}{w} \sum_{j=0}^{w-1} f_{i+j-1}\Bigr)^{\dl-1} \nonumber \\
& = A^{\dl-1} - B^{\dl-1} \nonumber \\
& = (A-B)(A^{\dl-2} + A^{\dl-3}B + \dots + B^{\dl-2})  \nonumber \\
& \stackrel{\text{(i)}}{\leq} (\dl-1)(A-B)A^{\dl-2} \nonumber \\
& \stackrel{\text{(ii)}}{=} \frac{ (\dl-1) A^{\dl-2} }{w} (f_{i-1}-f_{i+w-1} ) \nonumber .
\end{align} 
In step (i) we used the fact that $A\geq B$ implies $A^{\dl-2} \geq
A^pB^q$ for all $p, q \in \naturals$ so that $p+q=\dl-2$.  
In step (ii) we made the substitution $A-B =
\frac{1}w(f_{i-1} - f_{i+w-1})$.
Since $x_i = \epsilon A^{\dl-1}$, $A^{\dl-2} =
\bigl(\frac{x_i}{\epsilon}\bigr)^{\frac{\dl-2}{\dl-1}}$. Thus
\begin{align*}
\frac{x_{i}-x_{i-1}}{\epsilon} 
& \le \frac{ (\dl-1) \bigl(\frac{x_i}{\epsilon}\bigr)^{\frac{\dl-2}{\dl-1}} }{w} 
(f_{i-1}-f_{i+w-1}). 
\end{align*} 
Consider the term $(f_{i-1}-f_{i+w-1})$.  
Set $f_{i-1} = C^{\dr-1}$ and $f_{i+w-1} = D^{\dr-1}$, 
where $C=\Bigl( 1 - \frac1w \sum_{k=0}^{w-1} x_{i-1-k}\Bigr)$ and
$D=\Bigl( 1 - \frac1w \sum_{k=0}^{w-1} x_{i+w-1-k}\Bigr)$.
Note that $0\leq C, D\leq 1$.
Using again (\ref{equ:abidentity}),
\begin{align*}
(f_{i-1} \!-\! f_{i+w-1}) & = (C\!-\!D) (C^{\dr-2} + C^{\dr-3}D + \dots + D^{\dr-2}) \\
& \leq (\dr-1)(C-D) \nonumber. 
\end{align*}
Explicitly,
\begin{align*}
(C-D) & = \frac1w (\sum_{k=0}^{w-1} (x_{i+w-1-k} - x_{i-1-k})) 
\leq \frac1w \sum_{k=0}^{w-1} x_{i+k}, 
\end{align*}
which gives us the desired upper bound. By setting all $x_{i+k} =
1$ we obtain the second, slightly weaker, form.

To bound the spacing for the weighted averages we write $\xavg_i$ and $\xavg_{i-1}$ explicitly,
\begin{align*}
\xavg_i - \xavg_{i-1} & = \frac1{w^2}\Big(  (x_{i+w-1}-x_{i+w-2})  \\ 
&+  2(x_{i+w-2}-x_{i+w-3}) + \dots +  w(x_{i}-x_{i-1}) \\ 
&+ (w-1)(x_{i-1}-x_{i-2}) + \dots + (x_{i-w+1} - x_{i-w}) \Big) \\
& \leq \frac{1}{w^2} \sum_{k=0}^{w-1} x_{i+k} \leq \frac{1}{w}.
\end{align*}
\end{proof}

The proof of the following lemma is long.  Hence we relegate it to
Appendix~\ref{app:transitionlength}.
\begin{lemma}[Transition Length] \label{lem:transitionlength}
Let $w \geq  2\dl^2 \dr^2 $.
Let $(\epsilon, \x)$, $\epsilon \in (\epsilon^{\BPsmall}, 1]$, be a proper
one-sided FP of length $L$.  
Then, for all $0< \delta < \frac{3}{2^5 \dl^4 \dr^6 (1+12 \dl \dr)}$,
\begin{align*}
\vert \{ i: \delta < x_i < \xstable - \delta  \}\vert & \leq   w
\frac{c(\dl, \dr)}{\delta}, 
\end{align*}
where $c(\dl,\dr)$ is a strictly positive constant independent of $L$ and $\epsilon$.
\end{lemma}

Let us now show how we can construct a large class of
one-sided FPs which are not necessarily stable. In particular
we will construct increasing FPs. The proof of the following
theorem is relegated to Appendix~\ref{app:onesidedDE}.
\begin{theorem}[Existence of One-Sided FPs]\label{thm:onesidedDE}
Fix the parameters $(\dl, \dr, w)$ and let 
$\xunstab(1) < \chi$. 
Let $L \geq L(\dl, \dr, w, \chi)$, where
\begin{align*}
 &L(\dl, \dr, w, \chi)  =  \\ & \max\Big\{\frac{4 \dl w}{\dr(1\!-\!\frac{\dl}{\dr})(\chi\!-\!\xunstab(1))}, \frac{8 w}{\kappa^*(1)(\chi\!-\!\xunstab(1))^2}, \\
& \phantom{ L(\dl, \dr, w, \chi)  = \max\Big\{}  \frac{8 w}{\lambda^*(1)(\chi-\xunstab(1))(1-\frac{\dl}{\dr})}, \frac{w}{\frac{\dr}{\dl}-1}
\Big\}. 
\end{align*}
There exists a proper one-sided FP $\x$ of length $L$ that
{\em either} has entropy $\chi$ and channel parameter bounded by
\begin{align*}
\epsilon^{\BPsmall}(\dl, \dr) < \epsilon < 1,
\end{align*}
{\em or} has entropy bounded by
\begin{align*}
 \frac{(1-\frac{\dl}{\dr})(\chi-\xunstab(1))}{8}-\frac{\dl w}{2 \dr(L+1)} \leq \chi(\x) \leq \chi
\end{align*}
and channel parameter $\epsilon=1$.
\end{theorem}
{\em Discussion:} We will soon see that, for the range of parameters of
interest, the second 
alternative is not possible either. In the light of this,
the previous theorem asserts for this range of parameters the existence of a proper FP
of entropy $\chi$. In what follows, this FP will be the key
ingredient to construct the whole EXIT curve.

\subsection{Step (ii): Construction of EXIT Curve}

\begin{definition}[EXIT Curve for $(\dl, \dr, L, w)$-Ensemble]\label{def:EXIT} 
Let $(\epsilon^*, \x^*)$, $0 \leq \epsilon^* \leq 1$, denote a proper
one-sided FP of length $L'$ and entropy $\chi$.  Fix $1\leq L < L'$.

The {\em interpolated} family of constellations based on
$(\epsilon^*, \x^*)$ is denoted by $\{\underline{\epsilon}(\alpha),
\x(\alpha)\}_{\alpha=0}^{1}$. It is indexed from $-L$ to $L$.

This family is
constructed from the one-sided FP $(\epsilon^*, \x^*)$. By definition,
each element $\x(\alpha)$ is symmetric. Hence, it suffices to define the
constellations in the range $[-L, 0]$ and then to set $x_i(\alpha)=x_{-i}(\alpha)$ for $i \in [0, L]$.  
As usual, we set $x_i(\alpha)=0$ for $i\notin [-L,L]$. 
For $i \in [-L, 0]$ and $\alpha
\in [0, 1]$ define
\begin{align*}
x_i(\alpha)  & = \begin{cases} (4 \alpha -3)+(4-4\alpha) x^*_0,
& \alpha \in [\frac34, 1], \\ (4 \alpha -2)x^*_0-(4\alpha-3)
x^*_i, & \alpha \in [\frac12, \frac34), \\ 
a(i, \alpha),
& \alpha \in (\frac14, \frac12), \\
4 \alpha x^*_{i-L'+L}, & \alpha \in (0, \frac14], \end{cases} \\
\epsilon_i(\alpha) & = \frac{x_i(\alpha)}{g((x_{i-w+1}(\alpha), \dots,
(x_{i+w-1}(\alpha))},
\end{align*}
where for $\alpha \in (\frac14, \frac12)$,
\begin{align*}
a(i, \alpha) & = 
{x^*}^{4 (L'-L) (\frac12-\alpha) \!\!\!\!\mod(1)}_{
i-\lceil 4(\frac12-\alpha)(L'-L)\rceil} \cdot {x^*}^{1-4 (L'-L) (\frac12-\alpha) \!\!\!\!\mod(1)}_{
	i-\lceil 4(\frac12-\alpha) (L'-L) \rceil+1}.
\end{align*}

The constellations $\x(\alpha)$ are increasing (component-wise) as
a function of $\alpha$, with 
$\x(\alpha=0)=(0, \dots, 0)$ and with
$\x(\alpha=1)=(1, \dots, 1)$. 
\end{definition}
{\em Remark:} Let us clarify the notation occurring in the definition of the term $a(i, \alpha)$ above. The
expression for $a(i, \alpha)$ consists of the product of two consecutive
sections of $\x^*$, indexed by the subscripts $i-\lceil
4(\frac12-\alpha)(L'-L)\rceil$ and $i-\lceil 4(\frac12-\alpha) (L'-L) \rceil+1$.
The erasure values at the two sections are first raised to the powers $4 (L'-L)
(\frac12-\alpha)\mod(1)$ and $1-4 (L'-L) (\frac12-\alpha)\mod(1)$, before taking
their product. Here,$\mod(1)$ represents real numbers in the interval $[0,1]$. \\ 
{\em Discussion:} 
The interpolation is split into 4 phases. For $\alpha \in [\frac34,
1]$, the constellations decrease from the constant value $1$ to the
constant value $x^*_0$. For the range $\alpha \in [\frac12, \frac34]$,
the constellation decreases further, mainly towards the boundaries,
so that at the end of the interval it has reached the value $x_i^*$
at position $i$ (hence, it stays constant at position $0$). The third
phase is the most interesting one. For $\alpha \in [\frac14, \frac12]$
we ``move in'' the constellation $\x^*$ by ``taking out'' sections in the
middle and interpolating between two consecutive points.  In particular,
the value $a(i, \alpha)$ is the result of ``interpolating'' between two
consecutive $x^*$ values, call them $x^*_j$ and $x^*_{j+1}$, where the
interpolation is done in the exponents, i.e., the value is of the form
${x^*}_j^{\beta} \cdot {x^*}_{j+1}^{1-\beta}$. Finally, in the last phase all
values are interpolated in a linear fashion until they have reached $0$.

\begin{example}[EXIT Curve for $(3, 6, 6, 2)$-Ensemble]
Figure~\ref{fig:interpolation} shows a small example which illustrates this
interpolation for the $(\dl=3, \dr=6, L=6, w=2)$-ensemble.
We start with a FP of entropy $\chi=0.2$ for $L'=12$.
This constellation has $\epsilon^*=0.488223$ and 
\begin{align*}
\x^*=(&
0,
0,
0,
0,
0,
0.015, \\
&
0.131,
0.319,
0.408,
0.428,
0.431,
0.432,
0.432
).
\end{align*}
Note that, even though the constellation is quite short, $\epsilon^*$ is
close to $\epsilon^{\MAPsmall}(\dl=3, \dr=6) \approx 0.48815$, and
$x_0^*$ is close to $\xstab(\epsilon^{\MAPsmall}) \approx 0.4323$.
From $(\epsilon^*, \x^*)$ we create an EXIT curve for $L=6$.
\begin{figure}[htp]
\begin{centering}
\input{ps/interpolation}
\caption{Construction of EXIT curve for $(3, 6, 6, 2)$-ensemble.
The figure shows three particular points in the interpolation, namely
the points $\alpha=0.781$ (phase (i)), $\alpha=0.61$ (phase (ii)), 
and $\alpha=0.4$ (phase (iii)). For each parameter both the constellation $\x$
as well as the local channel parameters $\underline{\epsilon}$ are shown in the figure on left.
The right column of the figure illustrates a projection of the EXIT curve. I.e., we plot the average EXIT value of the constellation versus the channel value of the 0th section. 
For reference, also the EBP EXIT curve of the underlying $(3, 6)$-regular
ensemble is shown (gray line).
}
\label{fig:interpolation}
\end{centering}
\end{figure}
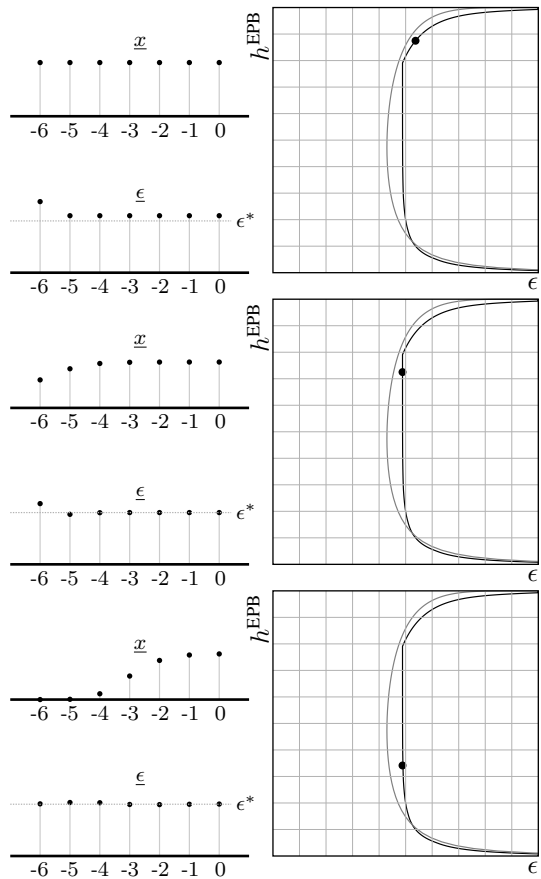
The figure shows $3$ particular points of the interpolation, one in each
of the first $3$ phases.

Consider, e.g., the top figure
corresponding to phase (i). The constellation $\x$ in this case is
completely flat. Correspondingly, the local channel values are also
constant, except at the left boundary, where they are slightly higher
to compensate for the ``missing'' $x$-values on the left.

The second figure from the top shows a point corresponding to phase
(ii).  As we can see, the $x$-values close to $0$ have not changed,
but the $x$-values close to the left boundary decrease towards the
solution $\x^*$. Finally, the last figure shows a point in phase (iii).
The constellation now ``moves in.'' In this phase, the $\epsilon$ values
are close to $\epsilon^*$, with the possible exception of $\epsilon$
values close to the right boundary (of the one-sided constellation). These values can become large.
\end{example}

The proof of the following theorem can be found in Appendix~\ref{app:fundamental}.
\begin{theorem}[Fundamental Properties of EXIT Curve]\label{the:propertiesEXIT}
Consider the parameters $(\dl, \dr, w)$.  Let $(\epsilon^*, \x^*)$, $\epsilon^*\in (\epsilon^{\BPsmall}, 1]$,  
denote a proper one-sided FP of length $L'$ and entropy $\chi>0$.  Then for $1 \leq L
 < L'$, the EXIT curve of Definition \ref{def:EXIT} has the
following properties:
\begin{itemize}
\item[(i)] {\em Continuity:} 
The curve $\{\underline{\epsilon}(\alpha), \x(\alpha)\}_{\alpha=0}^{1}$
is continuous for $\alpha \in [0, 1]$ and differentiable for
$\alpha=[0, 1]$ except for a finite set of points.
\item[(ii)] 
{\em Bounds in Phase (i):} 
For $\alpha \in [\frac34, 1]$, 
\begin{align*}
\epsilon_i(\alpha) & 
\begin{cases}
 =\epsilon_0(\alpha), & i \in [-L+w-1, 0], \\
\geq \epsilon_0(\alpha), & i \in [-L, 0].
\end{cases}
\end{align*}
\item[(iii)] 
{\em Bounds in Phase (ii):} 
For $\alpha \in [\frac12, \frac34]$ and $i \in [-L, 0]$,
\begin{align*}
\epsilon_i(\alpha) 
\geq \epsilon(x_0^*) \frac{x_{-L}^*}{x_0^*},
\end{align*}
where $\epsilon(x)=\frac{x}{(1-(1-x)^{\dr-1})^{\dl-1}}$.

\item[(iv)] 
{\em Bounds in Phase (iii):}
Let
\begin{align}\label{equ:gamma}
\gamma = &
(\frac{(\dr-1)(\dl-1)(\epsilon^*)^{\frac1{\dl-1}}(1+w^{1/8})}{w})^{\dl-1}.
\end{align}
Let $\alpha \in [\frac14, \frac12]$.
For $x_i(\alpha) > \gamma$, 
\begin{align*}
\epsilon_i(\alpha) & 
\begin{cases}
\leq \epsilon^* \big(1+\frac{1}{w^{1/8}}\big), & i \in [-L+w-1, -w+1], \\ 
\geq \epsilon^*\Big(1 - \frac1{1+w^{1/8}}\Big),   & i \in [-L, 0].
\end{cases}
\end{align*} 

For $x_i(\alpha) \leq \gamma$ and $w > \max\{2^4 \dl^2 \dr^2, 2^{16}\}$,
\begin{align*}
\epsilon_i(\alpha) & 
\geq \epsilon^*\Big(1 - \frac4{w^{1/8}}\Big)^{(\dr-2)(\dl-1)},  i \in [-L, 0]. 
\end{align*} 

\item[(v)] {\em Area under EXIT Curve:}
The EXIT value at position $i \in [-L, L]$
is defined by
\begin{align*}
h_i(\alpha) = (g(x_{i-w+1}(\alpha), \dots, x_{i+w-1}(\alpha)))^{\frac{\dl}{\dl-1}}.
\end{align*}
Let
\begin{align*}
A(\dl, \dr, w, L)=\int_{0}^{1}\frac1{2L+1}\sum_{i=-L}^{L} h_i(\alpha)d\epsilon_i(\alpha),
\end{align*}
denote the area of the EXIT integral. 
Then 
$$
\vert A(\dl, \dr, w, L) - (1-\frac{\dl}{\dr})  \vert \leq  \frac{w}{L} \dl \dr.
$$

\item[(vi)] {\em Bound on $\epsilon^*$:}
For $w > \max\{2^4 \dl^2 \dr^2, 2^{16}\}$,
\begin{align*}
\vert \epsilon^{\MAPsmall}(\dl, \dr) - \epsilon^* \vert \leq 
\frac{2\dl\dr\vert x_0^* -
\xstab(\epsilon^*)\vert+c(\dl, \dr, w,
L)}{(1-(\dl-1)^{-\frac1{\dr-2}})^2}
\end{align*}
where
\begin{align*}
c(\dl,& \dr, w, L) =  
4\dl \dr w^{-\frac18} + \frac{w\dl(2+\dr)}{L} \\
& + \dl \dr (x^*_{-L'+L}+x^*_0-x^*_{-L})  + \frac{2\dr\dl^2}{(1 \!-\! 4w^{-\frac{1}8})^{\dr}} w^{-\frac78}.
\end{align*}
\end{itemize}
\end{theorem}

\subsection{Step (iii): Operational Meaning of EXIT Curve}
\blemma[Stability of $\{(\underline{\epsilon}(\alpha), \x(\alpha))\}_{\alpha=0}^{1}$]\label{lem:stability}
Let $\{(\underline{\epsilon}(\alpha), \x(\alpha))\}_{\alpha=0}^{1}$
denote the EXIT curve constructed in Definition~\ref{def:EXIT}.  For $\beta \in (0, 1)$, let
\begin{align*}
\epsilon^{(\beta)} & = \inf_{\beta \leq \alpha \leq 1} \{ \epsilon_i(\alpha): i \in [-L, L]\}.
\end{align*}
Consider forward DE (cf. Definition~\ref{def:forwardDE}) with parameter $\epsilon$, $\epsilon <
\epsilon^{(\beta)}$. Then the sequence $\x^{(\ell)}$ (indexed from $-L$ to $L$) converges to a FP
which is point-wise upper bounded by $\x(\beta)$.
\elemma
\begin{proof}
Recall from Lemma~\ref{lem:forwardDE} that the sequence $\x^{(\ell)}$
converges to a FP of DE, call it $\x^{(\infty)}$. We claim that
$\x^{(\infty)} \leq \x(\beta)$.

We proceed by contradiction. Assume that $\x^{(\infty)}$ is not point-wise
dominated by $\x(\beta)$.
Recall that by construction of $\x(\alpha)$ the components are
decreasing in $\alpha$ and that they are continuous.  
Further, $\x^{(\infty)} \leq \epsilon < \x(1)$.
Therefore, 
\begin{align*}
\gamma = \inf_{\beta \leq \alpha \leq 1}\{\alpha \mid \x^{(\infty)} \leq \x(\alpha)\}
\end{align*}
is well defined. By assumption $\gamma>\beta$. Note that there must
exist at least one position $i \in [-L, 0]$ so that
$x_i(\gamma)=x_i^{(\infty)}$\footnote{It is not hard to show that under
forward DE, the constellation $\x^{(\ell)}$ is unimodal and symmetric around
$0$. This immediately follows from an inductive argument using
Definition~\ref{def:forwardDE}.}. But since $\epsilon < \epsilon_i(\gamma)$
and since $g(\dots)$ is monotone in its components,  
\begin{align*}
x_i(\gamma) 
& = \epsilon_i(\gamma) g(x_{i-w+1}(\gamma), \dots, x_{i+w-1}(\gamma)) \\
& > \epsilon g(x_{i-w+1}^{(\infty)}, \dots, x_{i+w-1}^{(\infty)}) = x^{(\infty)}_i,
\end{align*}
a contradiction.
\end{proof}

\subsection{Step (iv): Putting it all Together}
We have now all the necessary ingredients to prove
Theorem~\ref{the:main}. In fact, the only statement that needs proof is
(\ref{equ:epslowerbound}). First note that $\epsilon^{\BPsmall}(\dl,
\dr, L, w)$ is a non-increasing function in $L$. This follows by
comparing DE for two constellations, one, say, of length $L_1$
and one of length $L_2$, $L_2 > L_1$.  It therefore suffices to prove
(\ref{equ:epslowerbound}) for the limit of $L$ tending to infinity.

Let $(\dl, \dr, w)$ be fixed with $w>w(\dl, \dr)$, where 
$$w(\dl, \dr)=\max\Big\{2^{16}, 2^4\dl^2\dr^2, \frac{(2\dl\dr(1+\frac{2\dl}{1\!-\!2^{\!-\!1\!/\!(\dr\!-\!2)}}))^8}{(1\!-\!2^{\!-1\!/\!(\dr\!-\!2)})^{16}(\frac12(1\!-\!\frac{\dl}{\dr}))^8}\Big\}.$$ 
Our strategy is as follows. We pick $L'$ (length of constellation)
sufficiently large (we will soon see what ``sufficiently'' means)
and choose an entropy, call it $\hat{\chi}$. Then we apply
Theorem~\ref{thm:onesidedDE}. Throughout this section, we will use $\x^*$ and $\epsilon^*$
to denote the FP and the corresponding channel parameter  
guaranteed by Theorem~\ref{thm:onesidedDE}. We are faced with two possible scenarios.
Either there exists a {\em FP with the desired properties} or there
exists a {\em FP with parameter $\epsilon^*=1$} and entropy at most
$\hat{\chi}$. We will then show (using Theorem~\ref{the:propertiesEXIT})
that for sufficiently large $L'$ the second alternative is not possible.
As a consequence, we will have shown the existence of a FP with the
desired properties. Using again Theorem~\ref{the:propertiesEXIT}
we then show that $\epsilon^*$ is close to
$\epsilon^{\MAPsmall}$ and that $\epsilon^*$ is a lower bound for the
BP threshold of the coupled code ensemble.

Let us make this program precise.  Pick $\hat{\chi} =
\frac{\xunstab(1)+x^{\BPsmall}(\dl, \dr)}{2}$ and $L'$ ``large''.
In many of the subsequent steps we require specific lower bounds on $L'$.
Our final choice is one which obeys all these lower bounds.
Apply Theorem~\ref{thm:onesidedDE} with parameters $L'$ and
$\hat{\chi}$.  We are faced with two alternatives.

Consider first the possibility that the constructed one-sided FP $\x^*$
has parameter $\epsilon^*=1$ and entropy bounded by
\begin{align*}
\frac{(1-\frac{\dl}{\dr})(x^{\BPsmall} - \xunstab(1))}{16}-
\frac{\dl w}{2 \dr (L'+1)} 
\leq \chi(\x^*) \leq  \frac{x^{\BPsmall}\!+\!\xunstab(1)}{2}. 
\end{align*}
For sufficiently large $L'$ this can be simplified to
\begin{align}\label{equ:entropybnd}
\frac{(1-\frac{\dl}{\dr})(x^{\BPsmall} - \xunstab(1))}{32}
\leq \chi(\x^*) \leq  \frac{x^{\BPsmall}\!+\!\xunstab(1)}{2}. 
\end{align}
Let us now construct an EXIT curve based on $(\epsilon^*, \x^*)$
for a system of length $L$, $1 \leq L < L'$.
According to Theorem~\ref{the:propertiesEXIT}, it must be true that
\begin{align}\label{equ:epsstarbound}
 \epsilon^*& \leq  
 \epsilon^{\MAPsmall}(\dl, \dr) + \frac{2\dl\dr\vert x_0^* -
\xstab(\epsilon^*)\vert+c(\dl, \dr, w,
L)}{(1-(\dl-1)^{-\frac1{\dr-2}})^2}.
\end{align}
We claim that by choosing $L'$ sufficiently large and by choosing $L$ appropriately
we can guarantee that 
\begin{align}\label{equ:choiceofL}
\vert x_0^* - \xstab(\epsilon^*)\vert \leq \delta,\, \vert x_0^* - x^*_{-L} \vert \leq \delta, \, x_{-L'+L}^*\leq \delta,
\end{align}
where $\delta$ is any strictly positive number. If we assume this claim
for a moment, then we see that the right-hand-side of (\ref{equ:epsstarbound}) can be made
 strictly less than 1. Indeed, this follows from $w>w(\dl, \dr)$ (hypothesis of the theorem) by choosing $\delta$
sufficiently small (by making $L'$ large enough) and by choosing $L$ to be proportional to $L'$ (we will see how this is done in the sequel). This is a contradiction, since by assumption $\epsilon^*=1$.
This will show that the second alternative must apply.  

Let us now prove the bounds in (\ref{equ:choiceofL}).  In the sequel
we say that sections with values in the interval $[0, \delta]$ are
part of the {\em tail}, that sections with values in $[\delta,
\xstab(\epsilon^*)-\delta]$ form the {\em transition}, and that
sections with values in $[\xstab(\epsilon^*)-\delta, \xstab(\epsilon^*)]$
represent the {\em flat} part. Recall from Definition~\ref{def:entropy} that the
entropy of a constellation is the average (over all the $2L+1$ sections) erasure
fraction. The bounds in  (\ref{equ:choiceofL}) are equivalent to
saying that {\em both} the tail as well as the flat part must  
have length at least $L$.  From Lemma~\ref{lem:transitionlength},
for sufficiently small $\delta$, the transition has length at most
$\frac{wc(\dl, \dr)}{\delta}$ (i.e., the number of sections $i$ with erasure
value, $x_i$, in the interval $[\delta, \xstab(\epsilon^*)-\delta]$), a constant independent of $L'$.
 Informally, therefore, most of the length $L'$ consists of the tail or the flat
part. 

Let us now show all this more precisely.
First, we show that the flat part is large, i.e., it is at least
a fixed fraction of $L'$. We argue as follows.  Since the transition
contains only a constant number of sections, its contribution to the
entropy is small. More precisely, this contribution is upper bounded by $\frac{wc(\dl,
\dr)}{(L'+1)\delta}$. Further, the contribution to the entropy from
the tail is small as well, namely at most $\delta$.  Hence, the
total contribution to the entropy stemming from the tail plus the
transition is at most $\frac{wc(\dl, \dr)}{(L'+1)\delta}+\delta$.
However, the entropy of the FP is equal to
$\frac{x^{\BPsmall}+\xunstab(1)}{2}$. As a consequence, the flat
part must have length which is at least a fraction
$\frac{x^{\BPsmall}+\xunstab(1)}{2}-\frac{wc(\dl,
\dr)}{(L'+1)\delta}-\delta$ of $L'$. This fraction is strictly positive
if we choose $\delta$ small enough and $L'$ large enough.

By a similar argument we can show that the tail length is also a strictly
positive fraction of $L'$. From Lemma~\ref{lem:maximum},
$\xstab(\epsilon^*) > x^{\BPsmall}$.  Hence the flat part cannot
be too large since the entropy is equal to
$\frac{x^{\BPsmall}+\xunstab(1)}{2}$, which is strictly smaller
than $x^{\BPsmall}$. As a consequence, the tail has length 
at least a  fraction
$1-\frac{x^{\BPsmall}+\xunstab(1)}{2(x^{\BPsmall}-\delta)}-\frac{1+\frac{wc(\dl,\dr)}{\delta}}{L'+1}
$ of $L'$. As before, this fraction is also strictly positive if
we choose $\delta$ small enough and $L'$ large enough. Hence, by
choosing $L$ to be the lesser of the length of the flat part and
the tail, we conclude that the bounds in \eqref{equ:choiceofL} are
valid and that $L$ can be chosen arbitrarily large (by increasing $L'$).

Consider now the second case. In this case $\x^*$ is a proper one-sided FP  with
entropy equal to $\frac{x^{\BPsmall}+\xunstab(1)}{2}$ and
with parameter $\epsilon^{\BPsmall}(\dl, \dr) < \epsilon^* < 1$.  Now,
using again Theorem~\ref{the:propertiesEXIT}, we can show 
\begin{align*}
\epsilon^* & > \epsilon^{\MAPsmall}(\dl, \dr) \!-\! 2w^{-\frac18}
 \frac{4\dl\dr  + \frac{2\dr\dl^2}{(1-4w^{-\frac18})^{\dr}}}{(1\!-\!(\dl\!-\!1)^{-\frac{1}{\dr-1}})^2} \\
& \stackrel{\dl \geq 3}{\geq}  \epsilon^{\MAPsmall}(\dl, \dr) \!-\! 
 2w^{-\frac18}\frac{4\dl\dr  + \frac{2\dr\dl^2}{(1-4w^{-\frac18})^{\dr}}}{(1\!-\!2^{-\frac{1}{\dr}})^2}.
\end{align*}
To obtain the above expression, we take $L'$ to be sufficiently large in order
to bound the term in $c(\dl, \dr, w, L)$ which contains $L$. 
We also use \eqref{equ:choiceofL} and  choose $\delta$ to be sufficiently small to bound the corresponding terms.
 We also
replace $w^{-7/8}$ by $w^{-1/8}$ in $c(\dl, \dr, w, L)$.

To summarize: we conclude that for an entropy equal to
$\frac{x^{\BPsmall}(\dl, \dr)+\xunstab(1)}{2}$, for sufficiently
large $L'$, $\x^*$ must be a proper one-sided FP with parameter
$\epsilon^* $ bounded as above.

Finally, let us show that $\epsilon^*\big(1-\frac4{w^{1/8}}\big)^{\dr\dl}$
is a lower bound on the BP threshold. We start by claiming that
\begin{align*}
\epsilon^*\Big(1-\frac4{w^{1/8}}\Big)^{\dr\dl} & < 
\epsilon^*\Big(1-\frac4{w^{1/8}}\Big)^{(\dr-2)(\dl-1)}  \\ 
 & = \inf_{\frac14\leq \alpha\leq 1}\{\epsilon_i(\alpha): i \in [-L, L]\}.
\end{align*}
To prove the above claim we just need to check that
$\epsilon(x_0^*)x_{-L}^*/x_0^*$ (see bounds in phase (ii) of
Theorem~\ref{the:propertiesEXIT}) is greater than the above infimum. Since
in the limit of $L'\to \infty$, $\epsilon(x_0^*)x_{-L}^*/x_0^*
\to \epsilon^*$, for sufficiently large $L'$ the claim is true.

From the hypothesis of the theorem we have $w>2^{16}$.  Hence
$\epsilon^*(1-4w^{-1/8})^{\dr\dl}>0$. Apply forward DE (cf. Definition~\ref{def:forwardDE}) with parameter
$\epsilon<\epsilon^*(1-4w^{-1/8})^{\dr\dl}$ and length $L$. Denote the FP by
$\x^{\infty}$ (with indices belonging to $[-L,L]$). 
From Lemma~\ref{lem:stability} we then conclude that $\x^{\infty}$
 is point-wise upper bounded by $\x(\frac14)$. 
But for
$\alpha=1/4$ we have  
\begin{align*}
x_i(1/4)\leq x_0(1/4) = x^*_{-L'+L} \leq \delta < \xunstab(1) \quad \forall \,i,
\end{align*}
where we make use of the fact that $\delta$ can be chosen arbitrarily small. 
Thus $x_i^{(\infty)} < \xunstab(1)$ for all $i\in [-L ,L]$. Consider
a one-sided constellation, $\y$, with $y_i=x_0(1/4)<\xunstab(1)$ for all $i\in [-L, 0]$. 
Recall that for a one-sided constellation $y_i=y_0$ for all $i>0$ and as usual $y_i=0$ for
$i<-L$. Clearly, $\x^{(\infty)}\leq \y$. Now apply one-sided forward DE
to $\y$ with parameter $\epsilon$ (same as the one we applied to get $\x^{\infty}$) 
and call it's limit $\y^{(\infty)}$.
From part (i) of Lemma~\ref{lem:nontrivialonesided} we conclude that
the limit $\y^{(\infty)}$ is either proper or trivial. Suppose that
$\y^{\infty}$ is proper (implies non-trivial). Clearly, $y_i^{\infty}<\xunstab(1)$ for
all $i\in [-L, 0]$.  But from Lemma~\ref{lem:maximum} we have that for
any proper one-sided FP $y_0 \geq \xunstab(\epsilon) \geq \xunstab(1)$,
a contradiction.  Hence we conclude that $\y^{\infty}$ must be trivial
and so must be $\x^{\infty}$.

\section{Discussion and Possible Extensions}
\subsection{New Paradigm for Code Design}\label{sec:optimization}
The explanation of why convolutional-like LDPC ensembles perform so
well given in this paper gives rise to a new paradigm in
code design.

In most designs of codes based on graphs one encounters a
trade-off between the threshold and the error floor behavior.  E.g.,
for standard irregular graphs an optimization of the threshold tends
to push up the number of degree-two variable nodes. The same quantity,
on the other hand, favors the existence of low weight (pseudo)codewords.

For convolutional-like LDPC ensembles the important operational quantity
is the MAP threshold of the underlying ensemble. As, e.g., regular LDPC
ensembles show, it is simple to improve the MAP threshold {\em and}
to improve the error-floor performance -- just increase the minimum
variable-node degree. From this perspective one should simply pick as
large a variable-node degree as possible.

There are some drawbacks to picking large degrees. First, picking large degrees also
increases the complexity of the scheme. Second, although currently little 
is known about the scaling behavior of the convolutional-like LDPC ensembles,
it is likely that large degrees imply a slowing down of the convergence
of the performance of finite-length ensembles to the asymptotic limit.
This implies that one has to use large block lengths.  Third, the larger
we pick the variable-node degrees the higher the implied rate loss. Again,
this implies that we need very long codes in order
to bring down the rate loss to acceptable levels. It is tempting to
conjecture that the minimum rate loss that is required in order to achieve
the change of thresholds is related to the area under the EXIT curve
between the MAP and the BP threshold.  E.g., in Figure~\ref{fig:lrLexit}
this is the light gray area.  For the underlying ensemble this is exactly
the amount of guessing (help) that is needed so that a local algorithm
can decode correctly, assuming that the underlying channel parameter is
the MAP threshold.

Due to the above reasons, an actual code design will therefore 
try to maintain relatively small average degrees so as to keep this gray
area small. But the additional degree of freedom can be used to design
codes with good thresholds {\em and} good error floors.

\subsection{Scaling Behavior}
In our design there are three parameters that tend to infinity.
The number of variables nodes at each position, called $M$, the length of the constellation $L$, and
the length of the smoothing window $w$. Assume we fix $w$ and we are
content with achieving a threshold slightly below the MAP threshold.
How should we scale $M$ with respect to $L$ so that we achieve the
best performance? This question is of considerable practical
importance.  Recall that the total length of the code is of order $L
\cdot M$. We would therefore like to keep this product small. Further,
the rate loss is of order $1/L$ (so $L$ should be large) and $M$ should
be chosen large so as to approach the performance predicted by DE.
Finally, how does the number of required iterations scale as a function
of $L$?

Also, in the proof we assumed that we fix $L$ and let $M$ tend to infinity
so that we can use DE techniques. We have seen that in this limit the
boundary conditions of the system dictate the performance of the
system regardless of the size of $L$ (as long as $L$ is fixed and $M$
tends to infinity). Is the same behavior still true if we let $L$ tend
to infinity as a function of $M$? At what scaling does the behavior change?

\subsection{Tightening of Proof}
As mentioned already in the introduction, our proof is weak -- it
promises that the BP threshold approaches the MAP threshold of the
underling ensemble at a speed of $w^{-1/8}$. Numerical experiments
indicate that the actual convergence speed is likely to be exponential
and that the prefactors are very small. Why is the analytic statement
so loose and how can it be improved?

Within our framework it is clear that at many places the constants could
be improved at the cost of a more involved proof. It is therefore likely
that a more careful analysis following the same steps will give
improved convergence speeds. 

More importantly, for mathematical convenience we constructed an
``artificial'' EXIT curve by interpolating a particular fixed point
and we allowed the channel parameter to vary as a
function of the position. In the proof we then coarsely bounded the
``operational'' channel parameter by the minimum of all the individual
channel parameters. This is a significant source for the looseness
of the bound. A much tighter bound could be given if it were
possible to construct the EXIT curve by direct methods. As we have seen,
it is possible to show the existence of FPs of DE for a wide range of
EXIT values. The difficulty consists in showing that all these individual
FPs form a smooth one-dimensional manifold so that one can use the Area
Theorem and integrate with respect to this curve.

\subsection{Extensions to BMS Channels and General Ensembles}
Preliminary numerical evidence suggests that the behavior of the 
convolutional-like LDPC ensembles discussed in this paper is not
restricted to the BEC channel or to regular ensembles but is a general
phenomenon. We will be brief. A more detailed discussion can be found in the two recent papers
\cite{LMFC10,KMRU10}. Let us quickly discuss how one might want to attack 
the more general setup.

We have seen that the proof consists essentially of three steps.
\begin{itemize}
\item[(i)]{\em Existence of FP:} 
As long as we stay with the BEC, a similar procedure as the one used
in the proof of Theorem~\ref{thm:onesidedDE} can be used to show the
existence of the desired FP for more general {\em ensembles}.

General BMS {\em channels} are more difficult to handle, but FP theorems
do exist also in the setting of infinite-dimensional spaces. The most
challenging aspect of this step is to prove that the constructed FP
has the essential basic characteristics that we relied upon for our
later steps. In particular, we need it to be unimodal, to have a short
transition period, and to approach the FP density of the underlying
standard ensemble.

\item[(ii)] {\em Construction of EXIT Curve and Bounds:} 
Recall that in order to create a whole EXIT curve, we started with a FP
and interpolated the value of neighboring points. In order to ensure
that each such interpolated constellation is indeed a FP, we allowed
the local channel parameters to vary.  By choosing the interpolation properly, we were
then able to show that this variation is small.  As long as one remains
in the realm of BEC channels, the same technique can in principle be
applied to other ensembles. For general channels the construction seems
more challenging. It is not true in general that, given a constellation,
one can always find ``local'' channels that make this constellation a FP.
It is therefore not clear how an interpolation for general channels can
be accomplished. This is perhaps the most challenging hurdle for any
potential generalization.

\item[(iii)] {\em Operational Interpretation:}  
For the operational interpretation we relied upon the notion of physical
degradation. We showed that, starting with a channel parameter of a
channel which is upgraded w.r.t. to any of the local channels used in the
construction of the EXIT curve, we do not get stuck in a non-trivial
FP. For the BEC, the notion of degradation is very simple, it is the
natural order on the set of erasure probabilities, and this is a total
order. For general channels, an order on channels still exists in terms
of degradation, but this order is partial. We therefore require that
the local channels used in the construction of the EXIT curve are all degraded
w.r.t. a channel of the original channel family (e.g., the family of Gaussian channels)
with a parameter which is only slightly better than the parameter which
corresponds to the MAP threshold.
\end{itemize}

\subsection{Extension to General Coupled Graphical Systems}
Codes based on graphs are just one instance of graphical systems that
have distinct thresholds for ``local'' algorithms (what we called the
BP threshold) and for ``optimal'' algorithms (what we called the MAP
threshold). To be sure, coding is somewhat special -- it is conjectured
that the so-called replica-symmetric solution always determines the
threshold under MAP processing for codes based on graphs. Nevertheless,
it is interesting to investigate to what extent the coupling of general
graphical systems shows a similar behavior.  Is there a general class
of graphical models in which the same phenomenon occurs? If so, can
this phenomenon either be used to analyze systems or to devise better
algorithms? 

\section*{Acknowledgment}
We would like to thank N. Macris for his help in choosing the title
and sharing his insights and the reviewers for their thorough reading and numerous suggestions.
We would also like to thank D. J. Costello, Jr., P. Vontobel,
and A. R. Iyengar for their many comments and very helpful feedback
on an earlier draft.  Last but not least we would like to thank G.
D. Forney, Jr.  for handling our paper. The work of S. Kudekar
was supported by the grant from the Swiss National Foundation no
200020-113412.

\begin{appendices}
\section{Proof of Lemma~\ref{lem:lrLweight}\label{app:lrLweight}}
We proceed as follows. We first consider a ``circular'' ensemble.
This ensemble is defined in an identical manner as the $(\dl, \dr,
L)$ ensemble except that the  positions are now from $0$ to $K-1$
and index arithmetic is performed modulo $K$. This circular definition
symmetrizes all positions, which in turn simplifies calculations.

As we will see shortly, most codes in this circular ensemble have a minimum stopping
set distance which is a linear fraction of $M$.
To make contact with our original problem we now argue as follows.
Set $K=2 L+\dl$. If, for the circular ensemble, we take $\dl-1$ consecutive
positions and set them to $0$ then this ``shortened'' ensemble has length
$2 L +1$ and it is in one-to-one correspondence with the $(\dl, \dr, L)$ ensemble. Clearly,
no new stopping sets are introduced by shortening the ensemble. This
proves the claim.

Let $A(\dl, \dr, M, K, w)$ denote the expected number of stopping sets of
weight $w$ of the  ``circular'' ensemble. Let ${\mathcal C}$ denote a
code chosen uniformly at random from this ensemble.

Recall that every variable node at position $i$ connects to a check
node at positions $i-\dlh, \dots, i+\dlh$, modulo $K$. There are $M$
variable nodes at each position and $M \frac{\dl}{\dr}$ check nodes at
each position.  Conversely, the $M \dl$ edges entering the check nodes
at position $i$ come equally from variable nodes at position
$i-\dlh, \dots, i+\dlh$. These $M \dl$ edges are connected to the
check nodes via a random permutation.

Let $w_k$, $k \in \{0, \dots, K-1\}$, $0 \leq w_k \leq M$, denote the
{\em weight at position $i$}, i.e., the number of variable nodes at
position $i$ that have been set to $1$. Call $\ww=(w_0, \dots, w_{K-1})$
the {\em type}. We are interested in the expected number of stopping
sets for a particular type; call this quantity $A(\dl, \dr, M, K, \ww)$.
Since the parameters $(\dl, \dr, M, K)$ are understood from the context,
we shorten the notation to $A(\ww)$.  We claim that
\begin{align}
 &A(\ww) = 
\frac{\prod_{k=0}^{K-1}\binom{M}{w_k}\text{coef}\{p(x)^{M \frac{\dl}{\dr}}, 
x^{\sum_{i=-\dlh}^{\dlh} w_{k+i}} \}}{\prod_{k=0}^{K-1} \binom{M \dl}{\sum_{i=-\dlh}^{\dlh} w_{k+i}}} \nonumber \\
\stackrel{(a)}{\leq} & 
\prod_{k=0}^{K-1} \frac{(M\!+\!1)\binom{M}{\frac{\sum_{i=-\dlh}^{\dlh} w_{k+i}}{\dl}}\text{coef}\{p(x)^{M \frac{\dl}{\dr}}, 
x^{\sum_{i=-\dlh}^{\dlh} w_{k+i}} \}}{\binom{M \dl}{\sum_{i=-\dlh}^{\dlh} w_{k+i}}}. 
\label{equ:aupperbound}
\end{align}
where 
$p(x) = \sum_{i \neq 1} \binom{\dr}{i} x^i$.
This expression is easily explained. The $w_k$ variable nodes at 
position $k$ that are set to $1$ can be distributed over the $M$ variable nodes
in $\binom{M}{w_k}$ ways. Next, 
we have to distribute the 
$\sum_{i=-\dlh}^{\dlh} w_{k+i}$ ones among the
$M \frac{\dl}{\dr}$ check nodes in such a way
that every check node is fulfilled (since we are looking for
stopping sets, ``fulfilled'' means that a check node is either connected
to no variable node with associated value ``1'' or to {\em at least} two such nodes). This is encoded
by $\text{coef}\{p(x)^{M \frac{\dl}{\dr}}, x^{\sum_{i=-\dlh}^{\dlh} w_{k+i}}\}$.
Finally, we have to divide by the total number of possible connections;
there are $M \dl$ check node sockets at position $k$ and we 
distribute $\sum_{i=-\dlh}^{\dlh} w_{k+i}$ ones. This can be done in
$\binom{M \dl}{\sum_{i=-\dlh}^{\dlh} w_{k+i}}$ ways.
To justify step (a) note that 
\begin{align*}
\prod_{i=-\dlh}^{\dlh} \binom{M}{w_{k+i}}^{\frac1{\dl}} 
& \leq 2^{M \frac1{\dl} \sum_{i=-\dlh}^{\dlh} h(\frac{w_{k+i}}{M})} \\
& \stackrel{\text{Jensen}}{\leq}
2^{M h(\frac1{\dl} \sum_{i=-\dlh}^{\dlh} \frac{w_{k+i}}{M})} \\
& \leq (M+1) \binom{M}{\frac1{\dl} \sum_{i=-\dlh}^{\dlh} w_{k+i}}. 
\end{align*}
Note that, besides the factor $(M+1)$, which is negligible, 
each term in the product (\ref{equ:aupperbound}) has the exact form of the average stopping set weight
distribution of the standard $(\dl, \dr)$-ensemble of length $M$ and weight 
$\frac1{\dl}\sum_{i=-\dlh}^{\dlh} w_k$. (Potentially this weight is non-integral
but the expression is nevertheless well defined.)

We can therefore leverage known results concerning the stopping set
weight distribution for the underlying $(\dl, \dr)$-regular ensembles.
For the $(\dl, \dr)$-regular ensembles we know that the relative minimum distance is at least
$\hat{\omega}(\dl, \dr)$ with high probability \cite[Lemma D.17]{RiU08}.
Therefore, as long as $\frac1{\dl M} \sum_{i=-\dlh}^{\dlh} w_{k+i}
< \hat{\omega}(\dl, \dr)$, for all $0 \leq k <K$, $\frac{1}{M K}
\log A(\ww)$ is strictly negative and so most codes in the ensemble
do not have stopping sets of this type.  The claim now follows since in order for
the condition $\frac1{\dl M} \sum_{i=-\dlh}^{\dlh} w_{k+i} <
\hat{\omega}(\dl, \dr)$ to be violated for at least one position
$k$ we need $\frac1{M} \sum_{k=0}^{K-1} w_k$ to exceed $\dl
\hat{\omega}(\dl, \dr)$.

\section{Basic Properties of $h(x)$ \label{app:propertyofh(x)}}
Recall the definition of $h(x)$ from \eqref{equ:hfunction}. We have,
\begin{lemma}[Basic Properties of $h(x)$]\label{lem:propertyofh(x)}
Consider the $(\dl, \dr)$-regular ensemble with $\dl \geq 3$
and let $\epsilon \in (\epsilon^{\BPsmall}, 1]$. 
\begin{itemize}
\item[(i)] $h'(\xunstable) > 0$ and $h'(\xstable) < 0$; 
$|h'(x)| \leq \dl \dr$ for $x \in [0, 1]$.
\item[(ii)]
There exists a unique value $0\leq x_*(\epsilon) \leq \xunstable$ so
that $h'(x_*(\epsilon)) = 0$, and there exists a unique value $\xunstable
\leq x^*(\epsilon) \leq \xstable$ so that $h'(x^*(\epsilon))=0$.

\item[(iii)]
Let 
\begin{align*}
& \kappa_*(\epsilon) = \min\{ -h'(0), \frac{-h(x_*(\epsilon))}{x_*(\epsilon)} \}, \\
& \lambda_*(\epsilon) = \min\{ h'(\xunstable), \frac{-h(x_*(\epsilon))}{\xunstable - x_*(\epsilon)} \}, \\
& \kappa^*(\epsilon) = \min\{h'(\xunstable), \frac{h(x^*(\epsilon))}{x^*(\epsilon)-\xunstable}\}, \\
& \lambda^*(\epsilon) = \min\{ -h'(\xstable), \frac{h(x^*(\epsilon))}{\xstable - x^*(\epsilon)}\}. 
\end{align*}
The quantities $\kappa_*(\epsilon), \lambda_*(\epsilon),
\kappa^*(\epsilon)$, and $\lambda^*(\epsilon)$ are non-negative and
depend only on the channel parameter $\epsilon$ and the degrees $(\dl,
\dr)$. In addition, $\kappa_*(\epsilon)$ is strictly positive for
all $\epsilon \in [0, 1]$.

\item[(iv)]
For $0 \leq \epsilon \leq 1$,
\begin{align*}
x_*(\epsilon) > \frac{1}{\dl^2 \dr^2}.
\end{align*}

\item[(v)]
For $0 \leq \epsilon \leq 1$,
\begin{align*}
\kappa_*(\epsilon) \geq \frac1{8\dr^2}.
\end{align*}

\item[(vi)]
If we draw a line from $0$ with slope $-\kappa_*$, then $h(x)$
lies below this line for $x \in [0 , x_*]$.

If we draw a line from $\xunstable$ with slope $\lambda_*$, then
$h(x)$ lies below this line for all $x \in [x_*, \xunstable]$.

If we draw a line from $\xunstable$ with slope $\kappa^*$, then $h(x)$
lies above this line for $x \in [\xunstable , x^*]$.

Finally, if we draw a line from $\xstable$ with slope $-\lambda^*$, then
$h(x)$ lies above this line for all $x \in [x^*, \xstable]$.
\end{itemize}
\end{lemma}

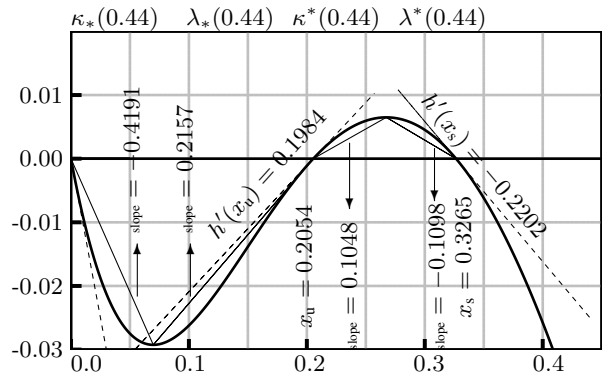
\begin{figure}[htp]
\begin{centering}
\input{ps/de36}
\caption{
Pictorial representation of the various quantities which appear in
Lemma~\ref{lem:propertyofh(x)}. We use the $(3,6)$ ensemble to transmit
over a BEC with erasure probability  $\epsilon = 0.44$.
The function $h(x) = 0.44(1-(1-x)^{5})^{2} - x$ is represented in
the figure by the smooth bold curve. 
The roots of
$h(x)=0$ or, equivalently, the FPs of DE are given by $0$,
$\xunstab(0.44) \approx 0.2054$, and $\xstab(0.44)
\approx 0.3265$. There are only two stationary points of $h(x)$,
i.e., only two points at which $h'(x)=0$. They are given by $x_*(0.44) \approx
0.0697$ and $x^*(0.44) \approx 0.2673$. Along with the curve $h(x)$, the figure
contains three dashed lines representing the tangents at the points $0$,
$\xunstab(0.44)$ and $\xstab(0.44)$. The slopes of the
tangents at $0$, $\xunstab(0.44)$ and $\xstab(0.44)$
are $h'(0) = -1$, $h'(\xunstab) = 0.1984$ and
$h'(\xstab) = -0.2202$, respectively.	Also shown 
are the four lines which bound $h(x)$ in the various regions. 
These lines are given (their end-points) by:
$\{(0,0), (x_*,h(x_*))\}$ , $\{(x_*,h(x_*)), (\xunstab(0.44),
0)\}$, $\{(\xunstab(0.44),0),(x^*,h(x^*)) \} $ and
$\{(x^*,h(x^*)),(\xstab(0.44),0)\}$ and have slopes $-0.4191$,
$0.2157$, $0.1048$ and $-0.1098$ respectively. Thus we have $\kappa^*(0.44)
= 0.1048$, $\lambda^*(0.44) = 0.1098$, $\kappa_*(0.44) = 0.4191$ and $\lambda_*(0.44)
=0.1984$.}
\label{fig:propertyofh(x)}
\end{centering}
\end{figure}

\begin{example}[$(3, 6)$-Ensemble]
Consider transmission using a code from the $(3,6)$
ensemble over a BEC with $\epsilon=0.44$. The fixed  point equation for
the BP decoder is given by
$$
x = 0.44 (1-(1-x)^5)^2.
$$
The function $h(x) = 0.44(1-(1-x)^5)^2 - x$
is shown in Figure~\ref{fig:propertyofh(x)}.  The equation $h(x)=0$
has exactly 3 real roots, namely, $0$, $\xunstab(0.44) \approx
0.2054$ and $\xstab(0.44) \approx 0.3265$. Further properties
of $h(x)$ are shown in Figure~\ref{fig:propertyofh(x)}.
\end{example}
Let us prove each part separately. In order to lighten our
notation, we drop the $\epsilon$ dependence for quantities like
$\xunstab$, $\xstab$, $x_*$, or $x^*$.
\begin{itemize}
\item[(i)]
Note that $h(x) > 0$ for all $x\in (\xunstab, \xstab)$, with equality at
the two ends.  This implies that $h'(\xunstab) > 0$ and that $h'(\xstab)
< 0$. With respect to the derivative, we have
\begin{align*}
\hspace{-0.5cm} |h'(x)| &=  
 |\epsilon (\dl\!-\!1)(\dr\!-\!1) (1\!-\!x)^{\dr-2}(1\!-\!(1\!-\!x)^{\dr-1})^{\dl-2}\!-\!1| \\ 
& \leq  (\dl-1)(\dr-1)+1 \leq \dl \dr. 
\end{align*}

\item[(ii)] 
We claim that $h''(x)=0$ has exactly one real solution in $(0,1)$.
We have
\begin{align}\label{equ:second_deri}
& h''(x)  = \nonumber \\ & \epsilon(\dl-1)(\dr-1)(1-x)^{\dr-3}(1-(1-x)^{\dr-1})^{\dl-3} \nonumber \\
& \times \Big[ (1-x)^{\dr-1} (\dl\dr-\dl - \dr) - \dr+2\Big].
\end{align}
Thus $h''(x) = 0$ for $x\in (0,1)$ only at 
\begin{align}\label{equ:crit_second_deri}
x = 1 - \Big(\frac{\dr-2}{\dl \dr- \dl - \dr}\Big)^{\frac1{\dr-1}}.
\end{align}
Since $\dl\geq 3$, the above solution is in $(0, 1)$.

Since $h(0)=h(\xunstab)=h(\xstab)=0$, we know from Rolle's theorem
that there must exist an $0\leq x_*\leq \xunstab$ and an $\xunstab
\leq x^* \leq \xstab$, such that $h'(x_*) = h'(x^*) = 0$.

Now suppose that there exists a $y\in (0,1)$, $x_* \neq y \neq x^*$, such
that $h'(y) = 0$, so that $h'(\cdot)$ vanishes at three distinct places
in $(0, 1)$.  Then by Rolle's theorem we conclude that $h''(x)=0$
has at least two roots in the interval $(0,1)$, a contradiction.

\item[(iii)] 
To check that the various quantities in part (iii) are strictly positive,
it suffices to verify that $h(x_*)\neq 0$ and $h(x^*) \neq 0$. But we
know from Lemma \ref{lem:stableandunstable} that $h(x)=0$ has exactly
two solutions, namely $\xunstab$ and $\xstab$, and neither of them
is equal to $x^*$ or $x_*$ since $h'(\xunstab)>0$.

\item[(iv)]
From \eqref{equ:second_deri}, for all $x\in [0,1]$ we can 
upper bound $\vert h''(x) \vert$ by
\begin{align}\label{equ:sec_deri_bound}
(\dl\!-\!1)(\dr\!-\!1)[\dl\dr\!-\!\dl\!-\!\dr\!-\! \dr\!+\!2] 
< \dl^2 \dr^2.
\end{align}
Note that $h'(0)=-1$ and, by definition, $h'(x_*)=0$, so that
$\frac{1}{x_*}=\frac{h'(x_*) - h'(0)}{x_*-0}$.  Consider the function
$h'(x)$, $x \in [0, x_*]$.  From the continuity of the function
$h'(x)$ and, using the mean-value theorem, we conclude that there
exists an $\eta\in (0,x_*)$ such that $h''(\eta)=\frac{h'(x_*) - h'(0)
}{x_*}$. But from \eqref{equ:sec_deri_bound} we know that $h''(\eta)
< \dl^2 \dr^2$.  It follows that $\frac{1}{x_*} = \frac{h'(x_*) -
h'(0) }{x_*} < \dl^2 \dr^2$.

\item[(v)]
To get the universal lower bound on $\kappa_*(\epsilon)$ note that the
dominant (i.e., smaller) term in the definition of $\kappa_*(\epsilon)$
is $\frac{-h(x_*(\epsilon))}{x_*(\epsilon)}$.  (The second term,
$-h'(0)$, is $1$.)  Recall that $x_*$ is the point where $h(x)$ takes
on the minimum value in the range $[0, \xunstable]$. We can therefore
rewrite  $\kappa_*(\epsilon)$ in the form $\frac1{x_*} \max_{0 \leq x
\leq \xunstable} \{-h(x)\}$.  To get a lower bound on $\kappa_*(\epsilon)$
we use the trivial upper bound $x_*(\epsilon) \leq 1$.  It therefore
remains to lower bound $\max_{0 \leq x \leq \xunstable} \{-h(x)\}$.  Notice
that $-h(x)$ is a decreasing function of $\epsilon$ for every $x\in
[0,\xunstab(1)]$. Thus, inserting $\epsilon=1$, we get
\begin{align*}
& \max_{0 \leq x \leq \xunstab(1)} 
[x-(1-(1-x)^{\dr-1})^{\dl-1}] \\
= & \max_{0 \leq x \leq \xunstab(1)} 
[(x^{\frac{1}{\dl-1}})^{\dl-1}-(1-(1-x)^{\dr-1})^{\dl-1}] \\
    \geq  & \max_{0 \leq x \leq (\dr-1)^{-\frac{\dl-1}{\dl-2}}} 
	  (x^{\frac{1}{\dl-1}}-(\dr-1) x) x^{\frac{\dl-2}{\dl-1}} .
\end{align*}
Let us see how we derived the last inequality. First we claim that for $x\in
[0,(\dr-1)^{-\frac{\dl-1}{\dl-2}}]$ we have $x^{\frac{1}{\dl-1}}\geq (\dr-1)x
\geq  1-(1-x)^{\dr-1}$. Indeed, this can be easily seen by using the identity
$1-(1-x)^{\dr-1}=x(1+(1-x)+\dots+(1-x)^{\dr-2})$ and $x\leq 1$. Then we use  
$A^{\dl-1} - B^{\dl-1} = (A-B)(A^{\dl-2}+A^{\dl-3}B+\dots+B^{\dl-2}) \geq (A-B)A^{\dl-2}$
for all $0\leq B\leq A.$ Finally we use
\begin{align*}
(\xunstab(1))^{\frac1{\dl-1}} = (1 - ( 1 - \xunstab(1))^{\dr-1}) \leq (\dr-1)\xunstab(1),
\end{align*}
so that
\begin{align}\label{equ:boundonxuone}
\xunstab(1) \geq (\dr-1)^{-\frac{\dl-1}{\dl-2}}.
\end{align}
As a consequence $[0, (\dr-1)^{-\frac{\dl-1}{\dl-2}}] \subseteq [0,
\xunstab(1)]$ and hence we get the last inequality.  Now we can further
lower bound the right-hand-side above by evaluating it at any element
of $[0, (\dr-1)^{-\frac{\dl-1}{\dl-2}}]$.

We pick $\hat{x}=2^{-\frac{\dl-1}{\dl-2}} (\dr-1)^{-\frac{\dl-1}{\dl-2}}$.
Continuing the chain of inequalities we get \begin{align*}
	  & \stackrel{x=\hat{x}}{\geq}
	  (2^{\dl-1} (\dr-1))^{-\frac{1}{\dl-2}}
	  (\hat{x})^{\frac{\dl-2}{\dl-1}} \\ & = (2^{\dl-1}
	  (\dr-1))^{-\frac{1}{\dl-2}} (2^{-1}(\dr-1)^{-1}) \\ & =
	  \frac1{2^{\frac{2\dl-3}{\dl-2}}(\dr-1)^{\frac{\dl-1}{\dl-2}}}
	  \stackrel{(a)}{\geq} \frac1{8(\dr-1)^2} \geq \frac1{8\dr^2}.
\end{align*} Since $\dl\geq 3$ we have $\frac{2\dl-3}{\dl-2}\leq 3$ and
$\frac{\dl-1}{\dl-2}\leq 2$. Hence we obtain $(a)$.  

\item[(vi)] 
Let us prove that for all $x\in (\xunstab, x^*)$, $h(x)$  is strictly
above the line which contains the point $(\xunstab, 0)$ and has slope
$\kappa^* $. Denote this line by $l(x)$. More precisely, we have $l(x)
= \kappa^* (x - \xunstab)$.  Suppose to the contrary that there exists
a point $y\in (\xunstab, x^*)$ such that $h(y)<l(y)$. In this case we
claim that the equation $h(x) - l(x)=0$ must have at least $4$ roots.

This follows from (a) $h(\xunstab)=l(\xunstab)$, (b) $h'(\xunstab)
\geq l'(\xunstab)$, (c) $h(y) < l(y)$, (d) $h(x^*) \geq l(x^*)$, and,
finally, (e) $h(1)<l(1)$, where $\xunstab < y < x^* < 1$.  If all these
inequalities are strict then the $4$ roots are distinct.  Otherwise,
some roots will have higher multiplicities.

But if $h(x) - l(x)=0$ has at least $4$ roots then $h''(x)-l''(x)=0$
has at least $2$ roots. Note that $l''(x)=0$, since $l(x)$ is a linear
function. This leads to a contradiction, since, as discussed in part
(ii), $h''(x)$ has only one (single) root in $(0, 1)$.

The other cases can be proved along similar lines. 
\end{itemize}

\section{Proof of Lemma~\ref{lem:transitionlength}}\label{app:transitionlength}
We split the transition into several stages.  Generically, in each of
the ensuing arguments we consider a section with associated value just
above the lower bound of the corresponding interval. We then show that,
after a fixed number of further sections, the value must exceed the
upper bound of the corresponding interval. Depending on the length $L$
and the entropy of the constellation there might
not be sufficiently many sections left in the constellation to pass all
the way to $\xstable-\delta$.  In this case the conclusion of the lemma
is trivially fulfilled. Therefore, in the sequel, we can always assume
that there are sufficiently many points in the constellation.

In the sequel, $\kappa_*(\epsilon)$ and $x_*(\epsilon)$ are the specific
quantities for a particular $\epsilon$, whereas $\kappa_*$ and $x_*$
are the strictly positive universal bounds valid for all $\epsilon$,
discussed in Lemma~\ref{lem:propertyofh(x)}. We write $\kappa_*$ and $x_*$
instead of $\frac{1}{8 \dr^2}$ and $\frac{1}{\dl^2 \dr^2}$ to emphasize
their operational meaning.
\begin{itemize}
\item[(i)] 
{\em Let $\delta>0$. Then there are at most $w
(\frac{1}{\kappa_* \delta}+1)$ sections $i$ with
value $x_i$ in the interval $[\delta, x_*(\epsilon)]$.}

Let $i$ be the smallest index so that $x_i \geq \delta$. If $x_{i+(w-1)}
\geq x_*(\epsilon)$ then the claim is trivially fulfilled. Assume
therefore that $x_{i+(w-1)} \leq x_*(\epsilon)$.  Using the monotonicity
of $g(\cdot)$,
\begin{align*}
x_i & = \epsilon g(x_{i-(w-1)}, \dots, x_i, \dots, x_{i+(w-1)}) \\
& \leq \epsilon g(x_{i+(w-1)}, \dots, x_{i+(w-1)}).
\end{align*}
This implies
\begin{align*}
x_{i+(w-1)} - x_i  & \geq x_{i+(w-1)} - \epsilon g(x_{i+(w-1)},\dots, x_{i+(w-1)})  \\
& \stackrel{(\ref{equ:hfunction})}{=} -h(x_{i+(w-1)}) 
\stackrel{\text{Lemma~\ref{lem:propertyofh(x)} (vi)}}{\geq} -l(x_{i+(w-1)}) \\
& \geq -l(x_i) 
\geq -l(\delta) = \kappa_*(\epsilon) \delta.
\end{align*}
This is equivalent to
\begin{align*}
x_{i+(w-1)} \geq x_i + \kappa_*(\epsilon) \delta.
\end{align*}
More generally, using the same line of reasoning, 
\begin{align*}
x_{i+l(w-1)} \geq x_i + l \kappa_*(\epsilon) \delta,
\end{align*}
as long as $x_{i+l (w-1)} \leq x_*(\epsilon)$. 

We summarize. The total distance we have to cover is $x_*-\delta$
and every $(w-1)$ steps we cover a distance of at least
$\kappa_*(\epsilon) \delta$ as long as we have not surpassed $x_*(\epsilon)$.  
Therefore, after $(w-1) \lfloor
\frac{x_*(\epsilon)-\delta}{\kappa_*(\epsilon) \delta} \rfloor$ steps
we have either passed $x_*$ or we must be strictly closer to $x_*$ than $\kappa_*(\epsilon) \delta$.
Hence, to cover the remaining distance we need at most $(w-2)$ extra
steps.  The total number of steps needed is therefore upper bounded by
$w-2+(w-1) \lfloor \frac{x_*(\epsilon)-\delta}{\kappa_*(\epsilon)
\delta} \rfloor$, which, in turn, is upper bounded by $w
(\frac{x_*(\epsilon)}{\kappa_*(\epsilon) \delta}+1)$.  The final claim
follows by bounding $x_*(\epsilon)$ with $1$ and $\kappa_*(\epsilon)$
by $\kappa_*$.

\item[(ii)] 
{\em From $x_*(\epsilon)$ up to $\xunstab(\epsilon)$ it takes at most
$ w (\frac{8}{3 \kappa_* (x_*)^2}+2)$ sections.}

Recall that $\xavg_i$ is defined by 
$\xavg_i = \frac1{w^2}\sum_{j,k =0}^{w-1}x_{i+j-k}.$
From Lemma~\ref{lem:avgprop} (i),
$x_i \leq \epsilon g(\xavg_i,\xavg_i, \dots, \xavg_i) = \xavg_i + h(\xavg_i)$.
Sum this inequality over all sections from $-\infty$ to $k \leq 0$,  
\begin{align*}
\sum_{i=-\infty}^{k} x_i \leq \sum_{i=-\infty}^{k} \xavg_i + \sum_{i=-\infty}^{k} h(\xavg_i).
\end{align*}
Writing $ \sum_{i=-\infty}^{k} \xavg_i$ in terms of the $x_i$, for all $i$, and rearranging terms,
\begin{align*}
- \sum_{i=-\infty}^{k} h(\xavg_i) & \leq
\frac{1}{w^2}\sum_{i=1}^{w-1} {w-i+1 \choose 2} (x_{k+i}-x_{k-i+1})  \\
& \leq \frac{w}6 (x_{k+(w-1)}-x_{k-(w-1)}).
\end{align*}
Let us summarize:
\begin{align}\label{equ:momentum}
x_{k+(w-1)}-x_{k-(w-1)} & \geq - \frac{6}{w} \sum_{i=-\infty}^{k} h(\xavg_i).
\end{align}

From (i) and our discussion at the beginning, we can assume that
there exists a section $k$ so that $x_*(\epsilon) \leq x_{k-(w-1)}$.
Consider sections $x_{k-(w-1)}, \dots, x_{k+(w+1)}$, so that in
addition $x_{k+(w-1)} \leq \xunstab(\epsilon)$.  If no such $k$ exists
then there are at most $2w-1$ points in the interval $[x_*(\epsilon),
\xunstab(\epsilon)]$, and the statement is correct a fortiori.

From (\ref{equ:momentum}) we know that we have to lower bound $-\frac{6}{w}
\sum_{i=-\infty}^{k} h(\xavg_i)$. Since by assumption $x_{k+(w-1)} \leq
\xunstab(\epsilon)$, it follows that $\xavg_{k} \leq \xunstab(\epsilon)$, so that
every contribution in the sum $-\frac{6}{w}
\sum_{i=-\infty}^{k} h(\xavg_i)$ is positive.  Further, by (the Spacing)
Lemma~\ref{lem:spacing}, $w(\xavg_i - \xavg_{i-1}) \leq 1$.  Hence,
\begin{align}\label{equ:negarea1}
-\frac{6}{w} \sum_{i=-\infty}^{k} h(\xavg_i) & \geq
-6 \sum_{i=-\infty}^{k} h(\xavg_i)(\xavg_{i} - \xavg_{i-1}).
\end{align} 
Since by assumption $x_*(\epsilon) \leq x_{k-(w-1)}$, it follows that
$\xavg_{k} \geq x_*(\epsilon)$ and by definition $x_{-\infty}=0$.
Finally, according to Lemma~\ref{lem:propertyofh(x)} (iii), $-h(x) \geq \kappa_*(\epsilon) x$
for $x \in [0, x_*(\epsilon)]$. Hence,
\begin{align}\label{equ:negarea2}
-6 \sum_{i=-\infty}^{k} h(\xavg_i)(\xavg_{i} - \xavg_{i-1}) 
& \geq  6 \kappa_*(\epsilon) \int_{0}^{\frac{x_*(\epsilon)}2} x dx \nonumber \\ & 
=  \frac34 \kappa_*(\epsilon) (x_*(\epsilon))^2.
\end{align}
The inequality in \eqref{equ:negarea2} follows since there must exist a section
with value greater than $\frac{x_*(\epsilon)}2$ and smaller than $x_*(\epsilon)$.
Indeed,  suppose, on the contrary, that there is no section with value
between $(\frac{x_*(\epsilon)}2, x_*(\epsilon))$. Since $\xavg_k \geq
x_*(\epsilon)$, we must then have that $\xavg_{k} - \xavg_{k-1} >
\frac{x_*(\epsilon)}2$. But by the Spacing Lemma~\ref{lem:spacing} we have that
$\xavg_k - \xavg_{k-1}\leq \frac1{w}$. This would imply that $\frac1{w} >  
\frac{x_*(\epsilon)}2$. In other words, $w<\frac{2}{x_*(\epsilon)}$. Using the
 universal lower bound on $x_*(\epsilon)$ from Lemma~\ref{lem:propertyofh(x)} (iv),
 we conclude that $w< 2 \dl^2 \dr^2$, a contradiction to the hypothesis of the
 lemma. 

Combined with (\ref{equ:momentum}) this implies that
\begin{align*}
x_{k+(w-1)}-x_{k-(w-1)} \geq \frac34 \kappa_*(\epsilon) (x_*(\epsilon))^2.
\end{align*}
We summarize. The total distance we have to cover is
$\xunstab(\epsilon)-x_*(\epsilon)$ and every $2(w-1)$ steps we cover a
distance of at least $\frac34 \kappa_*(\epsilon) (x_*(\epsilon))^2$ as long
as we have not surpassed $\xunstab(\epsilon)$.  Allowing for $2(w-1)-1$
extra steps to cover the last part, bounding again $w-1$ by $w$,
bounding $\xunstab(\epsilon)-x_*(\epsilon)$ by $1$ and replacing $\kappa_*(\epsilon)$
and $x_*(\epsilon)$ by their universal lower bounds,
proves the claim.

\item[(iii)]
{\em From $\xunstab(\epsilon)$ to $\xunstab(\epsilon)+\frac{3
\kappa_* (x_*)^2}{4(1+12 \dl \dr)}$ it takes at most
$2 w$ sections.}

Let $k$ be the smallest index so that $\xunstab(\epsilon)\leq
x_{k-(w-1)}$.  It follows that $\xavg_{k-2 w+1} \leq \xunstab(\epsilon)
\leq \xavg_{k}$.  Let $\hat{k}$ be the largest index so that
$\xavg_{\hat{k}}\leq \xunstab(\epsilon)$. From the previous line we
deduce that $k-2 w+1 \leq \hat{k} < k$, so that $k-\hat{k} \leq 2 w-1$.

 We use
again (\ref{equ:momentum}). Therefore, let us bound $-\frac{6}{w}
\sum_{i=-\infty}^{k} h(\xavg_i)$.  We have
\begin{align*}
-\frac{6}{w} \! \sum_{i=-\infty}^{k}& \!\!h(\xavg_i)  =
-\frac{6}{w} \sum_{i=-\infty}^{\hat{k}} \!\!h(\xavg_i) 
\!-\!\frac{6}{w} \! \sum_{i=\hat{k}+1}^{k} h(\xavg_i) \\ 
& \stackrel{(a)}{\geq} \frac34 \kappa_*(\epsilon) (x_*(\epsilon))^2 \!-\! 12 \dl \dr (x_{k+(w-1)}\!-\!\xunstab(\epsilon)). 
\end{align*}
We obtain $(a)$ as follows. There are two sums, one from $-\infty$ to $\hat{k}$
and another from $\hat{k}+1$ to $k$. Let us begin with the sum from $-\infty$ to $\hat{k}$. 
 First, we claim that $\xavg_{\hat{k}}\geq \frac{x_*(\epsilon)}2$. 
Indeed, suppose $\xavg_{\hat{k}}<\frac{x_*(\epsilon)}2$. Then, using the definition of $\hat{k}$, 
\begin{align*}
\xavg_{\hat{k}+1} - \xavg_{\hat{k}}
& >    \xunstable -  \frac{x_*(\epsilon)}2 \geq \frac{\xunstable}{2} \geq \frac{\xunstab(1)}{2}  \\
& \stackrel{(\ref{equ:boundonxuone})}{\geq} 
\frac{(\dr-1)^{-\frac{\dl-1}{\dl-2}}}{2} \geq \frac{1}{2 \dr^2}. 
\end{align*}
But from (the Spacing) Lemma~\ref{lem:spacing}, $\xavg_{\hat{k}+1}
- \xavg_{\hat{k}} \leq \frac1{w}$, a contradiction, since from the
hypothesis of the lemma $w\geq 2\dr^2$.  Using \eqref{equ:negarea1}
and \eqref{equ:negarea2} with the integral from $0$ to $x_*(\epsilon)/2$
we get the first expression in $(a)$. Note that the integral till
$x_*(\epsilon)/2$ suffices because either $\xavg_{\hat{k}}\leq x_*(\epsilon)$
or, following an argument similar to the one after \eqref{equ:negarea2}, 
there must exist a section with value between $(\frac{x_*(\epsilon)}2,
x_*(\epsilon))$. We now focus on the sum from $\hat{k}+1$ to $k$. 
 From the definition of $\hat{k}$, for all $i\in [\hat{k}+1, k]$, $\vert
h(\xavg_i)\vert \leq \dl \dr (\xavg_i - \xunstable)$. 
 Indeed, recall from Lemma~\ref{lem:propertyofh(x)} that $|h'(x)| \leq \dl \dr$
for $x \in [0, 1]$.  In particular, this implies that the line with
slope $\dl \dr$ going through the point $(\xunstab(\epsilon),
0)$ lies above $h(x)$ for $x \geq \xunstab(\epsilon)$. Further, $\xavg_i -
\xunstable \leq \xavg_k - \xunstable \leq x_{k+w-1} - \xunstable$. Finally,
using $k-\hat{k} \leq 2w-1$ we get the second expression in $(a)$.  

From (\ref{equ:momentum}) we now conclude that 
\begin{align*}
x_{k+w-1} &- \xunstab(\epsilon)  
 \geq \\ 
& \frac34 \kappa_*(\epsilon) (x_*(\epsilon))^2 - 12 \dl \dr (x_{k+w-1}-\xunstab(\epsilon)),
\end{align*}
which is equivalent to 
\begin{align*}
x_{k+(w-1)}-\xunstab(\epsilon) \geq 
\frac{3 \kappa_*(\epsilon) (x_*(\epsilon))^2
}{4(1+12 \dl \dr)}.
\end{align*}
The final claim follows by replacing again $\kappa_*(\epsilon)$
and $x_*(\epsilon)$ by their universal lower bounds $\kappa_*$
and $x_*$.

\item[(iv)] 
{\em From $\xunstab(\epsilon)+\frac{3 \kappa_*
(x_*)^2}{4(1+12 \dl \dr)}$ to $\xstab(\epsilon)-\delta$ it
takes at most  $w \frac{1}{\delta
\min\{\kappa^{\text{min}}, \lambda^{\text{min}}\}}$ steps, where}
\begin{alignat*}{2}
\kappa^{\text{min}} & = \min_{\epsilon^{\text{min}} \leq \epsilon \leq 1} 
\kappa^*(\epsilon), \phantom{xxx}
& \lambda^{\text{min}} & = \min_{\epsilon^{\text{min}} \leq \epsilon \leq 1} \lambda^*(\epsilon).
\end{alignat*}

From step (iii) we know that within a fixed number of steps we
reach at least $\frac{3 \kappa_* (x_*)^2}{4(1+12
\dl \dr)}$ above $\xunstab(\epsilon)$. On the other hand we know from
Lemma~\ref{lem:maximum} that $x_0 \leq \xstab(\epsilon)$. We
conclude that $\xstab(\epsilon)-\xunstab(\epsilon)
\geq \frac{3 \kappa_* (x_*)^2}{4(1+12
\dl \dr)}$. From Lemma~\ref{lem:stableandunstable} we know that
$\xstab(\epsilon^{\BPsmall})-\xunstab(\epsilon^{\BPsmall})=0$
and that this distance is strictly increasing for $\epsilon\geq
\epsilon^{\BPsmall}$. Therefore there exists a unique number, call
it $\epsilon^{\text{min}}$, 
$\epsilon^{\text{min}} > \epsilon^{\BPsmall}(\dl, \dr)$, so that
\begin{align*}
\xstab(\epsilon)-\xunstab(\epsilon) \geq \frac{3 \kappa_* (x_*)^2}{4(1+12 \dl \dr)},
\end{align*}
if and only if $\epsilon \geq \epsilon^{\text{min}}$.
 As defined above let,
\begin{alignat*}{2}
\kappa^{\text{min}} & = \min_{\epsilon^{\text{min}} \leq \epsilon \leq 1} 
\kappa^*(\epsilon), \phantom{xxx}
& \lambda^{\text{min}} & = \min_{\epsilon^{\text{min}} \leq \epsilon \leq 1} \lambda^*(\epsilon).
\end{alignat*}
Since $\epsilon^{\text{min}} > \epsilon^{\BPsmall}(\dl, \dr)$,
both $\kappa^{\text{min}}$ and $\lambda^{\text{min}}$ are strictly
positive.  Using similar reasoning as in step (i), we conclude that
in order to reach from $\xunstab(\epsilon)+\frac{3 \kappa_*
(x_*)^2}{4(1+12 \dl \dr)}$ to $\xstab(\epsilon)-\delta$ it
takes at most $w \frac{\xstab(\epsilon)-\xstab(\epsilon)}{\delta
\min\{\kappa^{\text{min}}, \lambda^{\text{min}}\}}$ steps, where
we have used the fact that, by assumption,  $\delta \leq  \frac{3
\kappa_* (x_*)^2}{4(1+12 \dl \dr)}$.
\end{itemize}

From these four steps we see that we need at most
\begin{align*}
& w(\frac1{\delta} [\frac{1}{\kappa_*}+\frac{1}{\min\{\kappa^{\text{min}}, \lambda^{\text{min}}\}}] + [\frac{2}{3\kappa_*(x_*)^2}+5]) \\
& \leq w\frac1{\delta} [\frac{1}{\kappa_*}+\frac{1}{\min\{\kappa^{\text{min}}, \lambda^{\text{min}}\}}+\frac{2}{3\kappa_*(x_*)^2}+5] \\
& \triangleq w \frac{c(\dl, \dr)}{\delta}
\end{align*}
sections in order to reach $\xstable-\delta$ once we reach $\delta$.
This constant depends on $(\dl, \dr)$ but it is independent
of $L$ and $\epsilon$.

\section{Proof of Theorem~\ref{thm:onesidedDE}}\label{app:onesidedDE}
To establish the existence of $\x$ with the desired properties, we use
the Brouwer FP theorem: it states that every continuous function $f$ from
a convex compact subset $S$ of a Euclidean space to $S$ itself has a FP.

Let $\z$ denote the one-sided forward DE FP for parameter
$\epsilon=1$. Let the length $L$ be chosen
in accordance with the the statement of the theorem. By assumption
$L>\frac{w}{\frac{\dr}{\dl}-1}$. Using Lemma~\ref{lem:nontrivialonesided}
part (ii),  we conclude that $\chi(\z)\geq \frac12(1-\frac{\dl}{\dr})$,
i.e., $\z$ is non-trivial.  By Lemma~\ref{lem:nontrivialonesided} part
(i), it is therefore proper, i.e., it is non-decreasing. Suppose that
$\chi(\z)\leq \chi$. In this case, it is easy to verify that the second
statement of the theorem is true. So in the remainder of the proof we
assume that $\chi(\z)>\chi$.

Consider the Euclidean space $[0, 1]^{L+1}$.  
Let $S(\chi)$ be the subspace
\begin{align*}
S(\chi)   = & \{ \x \in [0, 1]^{L+1}: \chi(\x) = \chi; x_i \leq z_i, i \in [-L,0]; \\ 
& x_{-L}\leq x_{-L+1}\leq \dots\leq x_0\}.
\end{align*}
First note that $S(\chi)$ is non-empty since $\z$ is non-trivial and
has entropy at least $\chi$.  We claim that $S(\chi)$ is convex and
compact. Indeed, convexity follows since $S(\chi)$ is a convex polytope
(defined as the intersection of half spaces).  Since $S(\chi) \subset
[0, 1]^{L+1}$ and $S(\chi)$ is closed, $S(\chi)$ is compact.

Note that any constellation belonging to $S(\chi)$ has entropy $\chi$
and is increasing, i.e., any such constellation is proper. 
 Our first step is to define a map $V(\x)$
which ``approximates'' the DE equation and is well-suited for applying
the Brouwer FP theorem. The final step in our proof is then to show that
the FP of the map $V(\x)$ is in fact a FP of one-sided DE.

The map $V(\x)$ is constructed as follows. 
For $\x\in S(\chi)$, let $U(\x)$ be the map, 
$$
(U(\x))_i = g(x_{i-w+1}, \dots, x_{i+w-1}), \quad i\in [-L,0].
$$

Define $V: S(\chi)\to S(\chi)$ to be the map
\begin{align*}
V(\x) & = 
\begin{cases}
U(\x) \frac{\chi}{\chi(U(\x))}, & \chi \leq \chi(U(\x)), \\
\alpha(\x) U(\x) + (1-\alpha(\x)) \z, & \text{otherwise},
\end{cases}
\end{align*}
where $$\alpha(\x) = \frac{\chi(\z) - \chi}{\chi(\z) - \chi(U(\x))}.$$
Let us show that this map is well-defined. First consider the
case $\chi \leq \chi(U(\x))$.  Since $\x\in S(\chi)$, $\x \leq \z$
(componentwise). By construction, it follows that  $U(\x) \leq U(\z) =
\z$, where the last step is true since $\z$ is the forward FP of DE for
$\epsilon=1$. We conclude that $U(\x) \frac{\chi}{\chi(U(\x))} \leq \z$.
Further, by construction $\chi(U(\x) \frac{\chi}{\chi(U(\x))})=\chi$.
It is also easy to check that $U(\x)$ is non-negative and that it is
non-decreasing.  It follows that in this case $V(\x) \in S(\chi)$.

Consider next the case $\chi > \chi(U(\x))$.  As we have seen, $\x \leq
\z$ so that $\chi(U(\x)) \leq \chi(U(\z))=\chi(\z)$.  Together with
$\chi > \chi(U(\x))$ this shows that $\alpha(\x) \in [0, 1]$. Further, the
choice of $\alpha(\x)$ guarantees that $\chi(V(\x))=\chi$. It is easy to
check that $V(\x)$ is increasing and bounded above by $\z$. This
shows that also in this case $V(\x) \in S(\chi)$.

We summarize, $V$ maps $S(\chi)$ into itself.

In order to be able to invoke Brouwer's theorem we need to show that
$V(\x)$ is continuous. This means we need to show that for every $\x
\in S(\chi)$ and for any $\varepsilon>0$, there exists a $\delta>0$
such that if $\y \in B(\x, \delta)\cap S(\chi)$ then $\| V(\y)-V(\x)\|_2\leq
\varepsilon$.

First, note that $U(\x)$ and $\chi(\x)$ are continuous maps on
$S(\chi)$. As a result, $\chi(U(\x))$, which is the composition of two
continuous maps, is also continuous.

Fix $\x \in S(\chi)$. We have three cases: (i) $\chi(U(\x)) > \chi$,
(ii) $\chi(U(\x)) < \chi$, and (iii) $\chi(U(\x)) = \chi$.

We start with (i). Let $\rho = \chi(U(\x)) - \chi$ and fix
$\varepsilon>0$.  From the continuity of $\chi(U(\x))$ we know that there
exists a ball $B(\x, \nu_1)$ of radius $\nu_1>0$ so that if $\y \in B(\x,
\nu_1) \cap S(\chi)$ then $\vert \chi(U(\x)) - \chi(U(\y))\vert \leq
\rho$, so that $\chi(U(\y)) \geq \chi$. It follows that for those $y$,
$V(\y) = U(\y) \frac{\chi}{\chi(U(\y))}$.

For a subsequent argument we will need also a tight bound on $\vert
\chi(U(\x)) - \chi(U(\y))\vert$ itself. Let us therefore choose $\gamma
= \min \{\varepsilon, \rho\}$, $\gamma>0$. And let us choose $\nu_1$
so that if $\y \in B(\x, \nu_1) \cap S(\chi)$ then $\vert \chi(U(\x))
- \chi(U(\y))\vert \leq \frac{\gamma\chi}{2(L+1)}$, so that $\chi(U(\y))\geq \chi$.

Further, since
$U(\cdot)$ is  continuous, there exists $\nu_2>0$ such that for all $\y\in
B(\x,\nu_2) \cap S(\chi)$, $\| U(\x) - U(\y)\|_2 \leq \frac{\varepsilon}2$. 
Choose $\nu = \min\{\nu_1, \nu_2\}$. Then for all $\y\in B(\x,\nu) \cap S(\chi)$,
\begin{align*}
&\| V(\x) - V(\y)\|_2 = \chi \Big\| \frac{U(\x)}{\chi(U(\x))} -
\frac{U(\y)}{\chi(U(\y))}\Big\|_2 \\
& \leq \chi \Big\| \frac{U(\x)}{\chi(U(\x))}
-\frac{U(\y)}{\chi(U(\x))}\Big\|_2 
 + \chi\Big\| \frac{U(\y)}{\chi(U(\x))} - \frac{U(\y)}{\chi(U(\y))}\Big\|_2 \\
& \stackrel{\chi(U(\x))>\chi}{\leq}  \| U(\x)-U(\y)\|_2 + \frac{\|U(\y)\|_2}{\chi}\Big\vert \chi(U(\x)) - \chi(U(\y))\Big\vert \\
& \leq  \frac{\varepsilon}2 + \frac{(L+1)}{\chi}\Big\vert\chi(U(\x)) - \chi(U(\y))\Big\vert \\
& \leq  \frac{\varepsilon}2 + 
\frac{(L+1)}{\chi}\frac{\gamma\chi}{2(L+1)} \leq \varepsilon,
\end{align*} 
where above we used the bound $\|U(\y)\|_2 \leq (L+1)$.

Using similar logic, one can prove (ii). 

Consider claim (iii). In this case 
$\chi(U(\x)) = \chi$, which implies that $V(\x) = U(\x)$. 
As before, there exists $0<\nu_1$ such that for all $\y \in B(\x,
\nu_1) \cap S(\chi)$,  $\|U(\x) - U(\y)\|_2 < \frac{\varepsilon}2$.
Let $\gamma = \min\{\chi(\z) - \chi, \chi\}$. 
Since we assumed that $\chi(\z)>\chi$, we have $\gamma>0$.  
Furthermore, there
exists $0<\nu_2$ such that for all $\y\in B(\x,\nu_2) \cap S(\chi)$,  $\vert\chi(U(\x))
- \chi(U(\y))\vert < \frac{\gamma\varepsilon}{2(L+1)}$.  Choose
$\nu = \min\{ \nu_1, \nu_2\}$. Consider $\y\in B(\x,\nu) \cap S(\chi)$. Assume
first that $\chi(U(\y)) \geq \chi$. Thus, as before,
\begin{align*}
\| V(\x) - V(\y) \|_2 \leq \varepsilon. 
\end{align*}
Now let us assume that $\chi(U(\y)) < \chi$. Then we have
\begin{align*}
& \| V(\x) - V(\y) \|_2 = \| U(\x) - \alpha(\y) U(\y) - (1 - \alpha(\y)) \z\|_2 \\
& \leq \alpha(\y)\| U(\x) - U(\y)\|_2 + |1-\alpha(\y)|\|U(\x) - U(\z) \|_2 \\
& \leq \frac{\varepsilon}2 + (L+1) \Big\vert\frac{\chi(U(\y)) - \chi(U(\x))}{\chi(\z) - \chi(U(\y))} \Big\vert\\
& \leq \frac{\varepsilon}2 + \frac12 \Big\vert\frac{\gamma\varepsilon}{\chi(\z) - \chi} \Big\vert < \varepsilon,
\end{align*}
where above we used: (i) $\|U(\x)
- U(\z)\|_2 \leq L+1$,   (ii) $\chi(U(\y)) < \chi$, (iii) $\chi(U(\x))
= \chi$ (when we explicitly write $\vert 1-\alpha(\y) \vert$).

We can now invoke Brouwer's FP theorem to conclude that
$V(\cdot)$ has a FP in $S(\chi)$, call it $\x$.

Let us now show that, as a consequence, either there exists a one-sided
FP of DE with parameter $\epsilon=1$ and entropy bounded between
$\frac{(1-\frac{\dl}{\dr})(\chi-\xunstab(1))}{8}-\frac{\dl w}{2\dr (L+1)}$
and $\chi$, or $\x$ itself is a proper one-sided FP of DE with entropy $\chi$.
Clearly, either $\chi \leq \chi(U(\x))$ or $\chi(U(\x)) < \chi$. In the first case, i.e., if
$\chi \leq \chi(U(\x))$, then $\x=V(\x)=U(\x) \frac{\chi}{\chi(U(\x))}$.
Combined with the non-triviality of $\x$, we conclude that $\x$ is a proper
one-sided FP with entropy $\chi$ and the channel parameter (given by $\frac{\chi}{\chi(U(\x))}$)
less than or equal to 1. Also, from
Lemma~\ref{lem:maximum} we then conclude that the channel parameter is strictly
greater than $\epsilon^{\BPsmall}(\dl, \dr)$.   

Assume now the second case, i.e., assume that $\chi(U(\x))<\chi$. This implies that
$$ \x = \alpha(\x) (U(\x)) + (1-\alpha(\x))\z.  $$
But since $\x \leq \z$, $$ \alpha(\x) \x + (1-\alpha(\x))\z \geq \x = \alpha(\x)
(U(\x)) + (1-\alpha(\x)) \z.  $$ As a result, $\x \geq (U(\x))$. We will
now show that this implies the existence of a one-sided FP of DE with
parameter $\epsilon=1$ and with entropy bounded 
between $\frac{(1-\frac{\dl}{\dr})(\chi-\xunstab(1))}{8}-\frac{\dl w}{2 \dr (L+1)}$ and $\chi$. 

Let $\x^{(0)}=\x$ and define $\x^{(\ell)}= U(\x^{(\ell-1)})$, $\ell
\geq 1$.  By assumption, $\x \geq U(\x)$, i.e., $\x^{(0)} \geq
\x^{(1)}$. By induction this implies that $\x^{(\ell-1)} \geq
\x^{(\ell)}$, i.e, the sequence $\x^{(\ell)}$ is monotonically
decreasing. Since it is also bounded from below, it converges 
 to a fixed point of DE with parameter $\epsilon=1$,
call it $\x^{(\infty)}$.

We want to show that  $\x^{(\infty)}$ is non-trivial and we want
to give a lower bound on its entropy.  We do this by comparing
$\x^{(\ell)}$ with a constellation that lower-bounds $\x^{(\ell)}$
and which converges under DE to a non-trivial FP.

We claim that at least the last $N=(L+1) \frac{\chi-\xunstab(1)}{2}$
components of $\x$ are above $\frac{\chi+\xunstab(1)}{2}$:
\begin{align*} \chi (L+1) & = \chi(\x) (L+1) \leq N + (L+1-N)
\frac{\chi+\xunstab(1)}{2}, \end{align*} where on the right hand
side we assume (worst case) that the last $N$ components have height $1$ and the
previous $(L+1-N)$ components have height $\frac{\chi+\xunstab(1)}{2}$.
If we solve the inequality for $N$ we get $N \geq 
(L+1) \frac{\chi-\xunstab(1)}{2-\chi-\xunstab(1)} \geq (L+1) \frac{\chi-\xunstab(1)}{2}$. 

Consider standard DE for the underlying regular $(\dl,\dr)$ ensemble
and $\epsilon=1$.  
We claim that it takes at most $m$
\begin{align*}
m = \max\{ 
	\frac{2}{\kappa^*(1)(\chi-\xunstab(1))},
	\frac{2}{\lambda^*(1)(1-\frac{\dl}{\dr})}
	\}
\end{align*}
DE steps to go from the value $\frac{\chi+\xunstab(1)}{2}$ to a
value above $\frac{1+\frac{\dl}{\dr}}{2}$.  The proof idea is along
the lines used in the proof of  Lemma~\ref{lem:transitionlength}.
Consider the function $h(x)$ as defined in (\ref{equ:hfunction}) for
$\epsilon=1$.  Note that $\xunstab(1) < \frac{\chi+\xunstab(1)}{2}$
and that $\frac{1+\frac{\dl}{\dr}}{2} < \xstab(1)=1$. Further,
the function $h(x)$ is unimodal and strictly positive in the range
$(\xunstab(1), \xstab(1))$ and $h(x)$ is equal to the change in $x$
which happens during one iteration, assuming that the current value is $x$. 
If $\frac{\chi+\xunstab(1)}{2} \geq \frac{1+\frac{\dl}{\dr}}{2}$
then the statement is trivially
true. Otherwise, the progress in each required step is at least equal to
\begin{align*}
& \min\{
h(\frac{\chi+\xunstab(1)}{2}),
	h(\frac{1+\frac{\dl}{\dr}}{2})
\} \\
\geq & \min\{\kappa^*(1) (\frac{\chi+\xunstab(1)}{2}-\xunstab(1)), \lambda^*(1) (1-\frac{1+\frac{\dl}{\dr}}{2})\}.
\end{align*}
We use Lemma~\ref{lem:propertyofh(x)} part (vi) to get the last inequality. 
The claim now follows by observing that the total distance that has to be covered
is no more than $1$.

Consider the constellation $\y^{(0)}$, which takes the value $0$ for $[-L, -N]$
and the value $\frac{\chi+\xunstab(1)}{2}$ for $[-N+1, 0]$. By construction,
$\y = \y^{(0)} \leq \x^{(0)} = \x$. Define $\y^{(\ell)} = U(\y^{(\ell-1)})$,
$\ell \geq 1$. By monotonicity we know that $U(\y^{(\ell)}) \leq
U(\x^{(\ell)})$ (and hence $\y^{(\infty)} \leq \x^{(\infty)}$). In particular
this is true for $\ell = m$. But note that at least the last $N- wm$ positions
of $y^{(m)}$ are above $\frac{1+\frac{\dl}{\dr}}{2}$. Also, by the choice of
$L$, $N- w m \geq N/2$.

Define the constellation $\vc^{(0)}$ which takes the value $0$ for $[-L,
-N/2]$ and the value $\frac{1+\frac{\dl}{\dr}}{2}$ for $[-N/2+1, 0]$. 
Define $\vc^{(\ell)}=\frac{1+\frac{\dl}{\dr}}{2}
U(\vc^{(\ell-1)})$, $\ell \geq 0$. 
 Again, observe that by definition $\vc^{(0)}\leq \y^{(m)}$ and $\frac{1+\frac{\dl}{\dr}}{2}\leq 1$,  hence we have  
$\vc^{(\infty)}\leq \y^{(\infty)}$. 
From Lemma~\ref{lem:nontrivialonesided} we know that 
for a length $N/2=(L+1)\frac{\chi-\xunstab(1)}{4}$ and a channel parameter $\frac{1+\frac{\dl}{\dr}}{2}$ 
the resulting FP of forward DE has entropy at least
\begin{align*}
\chi'= \frac{1-\frac{\dl}{\dr}}{4}-\frac{\dl w}{\dr (\chi-\xunstab(1)) (L+1)}>0.
\end{align*}
Above, $\chi'>0$ follows from the first assumption on $L$ in the hypothesis of the theorem. 
It follows that $\vc^{(\infty)}$ has
(unnormalized) entropy at least equal to $\chi'(N/2)$ and therefore
normalized entropy at least $\frac{\chi'(\chi-\xunstab(1))}{4}$.

Since $\x^{(\infty)} \geq \y^{(\infty)} \geq \vc^{(\infty)}$, 
 we conclude that $\x^{(\infty)}$ is a one-sided FP of DE
for parameter $\epsilon=1$ with entropy bounded between 
$\frac{(1-\frac{\dl}{\dr})(\chi-\xunstab(1))}{8}-\frac{\dl w}{2\dr(L+1)}$ and $\chi$.

\section{Proof of Theorem~\ref{the:propertiesEXIT}}\label{app:fundamental}
\begin{itemize}
\item[(i)] {\em Continuity:}
In phases (i), (ii), and (iv) the map is differentiable by construction.
In phase (iii) the map is differentiable in each ``period.''  Further,
by definition of the map, the (sub)phases are defined in such a way that
the map is continuous at the boundaries.

\item[(ii)] {\em Bounds in Phase (i):}
Consider $\alpha \in [\frac34, 1]$. By construction of the EXIT
curve, all elements $x_i(\alpha)$, $i \in [-L, 0]$, are the same. In
particular, they are all equal to $x_0(\alpha)$. Therefore, all values
$\epsilon_i(\alpha)$, $i \in [-L+w-1, 0]$, are identical, and equal
to $\epsilon_0(\alpha)$.

For points close to the boundary, i.e., for $i \in [-L, -L+w-2]$, some
of the inputs involved in the computation of $\epsilon_i(\alpha)$ are
$0$ instead of $x_0(\alpha)$. Therefore, the local channel parameter
$\epsilon_i(\alpha)$ has to be strictly bigger than $\epsilon_0(\alpha)$
in order to compensate for this. This explains the lower bound on
$\epsilon_i(\alpha)$.

\item[(iii)] {\em Bounds in Phase (ii):}
Let $i \in [-L, 0]$ and $\alpha \in [\frac12, \frac34]$. Then
\begin{align*}
x_{-L}^* & \leq x_i(\alpha)  
          = \epsilon_i(\alpha) g(x_{i-w+1}(\alpha), \dots, x_{i+w-1}(\alpha)) \\
         & \leq \epsilon_i(\alpha) g(x_{0}^*, \dots, x_{0}^*) 
          = \epsilon_i(\alpha) \frac{x_{0}^*}{\epsilon(x_{0}^*)}.
\end{align*}
This gives the lower bound $\epsilon_i(\alpha) \geq \epsilon(x_{0}^*)
\frac{x_{-L}^*}{x_{0}^*}$.

\item[(iv)] {\em Bounds in Phase (iii):}
Let $\alpha \in [\frac14, \frac12]$ and $i\in [-L, 0]$.  Note that
$x_0(\frac12)=x^*_0$ but that $x_0(\frac14)=x^*_{-L'+L}$.  The range
$[\frac14, \frac12]$ is therefore split into $L'-L$ ``periods.'' In each
period, the original solution $x^*$ is ``moved in'' by one segment.
Let $p \in \{1, \dots, L'-L\}$ denote the current period we are
operating in.  In the sequel we think of $p$ as fixed and consider in
detail the interpolation in this period. To simplify our notation, we
reparameterize the interpolation so that if $\alpha$ goes from $0$ to
$1$, we moved in the original constellation exactly by one more segment.
This alternative parametrization is only used in this section.  In part
(vi), when deriving bounds on $\epsilon^*$, we use again the original
parametrization.

Taking this reparameterization into account, 
for $\alpha \in [0, 1]$, according
to Definition~\ref{def:EXIT},
\begin{align*}
x_i(\alpha)=
\begin{cases}
(x^*_{i-p})^{\alpha} (x^*_{i-p+1})^{1-\alpha}, & i \in [-L, 0], \\
0, & i < -L.
\end{cases}
\end{align*}
We remark that $x_i(\alpha)$ decreases with $\alpha$. Thus we have for any $\alpha$, $x_i(1)\leq x_i(\alpha) \leq x_i(0)$. 
By symmetry, $x_i(\alpha)=x_{-i}(\alpha)$ for $i \geq 1$.

We start by showing that if $x_i(\alpha) > \gamma$ and $i \in [-L+w-1,-w+1]$ then
$\epsilon_i(\alpha)/\epsilon^* \leq 1+\frac{1}{w^{1/8}}$.  For $\alpha
\in [0, 1]$, define
\begin{align*} 
f_i(\alpha) & =  \Bigl( 1- \frac1w\sum_{k=0}^{w-1}
x_{i-k}(\alpha) \Bigr)^{\dr-1}.
\end{align*}
Further, define
\begin{align*} 
f_i^* & =  \Bigl( 1- \frac1w\sum_{k=0}^{w-1} x^*_{i-p+1-k} \Bigr)^{\dr-1}.  
\end{align*}
Note that the values $x_i^*$ in the last definition are the
values of the one-sided FP. In particular, this means that for $i \geq 0$
we have $x_i^*=x_0^*$.

From the definition of the EXIT curve we have 
\begin{align}\label{eq:epsi} 
\epsilon_i(\alpha) =
\frac{x_i(\alpha)}{\Bigl(1-\frac{1}{w}
\sum_{j=0}^{w-1} f_{i+j}(\alpha)\Bigr)^{\dl-1}}.
\end{align} 
By monotonicity,
\begin{align*}
\Bigl(1\!-\!\frac{\sum_{j=0}^{w-1} f_{i+j}(\alpha)}{w}\Bigr)^{\dl-1} & \!\!\!\!\!\geq \!
\Bigl(1\!-\!\frac{\sum_{j=0}^{w-1} f_{i+j-1}^*}{w}\Bigr)^{\dl-1} 
\!\!\!\!\!=\! \frac{x^*_{i-p}}{\epsilon^*}.
\end{align*}
In the first step we used the fact that $-L+w-1\leq i\leq -w+1$ and the second step
is true by definition. 

Substituting this into the denominator of \eqref{eq:epsi} results in 
\begin{align*}
\frac{\epsilon_i(\alpha)}{\epsilon^*} & 
\leq \Big(\frac{x^*_{i-p+1}}{x^*_{i-p}}\Big)^{1-\alpha}  
\leq \frac{x^*_{i-p+1}}{x^*_{i-p}} = 
1+\frac{1}{\frac{x^*_{i-p+1}}{(\Delta x^*)_{i-p+1}} -1},
\end{align*}
where we defined $(\Delta x^*)_i=x^*_{i}-x^*_{i-1}$.
If we plug the upper bound on $(\Delta x^*)_{i-p+1}$ due to (the Spacing) 
Lemma~\ref{lem:spacing} into this expression we get  
\begin{align*}
\frac{1}{x^*_{i-p+1}/(\Delta x^*)_{i-p+1}-1} \leq \frac1{\Big(\frac{x^*_{i-p+1}}{\epsilon^*} \Big)^{\frac1{\dl-1}}\frac{w}{(\dl-1)(\dr-1)} - 1}.
\end{align*}
By assumption $x_i(\alpha)> \gamma$. But from the monotonicity we have 
$x^*_{i-p+1} =x_i(0) \geq x_i(\alpha)$. Thus $x^*_{i-p+1} > \gamma$. 
This is equivalent to 
\begin{align}\label{equ:spacingequivalent}
\Big(\frac{x^*_{i-p+1}}{\epsilon^*}\Big)^{\frac1{\dl-1}}\frac{w}{(\dr-1)(\dl-1)} - 1 & \geq w^{1/8}.
\end{align}
As a consequence,
\begin{align*}
\frac{\epsilon_i(\alpha)}{\epsilon^*} & 
\leq 1+\frac{1}{x^*_{i-p+1}/(\Delta x^*)_i-1} \leq  1+\frac1{w^{1/8}},
\end{align*}
the promised upper bound.

Let us now derive the lower bounds. First suppose that $x_i(\alpha)>\gamma$. 
For $i\in [-L, 0]$ we can use again monotonicity to conclude that
\begin{align*}
x^*_{i-p} & \leq x_i(\alpha) =\epsilon_i(\alpha) 
\Bigl(1\!-\!\frac{\sum_{j=0}^{w-1} f_{i+j}(\alpha)}{w}\Bigr)^{\dl-1} \\
& \leq \epsilon_i(\alpha) \frac{x^*_{i-p+1}}{\epsilon^*}.
\end{align*}
This proves that
\begin{align*}
\frac{\epsilon_i(\alpha)}{\epsilon^*} & \geq \frac{x^*_{i-p}}{x^*_{i-p+1}} = 
1 - \frac{(\Delta x^*)_{i-p+1}}{x^*_{i-p+1}}.
\end{align*}
Note that this sequence of inequalities is true for the whole range $i \in [-L, 0]$.
Since $x^*_{i-p+1}=x_i(0) \geq x_i(\alpha)$, we have $x^*_{i-p+1} > \gamma$ and using 
\eqref{equ:spacingequivalent} we have
\begin{align*}
\frac{(\Delta x^*)_{i-p+1}}{x^*_{i-p+1}} \leq \frac{1}{1+w^{1/8}}.
\end{align*} 
As a consequence,
\begin{align*}
\frac{\epsilon_i(\alpha)}{\epsilon^*} & \geq  
1 - \frac{(\Delta x^*)_{i-p+1}}{x^*_{i-p+1}}  \geq 1 - \frac{1}{1+w^{1/8}}.
\end{align*}

It remains to consider the last case, i.e., we assume that $x_i(\alpha) \leq \gamma$.
From Lemma~\ref{lem:avgprop}~(iv) we have
\begin{align*}
& (x_{i-p}^*/\epsilon^*)^{\frac{1}{\dl-1}}  \geq \\ 
& \phantom{\geq} \Big(1 \!-\! \frac1w\sum_{k=0}^{w-1}x_{i-p+w-1-k}^* \Big)^{\dr-2}
\frac{\dr\!-\!1}{w^2} \!\!\sum_{j, k=0}^{w-1}x_{i-p+j-k}^*  \\
& \geq \Big(1 \!-\! \frac1w\sum_{k=0}^{w-1}x_{i-p+w-k}^* \Big)^{\dr-2}
\frac{\dr\!-\!1}{w^2} \!\!\sum_{j, k=0}^{w-1}x_{i-p+j-k}^*, \\
\end{align*}
and
\begin{align*}
& (x_{i-p+1}^*/\epsilon^*)^{\frac{1}{\dl-1}}  \geq \\ 
& \phantom{\geq} \Big(1 \!-\! \frac1w\sum_{k=0}^{w-1}x_{i-p+w-k}^* \Big)^{\dr-2}
\frac{\dr\!-\!1}{w^2}\!\!
\sum_{j, k=0}^{w-1}x_{i-p+1+j-k}^*.
\end{align*}
We start with \eqref{eq:epsi}.  Write $x_i(\alpha)$ in the numerator
explicitly as $(x^*_{i-p})^{\alpha} (x^*_{i-p+1})^{1-\alpha}$
and bound each of the two terms by the above expressions.
This yields
\begin{align*}
& \Big(\frac{\epsilon_i(\alpha)}{\epsilon^*}\Big)^{\frac1{\dl-1}} \geq \Big(1 - \frac1w\sum_{k=0}^{w-1}x^*_{i-p+w-k}\Big)^{(\dr-2)} \frac{\dr-1}{w^2}\\
& \frac{(\sum_{j, k=0}^{w-1}x^*_{i-p+j-k})^{\alpha}(\sum_{j, k=0}^{w-1}x^*_{i-p+1+j-k})^{1-\alpha}}{1-\frac{1}{w} \sum_{j=0}^{w-1} f_{i+j}(\alpha)}. 
\end{align*}
Applying steps, similar to those used to prove Lemma~\ref{lem:avgprop} (ii), to the above denominator, we get:
\begin{align*}
& 1-\frac{1}{w} \sum_{j=0}^{w-1} f_{i+j}(\alpha) \leq
\frac{\dr-1}{w^2}\sum_{j, k=0}^{w-1} x_{i+j-k}(\alpha) \\
& \leq \frac{\dr-1}{w^2}\sum_{j, k=0}^{w-1} (x_{i-p+j-k}^*)^{\alpha} (x_{i-p+1+j-k}^*)^{1-\alpha}.
\end{align*}
Combining all these bounds and canceling common terms yields
\begin{align}\label{equ:lowerbound}
& \Big(\frac{\epsilon_i(\alpha)}{\epsilon^*}\Big)^{\frac1{\dl-1}} \geq \Big(1 - \frac1w\sum_{k=0}^{w-1}x^*_{i-p+w-k}\Big)^{(\dr-2)} \nonumber \\
& \frac{(\sum_{j, k=0}^{w-1}x^*_{i-p+j-k})^{\alpha}(\sum_{j, k=0}^{w-1}x^*_{i-p+1+j-k})^{1-\alpha}}{\sum_{j, k=0}^{w-1}(x^*_{i-p+j-k})^{\alpha}(x^*_{i-p+1+j-k})^{1-\alpha}}. 
\end{align}

Applying  Holder's inequality\footnote{For any two $n-$length real
sequences $(a_0,a_1,\dots,a_{n-1})$ and $(b_0,b_1,\dots,b_{n-1})$ and
two real numbers $p,q \in (1,\infty)$ such that $\frac1p+\frac1q=1$, 
Holder's inequality asserts that
\begin{align*} \sum_{k=0}^{n-1} |a_k b_k| \leq \Big(\sum_{k=0}^{n-1}
|a_k|^p\Big)^{\frac1p}\Big(\sum_{k=0}^{n-1} |b_k|^q\Big)^{\frac1q}.
\end{align*}}
we get
\begin{align*}
& \frac{(\sum_{j, k=0}^{w-1}x^*_{i-p+j-k})^{\alpha}(
\sum_{j, k=0}^{w-1}x^*_{i-p+1+j-k})^{1-\alpha}}{\sum_{j=0}^{w-1}\sum_{k=0}^{w-1}
(x^*_{i-p+j-k})^{\alpha}(x^*_{i-p+1+j-k})^{1-\alpha}} & \geq  1.
\end{align*}
Putting everything together we now get
\begin{align}\label{equ:maininequality}
 \Big(\frac{\epsilon_i(\alpha)}{\epsilon^*}\Big)^{\frac1{\dl-1}} \geq \Big(1 - \frac1w\sum_{k=0}^{w-1}x^*_{i-p+w-k}\Big)^{\dr-2}.  
\end{align}

By assumption $x_i(\alpha) \leq \gamma$. Again from monotonicity we have $x_i(\alpha)\geq x_i(1) = x^*_{i-p} $.
 Thus $x^*_{i-p} \leq \gamma$. Combining this with 
Lemma~\ref{lem:avgprop} (iii) and \eqref{equ:gamma} in the hypothesis of the theorem, we obtain
\begin{align*}
& \frac{(\dr-1)(\dl-1)(1+w^{1/8})}{
w}  \geq  \frac{1}{w^2}\sum_{j, k=0}^{w-1}x^*_{i-p+j-k}.
\end{align*} 
Suppose that $x^*_{i-p+w-\lceil w^{7/8}\rceil}>\frac{1}{w^{1/8}}$.
Then from the above inequality we conclude that
\begin{align*}
& \frac{(\dr-1)(\dl-1)(1+w^{1/8})}{
w}  \geq  \frac{1}{w^{2}w^{1/8}}(1+2+\dots+w^{\frac78}),
\end{align*}
where we set to zero all the terms smaller than $x^*_{i-p+w-\lceil w^{7/8}\rceil}$.
Upper bounding $(1+w^{1/8})$ by $2w^{1/8}$ we get
\begin{align*}
& 4(\dr-1)(\dl-1)
  \geq  w^{1/2}.
\end{align*}
But this is contrary to the hypothesis of the theorem,
$w>2^4 \dl^2 \dr^2$.  Hence we must have $x^*_{i-p+w-\lceil
w^{7/8}\rceil}\leq\frac{1}{w^{1/8}}$.  Therefore,
\begin{align*}
\frac1w\sum_{k=0}^{w-1}x^*_{i-p+w-k} \leq \frac{1}w\Big( \frac{w-\lceil w^{7/8}\rceil}{w^{1/8}} + \lceil w^{7/8}\rceil + 1\Big),
\end{align*}
where we replace $x^*_{i-p+1},\dots, x^*_{i-p+w-\lceil
w^{7/8}\rceil}$ by
$\frac1{w^{1/8}}$ and the remaining $\lceil w^{7/8}\rceil+1$ values by $1$.
Thus we have
$$
\frac1w\sum_{k=0}^{w-1}x^*_{i-p+w-k} \leq \frac4{w^{1/8}}.
$$
Using $w\geq 2^{16}$ and combining everything, we get
\begin{align*}
 \Big(\frac{\epsilon_i(\alpha)}{\epsilon^*}\Big)^{\frac1{\dl-1}} \geq \Big(1 - \frac4{w^{1/8}}\Big)^{\dr-2}.  
\end{align*}

\item[(v)]
{\em Area under EXIT Curve:}\footnote{A slightly more involved proof
shows that the area under the EXIT curve (or more precisely, the value of the EXIT integral) is {\em equal} to the design rate,
assuming that the design rate is defined in an appropriate way (see the
discussion on page \pageref{dis:designrate}).  For our purpose it is
sufficient, however, to determine the area up to bounds of order $w/L$.
This simplifies the expressions and the proof.}
Consider the set of $M$ variable nodes at position $i$, $i \in [-L,
L]$. We want to compute their associated EXIT integral, i.e., we want
to compute $\int_{0}^{1} h_i(\alpha) d\epsilon_i(\alpha)$. We use the
technique introduced in \cite{MMU08}.

We consider the set of $M$ computation trees of height $2$ rooted
in all variable nodes at position $i$, $i \in [-L, L]$. For each such
computation tree there are $\dl$ check nodes and $1+\dl(\dr-1)$ variable
nodes.  Each of the leaf variable nodes of each computation
tree has a certain position in the range $[i-w+1,
i+w-1]$. These positions differ for each computation tree.  
For each computation tree assign to its root node the channel value 
$\epsilon_i(\alpha)$, whereas each leaf variable node at position $k$
``sees'' the channel value $x_k(\alpha)$.

In order to compute $\int_{0}^{1} h_i(\alpha) d\epsilon_i(\alpha)$ we proceed
as follows.  We apply the standard area theorem \cite[Theorem 3.81]{RiU08} to
the $M$ simple codes represented by these $M$ computation trees.  Each such
code has length $1+ \dl(\dr-1)$ and $\dl$ (linearly independent) check nodes.
As we will discuss shortly, the standard area theorem tells us the value of the
sum of the $1+ \dl(\dr-1)$ individual EXIT integrals associated to a particular
code. This sum consists of the EXIT integral of the root node as well as the
$\dl (\dr-1)$ EXIT integrals of the leaf nodes.  Assume that we can determine
the contributions of the EXIT integrals of the leaf nodes for each computation
tree. In this case we can subtract the average such contribution from the sum
and determine the average EXIT integral associated to the root node.  In the
 ensuing argument, we consider a fixed instance of a computation tree rooted in
$i$. We then average over the randomness of the ensemble. For the root node the
channel value stays the same for all instances, namely, $\epsilon_i(\alpha)$ as
given in Definition~\ref{def:EXIT} of the EXIT curve. Hence, for the root
node the average, over the ensemble, is taken only over the EXIT value.  Then,
exchanging the integral (w.r.t. $\alpha$) and the average  and using the fact
that each edge associated to the root node behaves independently, we conclude
that the average EXIT integral associated to the root node is equal to
$\int_{0}^{1} h_i(\alpha) d\epsilon_i(\alpha)$, the desired quantity.  Let us
now discuss this program in more detail.

For $i \in [-L+w-1, L-w+1]$
we claim that the average sum of the EXIT integrals associated to any such
computation tree is equal to $1+\dl(\dr-2)$.  This is true since for $i$
in this range, the positions of all leaf nodes are in the range $[-L, L]$.
Now applying the area theorem\footnote{To be precise, the proof of the area theorem 
given in \cite[Theorem 3.81]{RiU08} assumes that
the channel value of the root node, call it $\epsilon_i(\alpha)$,
stays within the range $[0,1]$. This does not apply in our setting; 
for $\alpha\to 0$, $\epsilon_i(\alpha)$ becomes unbounded. Nevertheless, it
is not hard to show, by explicitly writing down the sum
of all EXIT integrals, using integration by parts and finally using the fact that $(\x(\alpha),\underline{\epsilon}(\alpha))$ is a FP, that the result still
applies in this more general setting.}
one can conclude that the average sum of all the $1+\dl(\dr-1)$ EXIT
 integrals associated to the tree code equals the number of variable nodes minus
the number of check nodes: $1+\dl(\dr-1) - \dl=1+\dl(\dr-2)$.

For $i \in [-L, -L+w-2] \cup [L-w+2, L]$ the situation is more
complicated.  It can happen that some of the leaf nodes of the computation
tree see a perfect channel for all values $\alpha$ since their position
is outside $[-L, L]$. These leaf nodes are effectively not present in the
code and we should remove them before counting.  Although it would not
be too difficult to determine the exact average contribution for such a
root variable node we only need bounds --  the average sum of the EXIT integrals
associated to such a root node is at least $0$ and at most $1+\dl(\dr-2)$.

We summarize: If we consider all computation trees rooted in all variable
nodes in the range $[-L, L]$ and apply the standard area theorem to
each such tree, then the total average contribution is at least $M (2 L -2 w
+3)(1+\dl(\dr-2))$ and at most $M (2 L+1)(1+\dl(\dr-2))$.  From these
bounds we now have to subtract the contribution of all the leaf nodes
of all the computation trees and divide by $M$ in order to determine
bounds on $\sum_{i=-L}^{L} \int_0^1 h_i(\alpha) d\epsilon_i(\alpha)$.

Consider the expected contribution of the $\dl(\dr-1)$ EXIT integrals
of each of the $M$ computation trees rooted at $i$, $i \in [-L+w-1,
L-w+1]$. We claim that this contribution is equal to $M \dl(\dr-1)^2/\dr$.
For computation trees rooted in $i \in [-L, -L+w-2] \cup [L-w+2, L]$,
on the other hand, this contribution is at least $0$ and at most $M \dl
(\dr-1)$.

Let us start with computation trees rooted in $i$, $i \in [-L+w-1,
L-w+1]$.  Fix $i$. It suffices to consider in detail one ``branch''
of a computation tree since the EXIT integral is an expected value
and expectation is linear. By assumption the root node is at position
$i$. It is connected to a check node, let's say at position $j$, $j \in
[i, i+w-1]$, where the choice is made uniformly at random.  In turn,
this check node has $(\dr-1)$ children. Let the positions of these
children be $k_1, \dots, k_{\dr-1}$, where all these indices are in
the range $[k-w+1, k]$, and all choices are independent and are made
uniformly at random.

Consider now this check node in more detail and apply the standard
area theorem to the corresponding parity-check code of length $\dr$.
The message from the root node is $x_i(\alpha)$, whereas the messages from
the leaf nodes are $x_{k_l}(\alpha)$, $l=1, \dots, \dr-1$, respectively.
We know from the standard area theorem applied to this parity-check
code of length $\dr$ that the sum of the $\dr$ EXIT integrals is equal
to $\dr-1$. So the average contribution of one such EXIT integral is
$(\dr-1)/\dr$, and the average of $(\dr-1)$ randomly chosen such EXIT
integrals is $(\dr-1)^2/\dr$. Recalling that so far we only considered
$1$ out of $\dl$ branches and that there are $M$ computation trees,
the total average contribution of all leaf nodes of all computation
trees rooted in $i$  should therefore be $M \dl(\dr-1)^2/\dr$.

Let us now justify why the contribution of the leaf nodes is equal to the
``average'' contribution. Label the $\dr$ edges of the check node from
$1$ to $\dr$, where ``$1$'' labels the root node.  Further, fix $j$,
the position of the check node.  As we have seen, we get the associated
channels $(i, k_1, \dots, k_{\dr-1})$ if we root the tree in position
$i$, connect to check node $j$, and then connect further to $k_1, \dots,
k_{\dr-1}$. This particular realization of this branch happens with
probability $w^{-\dr}$ (given that we start in $i$) and the expected number
of branches starting in $i$ that have exactly the same ``type'' $(i,
k_1, \dots, k_{\dr-1})$ equals $M \dl w^{-\dr}$.  Consider a permutation of $(i,
k_1, \dots, k_{\dr-1})$ and keep $j$ fixed.  To be concrete, let's say we
consider the permutation $(k_3, i, k_2, \dots, k_1)$.  This situation
occurs if we root the tree in $k_3$, connect to check node $j$, and
then connect further to $i, k_2, \dots, k_1$. Again, this happens with
probability $w^{-\dr}$ and the expected number of such branches is $M \dl
w^{-\dr}$. It is crucial to observe that all permutations of $(i, k_1,
\dots, k_{\dr-1})$ occur with equal probability in these computation
trees and that all the involved integrals occur for computation graphs
that are rooted in a position in the range $[-L, L]$.  Therefore, the
``average'' contribution of the $(\dr-1)$ leaf nodes is just a fraction
$(\dr-1)/\dr$ of the total contribution, as claimed. Here, we have used
a particular notion of ``average.''  We have averaged not only over
various computation trees rooted at position $i$ but also over computation
trees rooted let's say in position $k_{l}$, $l=1, \dots \dr-1$. Indeed, we
have averaged over an equivalence class given by all permutations of $(i,
k_1, \dots, k_{\dr-1})$, with $j$, the position of the check node held
fixed. Since $i \in [-L+w-1, L-w+1]$, all these quantities are also
in the range $[-L, L]$, and so they are included in our consideration.

It remains to justify the ``average'' contributions that we get for
computation trees rooted in $i \in [-L, -L+w-2] \cup [L-w+2, L]$.
The notion of average is the same as we have used it above.
Even though we are talking about averages, for each computation
tree it is clear that the contribution is non-negative since all the
involved channel values $x_k(\alpha)$ are increasing functions in $\alpha$.
This proves that the average contribution is non-negative. Further,
the total uncertainty that we remove by each variable leaf node is at
most $1$.  This proves the upper bound.

We can now summarize. We have
\begin{align*}
\frac{\sum_{i=-L}^{L}
\int_{0}^{1} h_i(\alpha) d\epsilon_i(\alpha)}{2 L+1} & \leq 1-\frac{\dl}{\dr}+\frac{2(w\!-\!1)}{2 L+1} 
\frac{\dl(\dr-1)^2}{\dr}, \\
& \leq 1\!-\!\frac{\dl}{\dr} +\frac{w}{L} \dl \dr,\\
\frac{\sum_{i=-L}^{L}
\int_{0}^{1} h_i(\alpha) d\epsilon_i(\alpha)}{2 L+1} & \geq 1\!-\!\frac{\dl}{\dr}
\!-\!\frac{2(w\text{-}1)}{2 L\text{+}1}(1\text{+}\dl (\dr\text{-}1)\text{-}\frac{\dl}{\dr}) \\
&\geq 1\!-\!\frac{\dl}{\dr} -\frac{w}{L} \dl \dr.
\end{align*}

\item[(vi)] {\em Bound on $\epsilon^*$:} \\
Consider the EXIT function constructed according to
Definition~\ref{def:EXIT}.  Recall that the EXIT value at position $i \in [-L, L]$
is defined by
\begin{align}\label{equ:hgrelation}
h_i(\alpha) = (g(x_{i-w+1}(\alpha), \dots, x_{i+w-1}(\alpha)))^{\frac{\dl}{\dl-1}},
\end{align}
and the area under the EXIT curve is given by
\begin{align}\label{equ:integralnew}
A(\dl, \dr, w, L)=\int_{0}^{1}\frac1{2L+1}\sum_{i=-L}^{L} h_i(\alpha)d\epsilon_i(\alpha).
\end{align}  
As we have just seen this integral is close to the
design rate $R(\dl, \dr, \w, L)$, and from Lemma~\ref{lem:designrate}
we know that this design rate converges to $1-\dl/\dr$ for any fixed
$w$ when $L$ tends to infinity.

The basic idea of the proof is the following.  We will show that
$A(\dl, \dr, w, L)$ is also ``close'' to
$1-\frac{\dl}{\dr}+p^{\MAPsmall}(x(\epsilon^*))$, where
$p^{\MAPsmall}(\cdot)$ is the polynomial defined in
Lemma~\ref{lem:standardthresholds}. In other words, $x(\epsilon^*)$
must be ``almost'' a zero of $p^{\MAPsmall}(\cdot)$. But
$p^{\MAPsmall}(\cdot)$ has only a single positive root and this
root is at $\epsilon^{\MAPsmall}(\dl, \dr)$.

More precisely, we first find upper and lower bounds on $A(\dl,
\dr, w, L)$ by splitting the integral (\ref{equ:integralnew}) into 
four phases.  We will see that the main contribution to the area
comes from the first phase and that this contribution is close to
$1-\frac{\dl}{\dr} - p^{\MAPsmall}(x(\epsilon^*))$. For all other
phases we will show that the contribution can be bounded by a
function which does not depend on $(\epsilon^*, \x^*)$ and which
tends to $0$ if let $w$ and $L$ tend to infinity.

For $i=\{1, 2, 3, 4\}$, define $T_i$ as
\begin{align*}
T_i & = \int_{\frac{4-i}{4}}^{\frac{5-i}{4}}\frac2{2L+1}\sum_{i=-L+w-1}^{-w+1} h_i(\alpha)d\epsilon_i(\alpha).
\end{align*}
Further, let
\begin{align*}
T_5 & = \int_{0}^{1}\frac2{2L+1}\sum_{i=-L}^{-L+w-2} h_i(\alpha)d\epsilon_i(\alpha), \\
T_6 & = \int_{0}^{1} \frac1{2L+1}\sum_{i=-w+2}^{w-2} h_i(\alpha)d\epsilon_i(\alpha).
\end{align*}
Clearly, $A(\dl, \dr, w, L) = T_1+T_2+T_3+T_4+T_5+T_6$. 
We claim that for $w > \max\{2^4 \dl^2 \dr^2, 2^{16}\}$,
\begin{align*}
& T_1  = 1 \!-\! \frac{\dl}{\dr} \!-\! p^{\MAPsmall}(x_0^*), \\
- \dl \dr (x^*_0-x^*_{-L}) \leq & T_2 \leq \dr (x^*_0-x^*_{-L}), \\
-w^{\text{-}\frac18} \!- \! \frac{2\dr\dl^2}{w^{\frac78}(1 \!-\! 4w^{-\frac{1}8})^{\dr}\epsilon^{\BPsmall}(\dl,\dr)} 
\leq & T_3 \leq 4 \dl \dr w^{-\frac18}, \\
-\dl \dr x^*_{-L'+L} \leq & T_4 \leq \dr x^*_{-L'+L}, \\
-\frac{\dl w}{L} \leq & T_5 \leq \frac{w}{L}, \\
-\frac{\dl w}{L} \leq & T_6 \leq \frac{w}{L}. 
\end{align*}
If we assume these bounds for a moment, and simplify the expressions
slightly, we see that for $w > \max\{2^{16}, 2^4 \dl^2 \dr^2\}$,
\begin{align*}
 &|A(\dl, \dr, w, L) \!-\!1\!+\! \frac{\dl}{\dr}\!+\! p^{\MAPsmall}(x^*_0))|  \!\leq\! 4\dl\dr w^{-\frac18} \!+\! \frac{2w\dl}{L}  \\ 
& \!+\! \dl \dr (x^*_{-L'\!+\!L}\!+\!x^*_0\!-\!x^*_{-L}) \!+\!  \frac{2\dr\dl^2}{(1 \!-\! 4w^{-\frac{1}8})^{\dr}\epsilon^{\BPsmall}(\dl, \dr)}w^{-\frac78} .
\end{align*}
Now using the bound in part (v) on the area under the EXIT curve we get
\begin{align*}
\vert p^{\MAPsmall}(x^*_0) \vert & \leq c_1(\dl, \dr, w, L),
\end{align*}
where
\begin{align*}
&c_1(\dl, \dr, w, L) =  
4\dl \dr w^{-\frac18} + \frac{2w\dl}{L} + \frac{w\dl \dr}{L}\\
& \!+\! \dl \dr (x^*_{-L'\!+\!L}\!+\!x^*_0\!-\!x^*_{-L}) \!+\!  \frac{2\dr\dl^2}{(1 \!-\! 4w^{-\frac{1}8})^{\dr}\epsilon^{\BPsmall}(\dl, \dr)} w^{-\frac78}. 
\end{align*}
From this we can derive a bound on $\epsilon^*$ as follows.  
Using Taylor's expansion we get
$$
p^{\MAPsmall}(x^*_0) = p^{\MAPsmall}(\xstab(\epsilon^*)) + (x^*_0 - \xstab(\epsilon^*)) (p^{\MAPsmall}(\eta))',
$$
where $(p^{\MAPsmall}(x))'$ denotes the derivative w.r.t. $x$ and $\eta \in (x^*_0, \xstab(\epsilon^*))$.
From Lemma~\ref{lem:standardthresholds} one can verify that  $\vert (p^{\MAPsmall}(x))'\vert \leq 2\dl\dr$ for all $x\in [0,1]$. 
Thus, 
$$
\vert p^{\MAPsmall}(\xstab(\epsilon^*)) \vert  \leq 2\dl \dr \vert x^*_0 - \xstab(\epsilon^*)\vert + c_1(\dl, \dr, w, L). 
$$
Now using $p^{\MAPsmall}(\xstab(\epsilon^{\MAPsmall}))=0$ and the fundamental theorem of calculus we have
\begin{align*}
 p^{\MAPsmall}(\xstab(\epsilon^*)) = 
 -\int^{\xstab(\epsilon^{\MAPsmall})}_{\xstab(\epsilon^*)}(p^{\MAPsmall}(x))' dx.
\end{align*}
Further, for a $(\dl, \dr)$-regular ensemble we have 
$$
(p^{\MAPsmall}(x))' = (1-(1-x)^{\dr-1})^{\dl}\epsilon'(x),
$$
where we recall that $\epsilon(x)=x/(1-(1-x)^{\dr-1})^{\dl-1}$.
 Next, from
Lemma~\ref{lem:maximum} we have that $\epsilon^*>\epsilon^{\BPsmall}$. Thus
$\xstab(\epsilon^*)>x^{\BPsmall}$. Also, $\epsilon^{\MAPsmall}> \epsilon^{\BPsmall}$. As a consequence,  
$(1-(1-x)^{\dr-1})^{\dl} \geq (1-(1-x^{\BPsmall})^{\dr-1})^{\dl}$ and
 $\epsilon'(x)\geq 0$ 
 for all $x$ in the interval of the  above integral. 

Combining everything we get
\begin{align*}
\vert  p^{\MAPsmall}&(\xstab(\epsilon^*)) \vert  \geq (1-(1-x^{\BPsmall})^{\dr-1})^{\dl}\Big\vert \int^{\xstab(\epsilon^{\MAPsmall})}_{\xstab(\epsilon^*)} \epsilon'(x) dx \Big\vert  \\
& = (1-(1-x^{\BPsmall})^{\dr-1})^{\dl} \vert \epsilon(\xstab(\epsilon^{\MAPsmall})) - \epsilon(\xstab(\epsilon^*)) \vert.
\end{align*}

Define
\begin{align*}
c(\dl,& \dr, w, L) =  
4\dl \dr w^{-\frac18} + \frac{2w\dl}{L} + \frac{w\dl \dr}{L}\\
& + \dl \dr (x^*_{-L'+L}+x^*_0-x^*_{-L})  + \frac{2\dr\dl^2}{(1 \!-\! 4w^{-\frac{1}8})^{\dr}} w^{-\frac78}. 
\end{align*}
Then, using $\epsilon(\xstab(\epsilon^*))=\epsilon^*$ and $\epsilon(\xstab(\epsilon^{\MAPsmall}))=\epsilon^{\MAPsmall}(\dl, \dr)$, the final result is  
\begin{align*}
&\vert \epsilon^{\MAPsmall}(\dl, \dr) - \epsilon^* \vert  
 \leq \frac{2\dl \dr \vert x^*_0 - \xstab(\epsilon^*)\vert + c(\dl, \dr, w, L)}{\epsilon^{\BPsmall}(\dl,\dr)(1-(1-x^{\BPsmall})^{\dr-1})^{\dl}} \\
& \stackrel{(a)}{=} \frac{2\dl \dr \vert x^*_0 - \xstab(\epsilon^*)\vert + c(\dl, \dr, w, L)}{x^{\BPsmall}(\dl,\dr)(1-(1-x^{\BPsmall})^{\dr-1})} \\
& \stackrel{(b)}{\leq} \frac{2\dl \dr \vert x^*_0 - \xstab(\epsilon^*)\vert + c(\dl, \dr, w, L)}{(x^{\BPsmall}(\dl,\dr))^2}\\
& \stackrel{\text{Lemma}~\ref{lem:lowerboundxBP}}{\leq} \frac{2\dl\dr\vert x_0^* -
\xstab(\epsilon^*)\vert+c(\dl, \dr, w,
L)}{(1-(\dl-1)^{-\frac1{\dr-2}})^2}.
\end{align*}
To obtain  $(a)$ we use that $x^{\BPsmall}$ is a FP of standard DE for channel parameter $\epsilon^{\BPsmall}$. 
Also, we use $(1-(1-x^{\BPsmall})^{\dr-1})\geq x^{\BPsmall}(\dl, \dr)$ to get $(b)$. 

It remains to verify the bounds on the six integrals.  Our strategy
is the following. For $i \in [-L+w-1, -w+1]$ we evaluate the
integrals directly in phases (i), (ii), and (iii), using the general
bounds on the quantities $\epsilon_i(\alpha)$.  For the boundary
points, i.e., for $i \in [-L, -L+w-2]$ and $i\in [-w+2, 0]$, as well
as for all the positions in phase (iv), we use the following crude but
handy bounds, valid for $0 \leq \alpha_1 \leq \alpha_2 \leq 1$:
\begin{align}
& \int_{\alpha_1}^{\alpha_2} h_i(\alpha) d\epsilon_i(\alpha) \leq
h_i(\alpha_2)\epsilon_i(\alpha_2)-h_i(\alpha_1)\epsilon_i(\alpha_1) \nonumber \\
& \leq x_i(\alpha_2) (g(x_{i-w+1}(\alpha_2), \dots, x_{i+w-1}(\alpha_2))^{\frac{1}{\dl-1}} \nonumber \\
& \leq x_i(\alpha_2) \leq 1, \label{equ:intupperbound} \\
& \int_{\alpha_1}^{\alpha_2} h_i(\alpha) d\epsilon_i(\alpha) \geq - \int_{\alpha_1}^{\alpha_2} \epsilon_i(\alpha) dh_i(\alpha) \nonumber \\
&\geq \!-\! \dl \big\{(h_i(\alpha_2))^{\frac{1}{\dl}} \!-\!(h_i(\alpha_1))^{\frac{1}{\dl}} \big\}\geq \!-\! \dl (h_i(\alpha_2))^{\frac{1}{\dl}} \geq -\dl. \label{equ:intlowerbound}
\end{align}
To prove (\ref{equ:intupperbound}) use integration by parts to write
\begin{align*}
\int_{\alpha_1}^{\alpha_2} \!\!\!h_i(\alpha) d\epsilon_i(\alpha) & = 
\int_{\alpha_1}^{\alpha_2} \!\!\!d(h_i(\alpha)\epsilon_i(\alpha)) - 
\int_{\alpha_1}^{\alpha_2} \!\!\!\epsilon_i(\alpha) dh_i(\alpha). 
\end{align*}
Now note that $\epsilon_i(\alpha) \geq 0$ and that $h_i(\alpha)$ is an
increasing function in $\alpha$ by construction. The second term on the
right hand side of the above equality is therefore negative and we get
an upper bound if we drop it. We get the further bounds by inserting the
explicit expressions for $h_i$ and $\epsilon_i$ and by noting that $x_i$ as well as $g$
are upper bounded by $1$.

To prove (\ref{equ:intlowerbound}) we also use integration by parts, 
but now we drop the first term. Since $h_i(\alpha)$ is an increasing
function in $\alpha$ and it is continuous, it is invertible. We can therefore write the
integral in the form $\int_{h_i(\alpha_1)}^{h_i(\alpha_2)} \epsilon_i(h)
dh$. Now note that $\epsilon_i(h) h = x_i(h) g^{\frac{1}{\dl-1}}(h)=
x_i(h) h^{\frac{1}{\dl}} \leq h^{\frac{1}{\dl}}$, where we used the
fact that $h=g^{\frac{\dl}{\dl-1}}$ (recall the definition of $g(...)$ from \eqref{equ:hgrelation}). This shows that $\epsilon_i(h)
\leq h^{\frac{1-\dl}{\dl}}$. We conclude that
\begin{align*}
& \int_{\alpha_1}^{\alpha_2} \!\!\epsilon_i(\alpha) dh_i(\alpha) \leq 
\int_{h_i(\alpha_1)}^{h_i(\alpha_2)} \!\! h^{\frac{1-\dl}{\dl}} dh \\
& =\dl \big\{ h_i(\alpha_2)^{\frac{1}{\dl}} -  h_i(\alpha_1)^{\frac{1}{\dl}} \big\}  \leq \dl h_i(\alpha_2)^{\frac{1}{\dl}} \leq \dl.
\end{align*}
The bounds on $T_4$, $T_5$ and $T_6$ are straightforward applications
of (\ref{equ:intlowerbound}) and (\ref{equ:intupperbound}).  E.g.,
to prove that $T_6 \leq \frac{w}{L}$, note that there are $2 w-3$
positions that are involved. For each position we know from
(\ref{equ:intupperbound}) that the integral is upper bounded by
$1$. The claim now follows since $\frac{2 w-3}{2L-1} \leq \frac{w}{L}$.
Using \eqref{equ:intlowerbound} leads to the lower bound. Exactly the same line of reasoning leads to 
 both the bounds for $T_5$.

For the upper bound on $T_4$ we use the second inequality in
\eqref{equ:intupperbound}.  We then bound $x_i(\alpha)\leq 1$ and
use $h_i(...)^{\frac1{\dl}} = g(...)^{\frac1{\dl-1}}$, cf. 
(\ref{equ:hgrelation}).  Next, we bound each term in the sum by the
maximum term. This maximum is $h_0(\frac14)^{\frac{1}{\dl}}$.  This
term can further be upper bounded by $1-(1-x^*_{-L'+L})^{\dr-1} \leq \dr
x^*_{-L'+L}$. Indeed, replace all the $x$ values in $h_0(\frac14)$ by their maximum, $x^*_{-L'+L}$. 
 The lower bound follows in a similar way using the penultimate inequality in \eqref{equ:intlowerbound}.

Let us continue with $T_1$.  Note that for $\alpha\in
[3/4,1]$ and $i\in [-L+w-1, -w+1]$, $\epsilon_i(\alpha)=
\frac{x_i(\alpha)}{(1-(1-x_i(\alpha))^{\dr-1})^{\dl-1}}$ and that
$h_i(\alpha)=(1-(1-x_i(\alpha))^{\dr-1})^{\dl}$. A direct calculation
shows that
\begin{align*}
T_1 = \int_{\frac34}^{1} h_i(\alpha) d\epsilon_i(\alpha) 
& = p^{\MAPsmall}(1)-p^{\MAPsmall}(x_i(3/4)) \\
& = 1 - \frac{\dl}{\dr} - p^{\MAPsmall}(x_0(3/4)) \\
& = 1 - \frac{\dl}{\dr} - p^{\MAPsmall}(x_0^*).
\end{align*}

Let us now compute bounds on $T_2$.  
Using (\ref{equ:intupperbound}) we get
\begin{align*}
T_2 & \leq \frac{2}{2 L+1} \sum_{i=-L+w-1}^{-w+1}  
(h_i(3/4)\epsilon_i(3/4)-h_i(1/2)\epsilon_i(1/2)) \\
& \leq \{x^*_0(1-(1-x^*_0)^{\dr-1})
- x^*_{-L}(1-(1-x^*_{-L})^{\dr-1})\} \\
& \leq \dr(x^*_{0}-x^*_{-L}).
\end{align*} 
To obtain the second inequality we use $\epsilon_i(\alpha)h_i(\alpha) = x_i(\alpha)(h_i(\alpha))^{\frac1{\dl}}$.
Using the second inequality of \eqref{equ:intlowerbound} we lower bound $T_2$ as follows. We have
\begin{align*}
T_2 & \geq -\frac{2\dl}{2 L+1} \sum_{i=-L+w-1}^{-w+1}  
(h_i(3/4)^{\frac1{\dl}} - h_i(1/2))^{\frac1{\dl}}) \\
& \geq -\dl \{(1-x^*_{-L})^{\dr-1} - (1-x^*_{0})^{\dr-1}\} \\
& \geq -\dl\dr(x^*_{0} - x^*_{-L}).
\end{align*}
To obtain the second inequality we use $h_i(3/4) 
 =  (1 - (1- x^*_0)^{\dr-1})^{\dl}$ and $h_i(1/2) \geq (1 - (1-x^*_{-L})^{\dr-1})^{\dl}$.

It remains to bound $T_3$. 
For $i \in [-L+w-1, -w+1]$, consider
\begin{align}\label{equ:totalderivative}
\int_{\frac12}^{\frac14} d(h_i(\alpha)  \epsilon_i(\alpha)) = \epsilon^* 
(h_i(\frac12)-h_i(\frac14)),
\end{align}
where we have made use of the fact that for $\alpha=\frac14$ and $\alpha=\frac12$, 
$\epsilon_i(\alpha)=\epsilon^*$.
To get an upper bound on $T_3$ write 
\begin{align*}
\int_{\frac14}^{\frac12} \epsilon_i(\alpha))dh_i(\alpha)\!\geq\! 
\epsilon^*\bigl(1\!-\!\frac4{w^{\frac18}}\bigr)^{(\dr\!-\!2)(\dl\!-\!1)} 
(h_i(\frac12)\!-\!h_i(\frac14)).
\end{align*}
Here we have used the lower bounds on $\epsilon_i(\alpha)$ in phase
(iii) from Theorem~\ref{the:propertiesEXIT} and the fact that $w >\max\{2^{16}, 2^4 \dl^2 \dr^2\}$.  Again using
integration by parts, and upper bounding both $\epsilon^*$ and
$(h_i(1/2)-h_i(1/4))$ by $1$, we conclude that
\begin{align*}
& \int_{\frac12}^{\frac14} h_i(\alpha) d \epsilon_i(\alpha) 
\leq 1-\Bigl(1 - \frac4{w^{1/8}}\Bigr)^{(\dr-2)(\dl-1)} \\
& \leq 4 \dr \dl w^{-1/8}.
\end{align*}
Note that the right-hand-side is independent of $i$ so that this
bound extends directly to the sum, i.e.,
\begin{align*}
T_3 & \leq 4 \dr \dl w^{-1/8}.
\end{align*}

For the lower bound we can proceed in a similar fashion.  

We first apply integration by parts. Again using \eqref{equ:totalderivative}, the first term corresponding to the total derivative can be written as
\begin{align*}
\frac2{2L+1}\sum_{-L+w-1}^{-w+1}\epsilon^*(h_i(\frac12) - h_i(\frac14)).
\end{align*}
We write the other term in the integration by parts as follows. 
For every section number $i\in [-L+w-1, -w+1]$, let $\beta_i$ correspond to the
smallest  number in $[\frac14, \frac12]$ such that $x_i(\beta_i) > \gamma$.  
Recall
the definition of $\gamma$ from part (iv) of Theorem~\ref{the:propertiesEXIT}. 
If for any section number $i$, $x_i(\frac12)> \gamma$, then $\beta_i$ is
well-defined and $x_i(\alpha) > \gamma$ for
 all $\alpha \in [\beta_i, \frac12]$. Indeed, this follows from the continuity and the
monotonicity of $x_i(\alpha)$ w.r.t. $\alpha$. On the
other hand, if $x_i(\frac12) \leq \gamma$, we set $\beta_i=\frac12$. 
Then we can write the second term as 
\begin{align*}
\frac{-2}{2L+1}  \sum_{-L+w-1}^{-w+1} \Big(\int_{\frac14}^{\beta_i} \epsilon_i(\alpha)dh_i(\alpha) + \int_{\beta_i}^{\frac12} \epsilon_i(\alpha)dh_i(\alpha) \Big).
\end{align*}
We now lower bound the two integrals as follows. 
For $\alpha \in [\beta_i, \frac12]$ we use the upper bound on $\epsilon_i(\alpha)$ valid in phase
(iii) from Theorem~\ref{the:propertiesEXIT}. This gives us the lower bound 
\begin{align*}
\frac{-2}{2L+1}  \sum_{-L+w-1}^{-w+1} \epsilon^*\Big(1 + \frac1{w^{1/8}}\Big) (h_i(\frac12) - h_i(\frac14)),
\end{align*}
where above we used the fact that $h_i(\beta_i) \geq h_i(\frac14)$. 
 
 For $\alpha \in [\frac14, \beta_i]$ we use the universal 
bound $-\dl h_i(\beta_i)^{\frac{1}{\dl}}$ (on $\int_{\frac14}^{\beta_i} \epsilon_i(\alpha)dh_i(\alpha)$) stated in 
(\ref{equ:intlowerbound}). 
Since $1/4 \leq \beta_i \leq 1/2$, using the lower bound on $\epsilon_i(\beta_i) \geq \epsilon^*(1 - 4w^{-1/8})^{(\dr-2)(\dl-1)}$
 (in phase (iii) of Theorem~\ref{the:propertiesEXIT}), we get 
\begin{align*}
-\dl h_i(\beta_i)^{\frac{1}{\dl}} &= -\dl
\Big(\frac{x_i(\beta_i)}{\epsilon_i(\beta_i)}\Big)^{\frac1{\dl-1}}  \\ 
& \geq   
 -\dl \Big(\frac{\gamma^{\frac1{\dl-1}} }{\epsilon^{\BPsmall}(\dl, \dr) (1 \!-\! 4w^{-\frac{1}8})^{\dr}}\Big).
\end{align*}
Above we use $\epsilon^*\geq \epsilon^{\BPsmall}(\dl, \dr)$, replace $(\dr-2)$ by $\dr$ and $(\epsilon^{\BPsmall}(\dl, \dr))^{\frac1{\dl-1}}$ by $\epsilon^{\BPsmall}(\dl, \dr)$. 
Putting everything together,
\begin{align*} 
T_3 & \geq 
1-\Bigl(1 + \frac1{w^{1/8}}\Bigr) 
  -\dl \Big(\frac{\gamma^{\frac1{\dl-1}} }{\epsilon^{\BPsmall} (1 \!-\! 4w^{-\frac{1}8})^{\dr}}\Big), \\
& = -w^{-\frac18} - \dl \Big(\frac{\gamma^{\frac1{\dl-1}} }{\epsilon^{\BPsmall} (1 \!-\! 4w^{-\frac{1}8})^{\dr}}\Big).
\end{align*}
Since $\gamma^{\frac1{\dl-1}}\leq \frac{2\dr\dl}{w^{\frac78}} $, the final result is 
\begin{align*} 
T_3 & = -w^{-\frac18}-  \frac{2\dr\dl^2}{w^{\frac78}(1 \!-\! 4w^{-\frac{1}8})^{\dr}\epsilon^{\BPsmall}(\dl,\dr)} .
\end{align*}

\end{itemize}


\end{appendices}

\bibliographystyle{IEEEtran} 
\bibliography{lth,lthpub}
\end{document}

%% file: definition.tex
\definecolor{TODO}{rgb}{0.6,0.6,0.6} 

\definecolor{TOCHECK}{rgb}{0.8,0.8,0.8} 

\newtheorem{theorem}{Theorem}
\newcommand{\btheo}{\begin{theorem}}
\newcommand{\etheo}{\end{theorem}}
\newcommand{\bproof}{\begin{proof}}
\newcommand{\eproof}{\end{proof}}
\newtheorem{definition}[theorem]{Definition}
\newcommand{\bdefi}{\begin{definition}}
\newcommand{\edefi}{\end{definition}}
\newtheorem{fact}[theorem]{Fact}
\newcommand{\bprop}{\begin{fact}}
\newcommand{\eprop}{\end{fact}}
\newtheorem{corollary}[theorem]{Corollary}
\newcommand{\bcor}{\begin{corollary}}
\newcommand{\ecor}{\end{corollary}}
\newtheorem{example}[theorem]{Example}
\newcommand{\bex}{\begin{example}}
\newcommand{\eex}{\end{example}}
\newtheorem{lemma}[theorem]{Lemma}
\newcommand{\blemma}{\begin{lemma}}
\newcommand{\elemma}{\end{lemma}}
\newtheorem{remark}[theorem]{Remark}
\newcommand{\bremark}{\begin{remark}}
\newcommand{\eremark}{\end{remark}}
\newtheorem{conj}[theorem]{Conjecture}
\newcommand{\bconj}{\begin{conj}}
\newcommand{\econj}{\end{conj}}

\newcommand{\naturals}{\ensuremath{\mathbb{N}}}
\newcommand{\integers}{\ensuremath{\mathbb{Z}}}

\newcommand{\prob}{\ensuremath{\mathbb{P}}}

\def\0{{\tt 0}} 
\def\1{{\tt 1}} 
\def\?{{\tt *}} 
\DeclareMathOperator{\w}{w}    
\newcommand{\BP}{\ensuremath{\text{BP}}} 
\newcommand{\MAPsmall}{\ensuremath{\text{\tiny MAP}}} 
\newcommand{\BPsmall}{\ensuremath{\text{\tiny BP}}} 
\newcommand{\EBPsmall}{\ensuremath{\text{\tiny EBP}}} 
\newcommand{\qed}{{\hfill \footnotesize $\blacksquare$}}
\renewcommand{\mid}{\,|\,}

\newcommand{\dr}{{\mathtt r}}
\newcommand{\dl}{{\mathtt l}}
\newcommand{\dlh}{\hat{\mathtt l}}

\newcommand{\exitf}{h}

%% file: ps/36protograph.tex
\setlength{\unitlength}{1.0bp}%
\begin{picture}(210,86)(0,0)
\put(0,0)
{
\put(0,0){\rotatebox{0}{\includegraphics[scale=1.0]{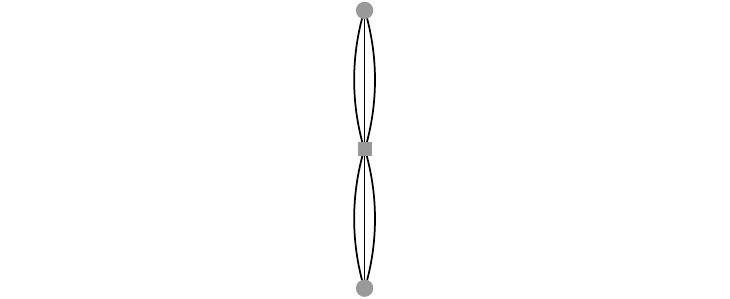}}}
}
\end{picture}

%% file: ps/36protographchain.tex
\setlength{\unitlength}{1.0bp}%
\begin{picture}(210,86)(0,0)
\put(0,0)
{
\put(0,0){\rotatebox{0}{\includegraphics[scale=1.0]{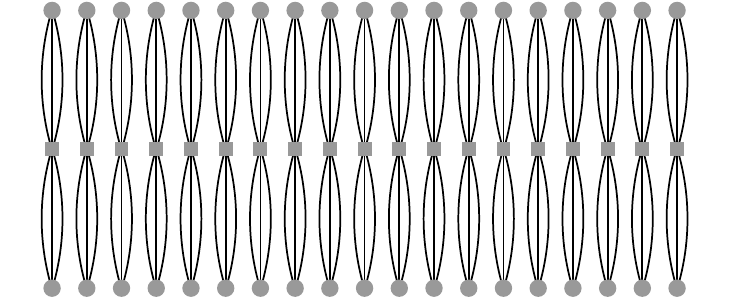}}}
{
\footnotesize
\put(15,-7){\makebox(0,0)[b]{$\text{-}L$}}
\put(105,-7){\makebox(0,0)[b]{$0$}}
\put(195,-7){\makebox(0,0)[b]{$L$}}
}
}
\end{picture}

%% file: ps/chain.tex
\setlength{\unitlength}{1bp}%
\setlength{\unitlength}{1bp}%
\begin{picture}(210,200)(-5,-10)
\put(0,0)
{
\put(-5,0){\includegraphics[scale=1.0]{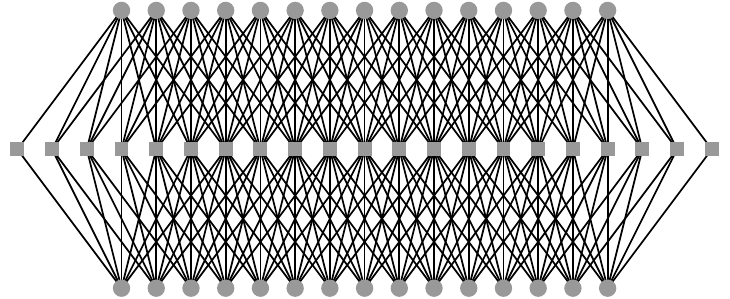}}
\put(30,-10){\makebox(0,0)[b]{$\text{-}L$}}
\put(45,-10){\makebox(0,0)[b]{$\cdots$}}
\multiputlist(60,-10)(10,0)[b]{$\text{-}4$,$\text{-}3$,$\text{-}2$,$\text{-}1$,$0$,$1$,$2$, $3$, $4$}
\put(155,-10){\makebox(0,0)[b]{$\cdots$}}
\put(170,-10){\makebox(0,0)[b]{$L$}}
}
\put(0,110)
{
\put(-5,0){\includegraphics[scale=1.0]{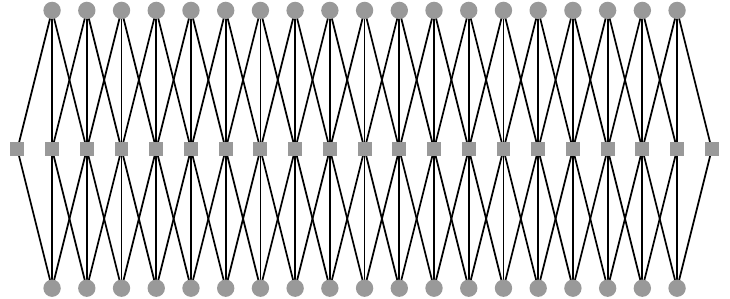}}
\put(10,-10){\makebox(0,0)[b]{$\text{-}L$}}
\put(30,-10){\makebox(0,0)[lb]{$\cdots$}}
\multiputlist(60,-10)(10,0)[b]{$\text{-}4$,$\text{-}3$,$\text{-}2$,$\text{-}1$,$0$,$1$,$2$, $3$, $4$}
\put(150,-10){\makebox(0,0)[lb]{$\cdots$}}
\put(190,-10){\makebox(0,0)[b]{$L$}}
}
\end{picture}

%% file: ps/ebpexit36.tex
\setlength{\unitlength}{1.0bp}%
\begin{picture}(274,108)(-14,-8)
\put(0,0)
{
\put(0,0){\includegraphics[scale=1.0]{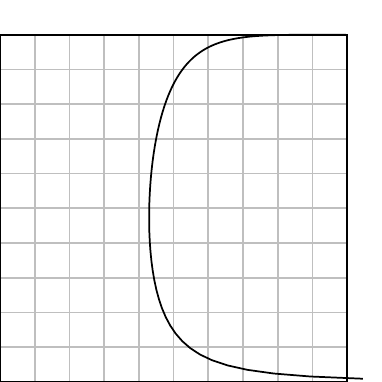}}
\small
\multiputlist(0,-8)(20,0)[cb]{$~$,$0.2$,$0.4$,$0.6$,$0.8$}
\multiputlist(-14,0)(0,20)[l]{$~$,$0.2$,$0.4$,$0.6$}
\put(-14,-8){\makebox(0,0)[lb]{$0.0$}}
\put(100,-8){\makebox(0,0)[rb]{$\epsilon$}}
\put(-14,100){\makebox(0,0)[lt]{\rotatebox{90}{$h^{\text{\scriptsize EPB}}$}}}
\put(102,2){\makebox(0,0)[lb]{\rotatebox{90}{$(1, \xunstab(1))$}}}
}
\put(135,0)
{
	\put(0,0){\includegraphics[scale=1.0]{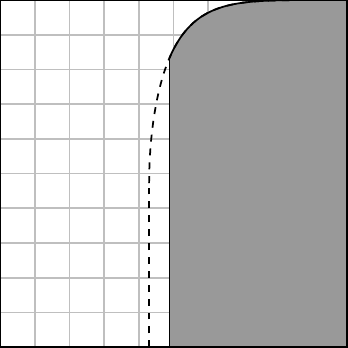}}
	\small
	\multiputlist(0,-8)(20,0)[cb]{$~$,$0.2$,$0.4$,$0.6$,$0.8$}
	\multiputlist(-14,0)(0,20)[l]{$~$,$0.2$,$0.4$,$0.6$}
	\put(-14,-8){\makebox(0,0)[lb]{$0.0$}}
	\put(100,-8){\makebox(0,0)[rb]{$\epsilon$}}
	\put(-14,100){\makebox(0,0)[lt]{\rotatebox{90}{$h {} (\epsilon)$}}}
	\put(41,2){\makebox(0,0)[rb]{\rotatebox{90}{$\epsilon^{\BPsmall}$}}}
	\put(51,2){\makebox(0,0)[lb]{\rotatebox{90}{$\epsilon^{\MAPsmall}$}}}
	\put(75, 60){\makebox(0,0)[t]{$\int h^{\BPsmall}{}=\frac12$}}
}
\end{picture}

%% file: ps/lrLexit.tex
\setlength{\unitlength}{1.0bp}%
\begin{picture}(120,120)(0,0)
\put(0,0)
{
\put(-10,120){\makebox(0,0)[lt]{\rotatebox{90}{$h^{\text{\scriptsize EPB}}$}}}
\put(120,-8){\makebox(0,0)[rb]{$\epsilon$}}
\put(0,0){\rotatebox{0}{\includegraphics[scale=1.0]{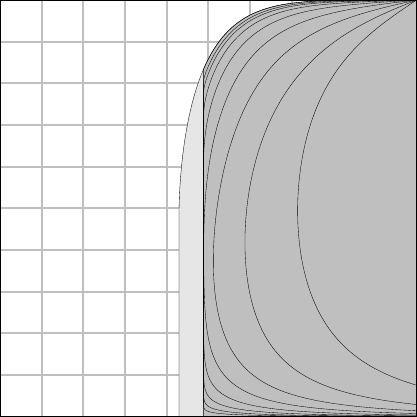}}}
\put(48,5){\makebox(0,0)[rb]{\rotatebox{90}{$\epsilon^{\BPsmall}(\text{\footnotesize3, 6}) \approx \text{\footnotesize 0.4294}$}}}
\put(60,5){\makebox(0,0)[rb]{\rotatebox{90}{$\epsilon^{\MAPsmall}{(\text{\footnotesize3, 6})} \approx \text{\footnotesize 0.4881}$}}}
\put(88,57){\makebox(0,0)[l]{\rotatebox{0}{\footnotesize $L\!=\!1$}}}
\put(73,68){\makebox(0,0)[l]{\rotatebox{0}{\footnotesize $L\!=\!2$}}}
}
\end{picture}

%% file: ps/lrLwiggle.tex
\setlength{\unitlength}{1.0bp}%
\begin{picture}(120,120)(0,0)
\put(0,0)
{
\put(0,0){\rotatebox{0}{\includegraphics[scale=1.0]{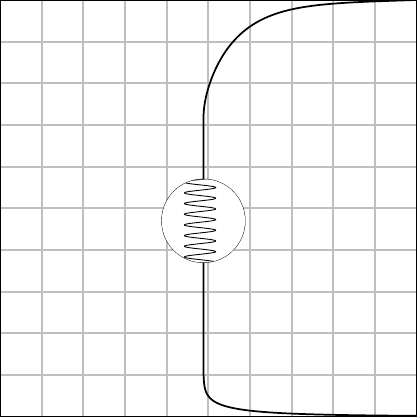}}}
\put(-10,120){\makebox(0,0)[lt]{\rotatebox{90}{$h^{\text{\scriptsize EPB}}$}}}
\put(120,-8){\makebox(0,0)[rb]{$\epsilon$}}
}
\end{picture}

%% file: ps/lrLwwiggle.tex
\setlength{\unitlength}{1.0bp}%
\begin{picture}(250,120)(0,0)
\put(0,0)
{
\put(0,0){\rotatebox{0}{\includegraphics[scale=1.0]{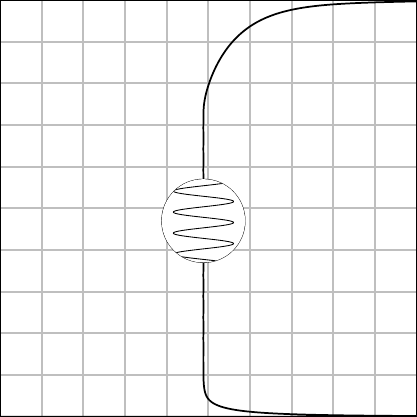}}}
\put(2,118){\makebox(0,0)[lt]{\rotatebox{90}{$h^{\text{\scriptsize EPB}}$}}}
\put(120,-8){\makebox(0,0)[rb]{$\epsilon$}}
\put(130,0){\rotatebox{0}{\includegraphics[scale=1.0]{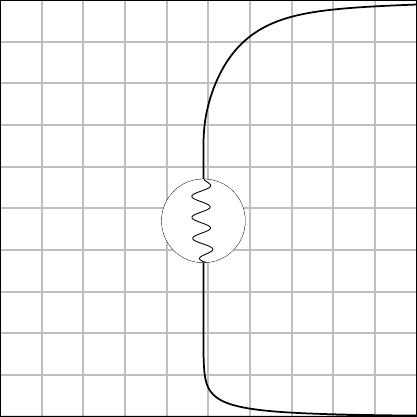}}}
\put(132,118){\makebox(0,0)[lt]{\rotatebox{90}{$h^{\text{\scriptsize EPB}}$}}}
\put(250,-8){\makebox(0,0)[rb]{$\epsilon$}}
}
\end{picture}

%% file: ps/accordeonfp.tex
\setlength{\unitlength}{1.0bp}%
\begin{picture}(240,65)(0,0)
\put(0,0)
{
\put(0,0){\rotatebox{0}{\includegraphics[scale=1.0]{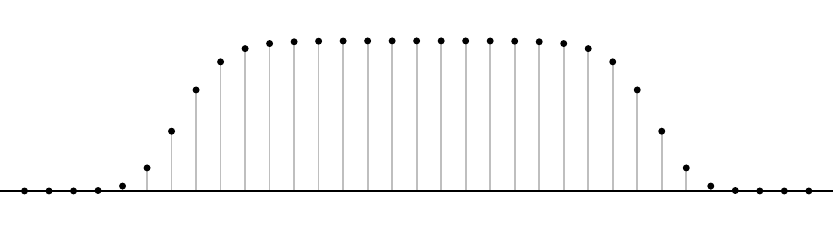}}}
\footnotesize
\multiputlist(6,0)(14.2,0)[b]{$\text{-}16$,$\text{-}14$,$\text{-}12$,$\text{-}10$,$\text{-}8$,$\text{-}6$,$\text{-}4$,$\text{-}2$,0,2,4,6,8,10,12,14,16}
}
\end{picture}

%% file: ps/one-sided_fixed_point.tex
\setlength{\unitlength}{1.0bp}%
\begin{picture}(240,65)(0,0)
\put(0,0)
{
\put(0,0){\rotatebox{0}{\includegraphics[scale=1.0]{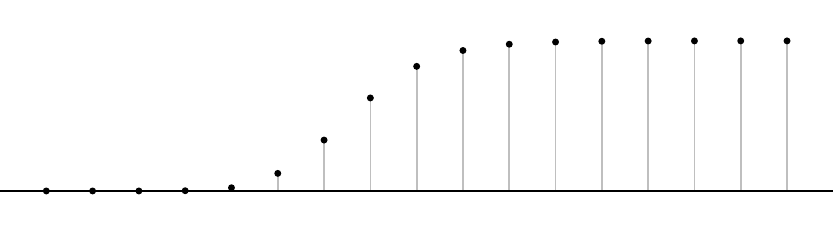}}}
\footnotesize
\multiputlist(8,0)(27.5,0)[b]{$\text{-}16$,$\text{-}14$,$\text{-}12$,$\text{-}10$,$\text{-}8$,$\text{-}6$,$\text{-}4$,$\text{-}2$,0}
}
\end{picture}

%% file: ps/interpolation.tex
\setlength{\unitlength}{0.5bp}%
\begin{picture}(400,644)
\put(0,0)
{
\put(0,440)
{
\put(0,0){\rotatebox{0}{\rotatebox{0}{\includegraphics[scale=0.5]{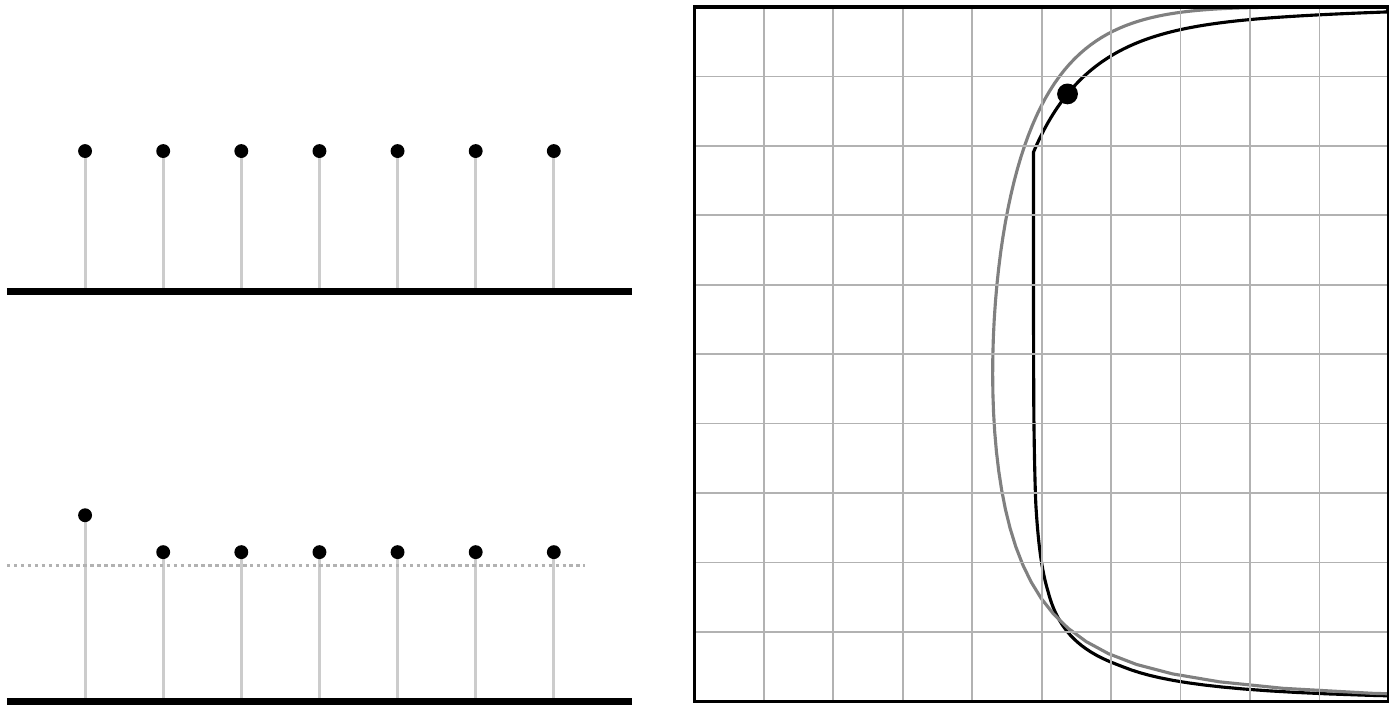}}}}
\put(197,200){\makebox(0,0)[rt]{\rotatebox{90}{$h^{\text{\scriptsize EPB}}$}}}
\put(400,-2){\makebox(0,0)[tr]{$\epsilon$}}
\footnotesize
\put(180,42){\makebox(0,0)[c]{$\epsilon^*$}}
\put(100,175){\makebox(0,0)[c]{$\x$}}
\put(100,60){\makebox(0,0)[c]{$\underline{\epsilon}$}}
\multiputlist(24,115)(22.75,0)[t]{$\text{-}6$,$\text{-}5$,$\text{-}4$,$\text{-}3$,$\text{-}2$,$\text{-}1$,$0$,}
\multiputlist(24,-3)(22.75,0)[t]{$\text{-}6$,$\text{-}5$,$\text{-}4$,$\text{-}3$,$\text{-}2$,$\text{-}1$,$0$,}
}
\put(0, 220)
{
\put(0,0){\rotatebox{0}{\rotatebox{0}{\includegraphics[scale=0.5]{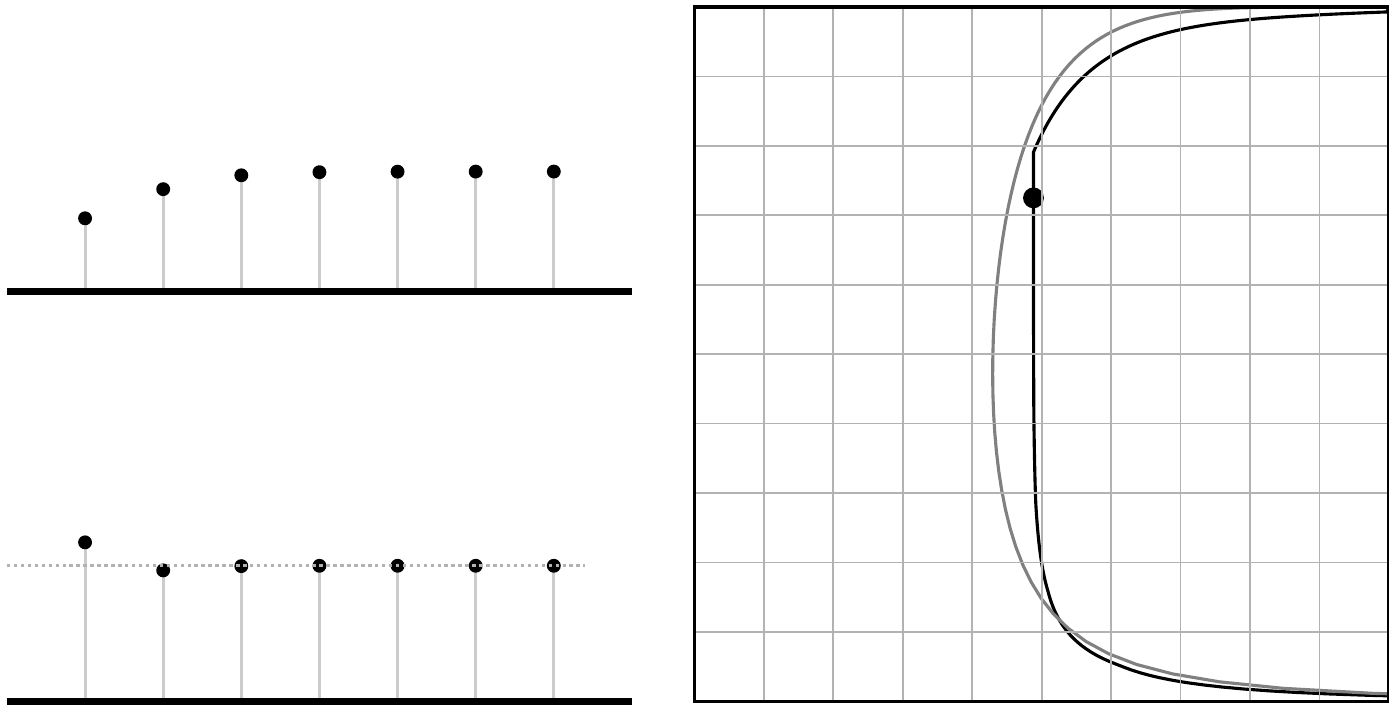}}}}
\put(197,200){\makebox(0,0)[rt]{\rotatebox{90}{$h^{\text{\scriptsize EPB}}$}}}
\put(400,-2){\makebox(0,0)[tr]{$\epsilon$}}
\footnotesize
\put(100,170){\makebox(0,0)[c]{$\x$}}
\put(180,42){\makebox(0,0)[c]{$\epsilon^*$}}
\put(100,55){\makebox(0,0)[c]{$\underline{\epsilon}$}}
\multiputlist(24,115)(22.75,0)[t]{$\text{-}6$,$\text{-}5$,$\text{-}4$,$\text{-}3$,$\text{-}2$,$\text{-}1$,$0$,}
\multiputlist(24,-3)(22.75,0)[t]{$\text{-}6$,$\text{-}5$,$\text{-}4$,$\text{-}3$,$\text{-}2$,$\text{-}1$,$0$,}
}
\put(0,0)
{
\put(0,0){\rotatebox{0}{\rotatebox{0}{\includegraphics[scale=0.5]{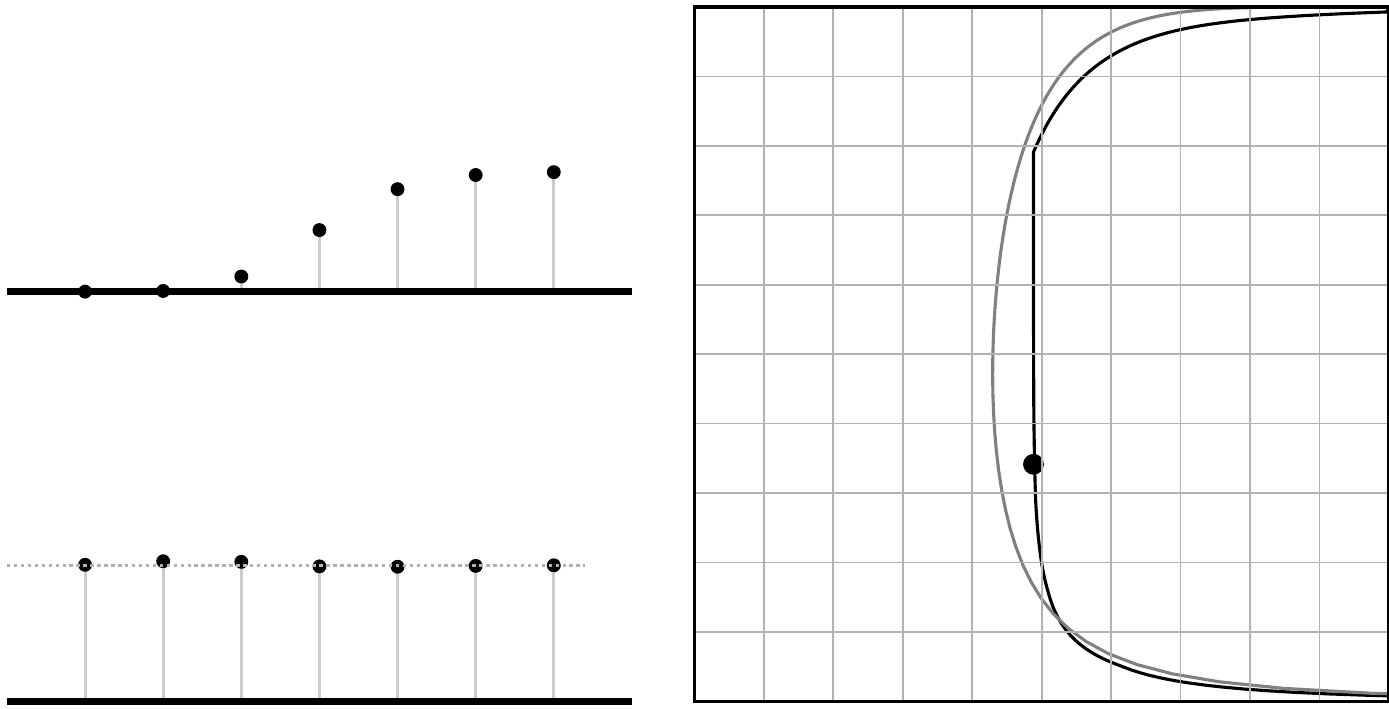}}}}
\put(197,200){\makebox(0,0)[rt]{\rotatebox{90}{$h^{\text{\scriptsize EPB}}$}}}
\put(400,-2){\makebox(0,0)[tr]{$\epsilon$}}
\footnotesize
\put(100,160){\makebox(0,0)[c]{$\x$}}
\put(180,42){\makebox(0,0)[c]{$\epsilon^*$}}
\put(100,60){\makebox(0,0)[c]{$\underline{\epsilon}$}}
\multiputlist(24,115)(22.75,0)[t]{$\text{-}6$,$\text{-}5$,$\text{-}4$,$\text{-}3$,$\text{-}2$,$\text{-}1$,$0$,}
\multiputlist(24,-3)(22.75,0)[t]{$\text{-}6$,$\text{-}5$,$\text{-}4$,$\text{-}3$,$\text{-}2$,$\text{-}1$,$0$,}
}
}
\end{picture}

%% file: ps/de36.tex
\setlength{\unitlength}{1.0bp}%
\begin{picture}(200,140)(0,-10)
\put(0,0)
{
\put(0,0){\includegraphics[scale=2.0]{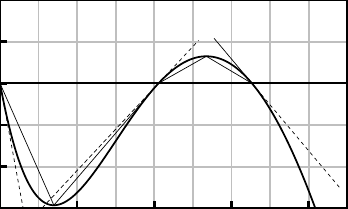}}
\small
\multiputlist(0,-8)(45,0)[cb]{$~$,$0.1$,$0.2$,$0.3$,$0.4$}
\multiputlist(-3,-24)(0,24)[r]{$~$,$\text{-}0.03$,$\text{-}0.02$,$\text{-}0.01$, $0.00$,$0.01$}
\put(0,-8){\makebox(0,0)[lb]{$0.0$}}
\put(145,12){\makebox(0,0)[lb]{\rotatebox{90}{$\xstab=0.3265$}}}
\put(85,8){\makebox(0,0)[lb]{\rotatebox{90}{$\xunstab=0.2054$}}}
\put(50,40){\makebox(0,0)[lb]{\rotatebox{48}{$h'(\xunstab)=0.1984$}}}
\put(130,40){\makebox(0,0)[lb]{\rotatebox{-50}{$h'(\xstab)=-0.2202$}}}
\put(105,78){\vector(0,-1){25}}
\put(102,0){\makebox(0,0)[lb]{\rotatebox{90}{${\tiny\text{slope}}=0.1048$}}}
\put(137,76){\vector(0,-1){20}}
\put(135,0){\makebox(0,0)[lb]{\rotatebox{90}{${\tiny\text{slope}}=-0.1098$}}}
\put(45,20){\vector(0,1){20}}
\put(40,44){\makebox(0,0)[lb]{\rotatebox{90}{${\tiny\text{slope}}=0.2157$}}}
\put(25,20){\vector(0,1){20}}
\put(20,44){\makebox(0,0)[lb]{\rotatebox{90}{${\tiny\text{slope}}=-0.4191$}}}
\put(0,120){\makebox(0,0)[lb]{$\kappa_*(0.44)$}}
\put(60,120){\makebox(0,0)[b]{$\lambda_*(0.44)$}}
\put(100,120){\makebox(0,0)[b]{$\kappa^*(0.44)$}}
\put(140,120){\makebox(0,0)[b]{$\lambda^*(0.44)$}}
}
\end{picture}

%% file: cap.bbl
\newcommand{\SortNoop}[1]{}
\begin{thebibliography}{10}
\providecommand{\url}[1]{#1}
\csname url@rmstyle\endcsname
\providecommand{\newblock}{\relax}
\providecommand{\bibinfo}[2]{#2}
\providecommand\BIBentrySTDinterwordspacing{\spaceskip=0pt\relax}
\providecommand\BIBentryALTinterwordstretchfactor{4}
\providecommand\BIBentryALTinterwordspacing{\spaceskip=\fontdimen2\font plus
\BIBentryALTinterwordstretchfactor\fontdimen3\font minus
  \fontdimen4\font\relax}
\providecommand\BIBforeignlanguage[2]{{%
\expandafter\ifx\csname l@#1\endcsname\relax
\typeout{** WARNING: IEEEtran.bst: No hyphenation pattern has been}%
\typeout{** loaded for the language `#1'. Using the pattern for}%
\typeout{** the default language instead.}%
\else
\language=\csname l@#1\endcsname
\fi
#2}}

\bibitem{MMRU04}
C.~M{\'e}asson, A.~Montanari, T.~Richardson, and R.~Urbanke, ``Life above
  threshold: From list decoding to area theorem and {MSE},'' in \emph{Proc. of
  the IEEE Inform. Theory Workshop}, San Antonio, TX, USA, Oct. 2004, e-print:
  cs.IT/0410028.

\bibitem{MMU08}
C.~M{\'e}asson, A.~Montanari, and R.~Urbanke, ``Maxwell construction: The
  hidden bridge between iterative and maximum a posteriori decoding,''
  \emph{IEEE Trans. Inform. Theory}, vol.~54, no.~12, pp. 5277--5307, 2008.

\bibitem{Tan81b}
R.~M. Tanner, ``Error-correcting coding system,'' Oct. 1981, {U.S.} Patent $\#$
  4,295,218.

\bibitem{Tan87}
------, ``Convolutional codes from quasi-cyclic codes: a link between the
  theories of block and convolutional codes,'' University of California, Santa
  Cruz, Tech Report UCSC-CRL-87-21, Nov. 1987.

\bibitem{FeZ99}
A.~J. Felstr{\"{o}}m and K.~S. Zigangirov, ``Time-varying periodic
  convolutional codes with low-density parity-check matrix,'' \emph{IEEE Trans.
  Inform. Theory}, vol.~45, no.~5, pp. 2181--2190, Sept. 1999.

\bibitem{EnZ99}
K.~Engdahl and K.~S. Zigangirov, ``On the theory of low density convolutional
  codes {I},'' \emph{Problemy Peredachi Informatsii}, vol.~35, no.~4, pp.
  295--310, 1999.

\bibitem{ELZ99}
K.~Engdahl, M.~Lentmaier, and K.~S. Zigangirov, ``On the theory of low-density
  convolutional codes,'' in \emph{AAECC-13: Proceedings of the 13th
  International Symposium on Applied Algebra, Algebraic Algorithms and
  Error-Correcting Codes}.\hskip 1em plus 0.5em minus 0.4em\relax London, UK:
  Springer-Verlag, 1999, pp. 77--86.

\bibitem{LTZ01}
M.~Lentmaier, D.~V. Truhachev, and K.~S. Zigangirov, ``To the theory of
  low-density convolutional codes. {II},'' \emph{Probl. Inf. Transm.}, vol.~37,
  no.~4, pp. 288--306, 2001.

\bibitem{TSSFC04}
R.~M. Tanner, D.~Sridhara, A.~Sridharan, T.~E. Fuja, and D.~J. Costello, Jr.,
  ``{LDPC} block and convolutional codes based on circulant matrices,''
  \emph{IEEE Trans. Inform. Theory}, vol.~50, no.~12, pp. 2966 -- 2984, Dec.
  2004.

\bibitem{SLCZ04}
A.~Sridharan, M.~Lentmaier, D.~J. Costello, Jr., and K.~S. Zigangirov,
  ``Convergence analysis of a class of {LDPC} convolutional codes for the
  erasure channel,'' in \emph{Proc. of the Allerton Conf. on Commun., Control,
  and Computing}, Monticello, IL, USA, Oct. 2004.

\bibitem{LSZC10}
M.~Lentmaier, A.~Sridharan, K.~S. Zigangirov, and D.~J. Costello, Jr.,
  ``Iterative decoding threshold analysis for {LDPC} convolutional codes,''
  \emph{IEEE Trans. Info. Theory}, Oct. 2010.

\bibitem{LSZC05}
------, ``Terminated {LDPC} convolutional codes with thresholds close to
  capacity,'' in \emph{Proc. of the IEEE Int. Symposium on Inform. Theory},
  Adelaide, Australia, Sept. 2005.

\bibitem{RiU08}
T.~Richardson and R.~Urbanke, \emph{Modern Coding Theory}.\hskip 1em plus 0.5em
  minus 0.4em\relax Cambridge University Press, 2008.

\bibitem{LeF10}
M.~Lentmaier and G.~P. Fettweis, ``On the thresholds of generalized {LDPC}
  convolutional codes based on protographs,'' in \emph{Proc. of the IEEE Int.
  Symposium on Inform. Theory}, Austin, USA, 2010.

\bibitem{MPZC08}
D.~G.~M. Mitchell, A.~E. Pusane, K.~S. Zigangirov, and D.~J. Costello, Jr.,
  ``Asymptotically good {LDPC} convolutional codes based on protographs,'' in
  \emph{Proc. of the IEEE Int. Symposium on Inform. Theory}, Toronto, CA, July
  2008, pp. 1030 -- 1034.

\bibitem{LFZC09}
M.~Lentmaier, G.~P. Fettweis, K.~S. Zigangirov, and D.~J. Costello, Jr.,
  ``Approaching capacity with asymptotically regular {LDPC} codes,'' in
  \emph{Information Theory and Applications}, San Diego, USA, Feb. 8--Feb. 13,
  2009, pp. 173--177.

\bibitem{SPVC06}
R.~Smarandache, A.~Pusane, P.~Vontobel, and D.~J. Costello, Jr.,
  ``Pseudo-codewords in {LDPC} convolutional codes,'' in \emph{Proc. of the
  IEEE Int. Symposium on Inform. Theory}, Seattle, WA, USA, July 2006, pp. 1364
  -- 1368.

\bibitem{SPVC09}
------, ``Pseudocodeword performance analysis for {LDPC} convolutional codes,''
  \emph{IEEE Trans. Inform. Theory}, vol.~55, no.~6, pp. 2577--2598, June 2009.

\bibitem{PISWC10}
M.~Papaleo, A.~Iyengar, P.~Siegel, J.~Wolf, and G.~Corazza, ``Windowed erasure
  decoding of {LDPC} convolutional codes,'' in \emph{Proc. of the IEEE Inform.
  Theory Workshop}, Cairo, Egypt, Jan. 2010, pp. 78 -- 82.

\bibitem{LMSSS97}
M.~Luby, M.~Mitzenmacher, A.~Shokrollahi, D.~A. Spielman, and V.~Stemann,
  ``Practical loss-resilient codes,'' in \emph{Proc. of the 29th annual ACM
  Symposium on Theory of Computing}, 1997, pp. 150--159.

\bibitem{T03}
J.~Thorpe, ``Low-density parity-check ({LDPC}) codes constructed from
  protographs,'' Aug. 2003, {J}et {P}ropulsion {L}aboratory, {INP} {P}rogress
  {R}eport 42-154.

\bibitem{DDJ06}
D.~Divsalar, S.~Dolinar, and C.~Jones, ``Constructions of {P}rotograph {LDPC}
  codes with linear minimum distance,'' in \emph{Proc. of the IEEE Int.
  Symposium on Inform. Theory}, Seattle, WA, USA, July 2006.

\bibitem{MLC10}
D.~G.~M. Mitchell, M.~Lentmaier, and D.~J. Costello, Jr., ``New families of
  {LDPC} block codes formed by terminating irregular protograph-based {LDPC}
  convolutional codes,'' in \emph{Proc. of the IEEE Int. Symposium on Inform.
  Theory}, Austin, USA, June 2010.

\bibitem{MMRU09}
C.~M{\'e}asson, A.~Montanari, T.~Richardson, and R.~Urbanke, ``The generalized
  area theorem and some of its consequences,'' \emph{IEEE Trans. Inform.
  Theory}, vol.~55, no.~11, pp. 4793--4821, Nov. 2009.

\bibitem{RiU04b}
T.~Richardson and R.~Urbanke, ``Multi-edge type {LDPC} codes,'' 2002, presented
  at the {W}orkshop honoring {P}rof. {B}ob {M}c{E}liece on his 60th birthday,
  Caltech, USA.

\bibitem{LMFC10}
M.~Lentmaier, D.~G.~M. Mitchell, G.~P. Fettweis, and D.~J. Costello, Jr.,
  ``Asymptotically good {LDPC} convolutional codes with {AWGN} channel
  thresholds close to the {S}hannon limit,'' Sept. 2010, 6th International
  Symposium on Turbo Codes and Iterative Information Processing.

\bibitem{KMRU10}
S.~Kudekar, C.~M{\'e}asson, T.~Richardson, and R.~Urbanke, ``Threshold
  saturation on {BMS} channels via spatial coupling,'' Sept. 2010, 6th
  International Symposium on Turbo Codes and Iterative Information Processing.

\end{thebibliography}
